\documentclass[11pt]{llncs}

\makeatletter
\renewcommand*\l@author[2]{}
\renewcommand*\l@title[2]{}
\makeatletter

\usepackage{wrapfig}

\usepackage{fullpage}

\usepackage{url}
\usepackage[latin1]{inputenc}
\usepackage[T1]{fontenc}
\usepackage{latexsym}
\usepackage{amsmath}
\usepackage{amssymb}
\usepackage{xspace}
\usepackage{ifthen}
\usepackage{todonotes}
\usepackage[colorlinks=false]{hyperref} 
\usepackage{authblk}

\usepackage{tikz}
\usetikzlibrary{arrows,decorations.pathreplacing}
\newcommand{\cmd}[1]{\operatorname{#1}}
\newcommand{\varid}{vid}

\usepackage{theorem}

\usepackage{stmaryrd}

\newcommand{\defterm}[1]{\emph{#1}}

\newcommand{\Tpre}{T_{\operatorname{pre}}}
\newcommand{\Tonl}{T_{\operatorname{onl}}}
\newcommand{\Ttot}{T_{\operatorname{tot}}}

\newboolean{LongVersion}
\setboolean{LongVersion}{true}

\newcommand{\shortOnly}[1]{\ifthenelse{\not \boolean{LongVersion}}{#1}{}}
\newcommand{\longOnly}[1]{\ifthenelse{\boolean{LongVersion}}{#1}{}}

\newcommand{\Section}[1]{\shortOnly{\vspace{-1ex}}\section{#1}}

\newcommand{\Subsection}[1]{\shortOnly{\vspace{-1ex}}\subsection{#1}}

\newcommand{\Subsubsection}[1]{\shortOnly{\vspace{-2ex}}\subsubsection{#1}}

\newcommand{\Steal}{\shortOnly{\vspace{-2ex}}}

\newcommand{\repa}[1]{[ #1 ]_{\ant{A}}}
\newcommand{\repb}[1]{[ #1 ]_{\ant{B}}}
\newcommand{\repab}[2]{[ #1 | #2 ]}
\newcommand{\reps}[1]{[ #1 ]}

\newcommand{\putbox}{\hspace*{\fill} $\Box$}

\newcommand{\M}{{\mathcal{M}}}

\newcommand{\FMPC}{\ensuremath{\mathcal{F}_\textsc{2PC}}\xspace}
\newcommand{\DEAL}{\ensuremath{\mathcal{F}_\textsc{Deal}}\xspace}

\newcommand{\Ali}{{\ant{A}}\xspace}
\newcommand{\Aliadv}{$\Ali^*$\xspace}
\newcommand{\Bob}{{\ant{B}}\xspace}
\newcommand{\Bobadv}{$\Bob^*$\xspace}
\newcommand{\Sim}{{\ant{Sim}}\xspace}
\newcommand{\B}[1]{\textbf{#1}}

\newcommand{\LABIT}{\ensuremath{\operatorname{LaBit}}\xspace}
\newcommand{\WABIT}{\ensuremath{\operatorname{WaBit}}\xspace}
\newcommand{\ABIT}{\ensuremath{\operatorname{aBit}}\xspace}
\newcommand{\AAND}{\ensuremath{\operatorname{aAND}}\xspace}
\newcommand{\LAAND}{\ensuremath{\operatorname{LaAND}}\xspace}
\newcommand{\OT}{\ensuremath{\operatorname{OT}}\xspace}
\newcommand{\AOT}{\ensuremath{\operatorname{aOT}}\xspace}
\newcommand{\LAOT}{\ensuremath{\operatorname{LaOT}}\xspace}
\newcommand{\EQ}{\ensuremath{\operatorname{EQ}}\xspace}

\newcommand{\ABITs}{\ensuremath{\operatorname{aBit}}s\xspace}
\newcommand{\AANDs}{\ensuremath{\operatorname{aAND}}s\xspace}
\newcommand{\LAANDs}{\ensuremath{\operatorname{LaAND}}s\xspace}

\newcommand{\AOTs}{\ensuremath{\operatorname{aOT}}s\xspace}
\newcommand{\LAOTs}{\ensuremath{\operatorname{LaOT}}s\xspace}

\newcommand{\LL}{{\mathcal{L}}}

\renewcommand{\SS}{{\mathcal{S}}}

\newcommand{\col}{\operatorname{col}}

\newcommand{\lemmlab}[1]{\label{lemm:#1}}
\newcommand{\lemmref}[1]{Lemma~\ref{lemm:#1}}

\newcommand{\thmlab}[1]{\label{thm:#1}}
\newcommand{\thmref}[1]{Thm.~\ref{thm:#1}}

\newcommand{\seclab}[1]{\label{sec:#1}}
\newcommand{\secref}[1]{Sect.~\ref{sec:#1}}

\newcommand{\eqlab}[1]{\label{eq:#1}}
\renewcommand{\eqref}[1]{(\ref{eq:#1})}

\newcommand{\figlab}[1]{\label{fig:#1}}
\newcommand{\figref}[1]{Fig.~\ref{fig:#1}}

\newcommand{\prob}[1]{\operatorname{Pr}\left[{#1}\right]}
\newcommand{\E}[1]{\operatorname{E}\left[{#1}\right]}
\newcommand{\price}{\operatorname{price}}
\newcommand{\success}{\operatorname{success}}
\newcommand{\img}{\operatorname{img}}
\newcommand{\leak}{\operatorname{leak}}

\newcommand{\prg}{\operatorname{prg}}
\newcommand{\poly}{\operatorname{poly}}

\newcommand{\oh}{\smash{\frac12}}

\newcommand{\func}[1]{\textsc{#1}}

\newenvironment{boxfig}[2]{% {#1}{#2} = {Caption}{label}
     \begin{figure}[ht!]
     \newcommand{\FigCaption}{#1}
     \newcommand{\FigLabel}{#2}
       \vspace{-0.25cm}
     \begin{center}
       \begin{small}
         \begin{tabular}{@{}|@{~~}l@{~~}|@{}}
           \hline
           \rule[-1.5ex]{0pt}{1ex}\begin{minipage}[b]{.96\linewidth}
             \vspace{1ex}
             \smallskip
             }{%
           \end{minipage}\\
           \hline
         \end{tabular}
       \end{small}
       \vspace{-0.25cm}
       \caption{\FigCaption}
       \figlab{\FigLabel}
     \end{center}
       \longOnly{\vspace{-0.5cm}}
       \shortOnly{\vspace{-0.5cm}}
   \end{figure}
}

\newcommand{\ant}[1]{\ensuremath{\mathsf{#1}}}

\newcommand{\Def}{\stackrel{\text{\normalfont\tiny{def}}}{=}} %% Definition

\newcommand{\eqq}{\stackrel{\text{\normalfont\tiny{?}}}{=}} %% Definition

\newcommand{\inR}{\in_{\text{\normalfont\tiny{R}}}}

\newcommand{\applab}[1]{\label{app:#1}}
\newcommand{\appref}[1]{App.~\ref{app:#1}}

\newcommand{\zo}{\{0,1\}}

\title{A New Approach to Practical  \\ Active-Secure Two-Party Computation} 

\author{Jesper Buus Nielsen\inst{1}, Peter Sebastian Nordholt\inst{1},
  Claudio Orlandi\inst{2}, Sai Sheshank Burra\inst{3}}

\institute{Aarhus University \and Bar-Ilan University \and Indian
  Institute of Technology Guwahati}

\begin{document}

\longOnly{\maketitle}

\shortOnly{
\begin{center}
\begin{Large}
A New Approach to Practical Active-Secure Two-Party Computation
\end{Large}
\end{center}
}

\begin{abstract}
  We propose a new approach to practical two-party computation secure
  against an active adversary. All prior practical protocols were
  based on Yao's garbled circuits. We use an OT-based approach and get
  efficiency via OT extension in the random oracle model. To get a
  practical protocol we introduce a number of novel techniques for
  relating the outputs and inputs of OTs in a larger construction.
  
  We also report on an implementation of this approach, that shows
  that our protocol is more efficient than any previous one: For big
  enough circuits, we can evaluate more than $20000$ Boolean gates per
  second. As an example, evaluating one oblivious AES encryption
  ($\sim 34000$ gates) takes $64$ seconds, but when repeating the task
  $27$ times it only takes less than $3$ seconds per instance.  
\end{abstract}

\vspace{-2ex}

\renewcommand{\topfraction}{0.9}
\renewcommand{\bottomfraction}{0.9}
\renewcommand{\textfraction}{0.1}

\newcommand{\cnote}[1]{{\sf (Claudio's Note:} {\sl{#1}} {\sf EON)}}

\pagestyle{plain}
\sloppy
\pagenumbering{arabic}

\longOnly{
\clearpage
\tableofcontents
\clearpage
}

\Section{Introduction}

Secure two-party computation (2PC), introduced by Yao~\cite{DBLP:conf/focs/Yao82b}, allows two parties to jointly
compute any function of their inputs in such a way that 1) the output of
the computation is correct and 2) the inputs are kept private. Yao's
protocol is secure only if the participants are \emph{semi-honest}
(they follow the protocol but try to learn more than they should by
looking at their transcript of the protocol). A more realistic
security definition considers \emph{malicious adversaries}, that can
arbitrarily deviate from the protocol.

%In the theory of 2PC 
A large number of approaches to 2PC have been
proposed, falling into three main types, those based on Yao's garbled
circuit techniques, those based on some form of homomorphic encryption
and those based on oblivious transfer. Recently a number of efforts to
implement 2PC in practice have been reported on; In sharp contrast to
the theory, almost all of these are based on one type of 2PC, namely Yao's
garbled circuit technique. One of the main advantages of Yao's garbled circuits is that it is 
primarily based on symmetric primitives: It uses one OT per input 
bit, but then uses only a few calls to, e.g., a hash function per 
gate in the circuit to be evaluated. The other approaches are 
heavy on public-key primitives which are typically orders of 
magnitude slower than symmetric primitives.

However, in 2003 Ishai \emph{et al.} introduced the idea of extending
OTs \emph{efficiently}~\cite{DBLP:conf/crypto/IshaiKNP03}---their 
protocol allows to turn $\kappa$ seed OTs based on public-key crypto
into any polynomial $\ell=\poly(\kappa)$ number of OTs using only $O(\ell)$ invocations of a cryptographic hash function.
For big enough $\ell$ the cost of the $\kappa$ seed OTs is amortized away
and OT extension essentially turns OT into a symmetric primitive in terms 
of its computational complexity. Since the basic approach
of basing 2PC on OT in~\cite{DBLP:conf/stoc/GoldreichMW87} is
efficient in terms of consumption of OTs and communication, this gives
the hope that OT-based 2PC too could be practical. This paper reports
on the first implementation made to investigate the practicality of
OT-based 2PC.

%\footnote{Independently from   us~\cite{cryptoeprint:2011:257} also reported on an implementation   of a 2PC protocol based on OT-extension, however their protocol is   only secure against a semi-honest adversary.}.

%The approach to extending OTs efficiently in
%\cite{DBLP:conf/crypto/IshaiKNP03} is only passive-secure: later attempts %in~\cite{DBLP:conf/tcc/HarnikIKN08,cryptoeprint:2007:215,DBLP:conf/crypto/IshaiPS08} 
%are asymptotically very efficient but they do not seem efficient in practice. Also,
%only the passive-secure 2PC protocol in~\cite{DBLP:conf/stoc/GoldreichMW87} is
%practical---the active-secure protocol consumes $O(\kappa)$ OTs per
%gate, where $\kappa$ is the security parameter. 

Our starting point is the efficient passive-secure OT extension
protocol of~\cite{DBLP:conf/crypto/IshaiKNP03} and passive-secure 2PC
of~\cite{DBLP:conf/stoc/GoldreichMW87}. In order to get active
security and preserve the high practical efficiency of these protocols
we chose to develop substantially different techniques,
differentiating from other works that were only interested in
\emph{asymptotic}
efficiency~\cite{DBLP:conf/tcc/HarnikIKN08,cryptoeprint:2007:215,DBLP:conf/crypto/IshaiPS08}. We report a number of contributions to the theory and practice of 2PC:
\begin{enumerate}
\item

  We introduce a new technical idea to the area of extending OTs
  efficiently, which allows to dramatically improve the practical 
  efficiency of active-secure OT extension.  Our protocol has the same
  asymptotic complexity as the previously best protocol
  in~\cite{DBLP:conf/tcc/HarnikIKN08}, but it is only a small factor 
  slower than the passive-secure protocol
  in~\cite{DBLP:conf/crypto/IshaiKNP03}.

\item

  We give the first implementation of the idea of extending OTs
  efficiently. The protocol is active-secure and generates
  $500\mathord{,}000$ OTs per second, showing that implementations
  needing a large number of OTs can be practical.

\item

  We introduce new technical ideas which allow to relate the outputs
  and inputs of OTs in a larger construction, via the use of information
  theoretic tags. This can be seen as a new flavor of committed OT that 
  only requires symmetric cryptography. In combination with  
  our first contribution, our protocol shows how to efficiently extend 
  committed OT. Our protocols assume
  the existence of OT and are secure in the random oracle model.

\item

  We give the first implementation of practical 2PC not based on Yao's
  garbled circuit technique. Introducing a new practical
  technique is a significant contribution to the field in itself. In
  addition, our protocol shows favorable timings compared to the
  Yao-based implementations.

\end{enumerate}

\vspace{-0.3cm}
\Subsection{Comparison with Related Work}

The question on the \emph{asymptotical} computational overhead of
cryptography was (essentially) settled
in~\cite{DBLP:conf/stoc/IshaiKOS08}. On the other hand, there is
growing interest in understanding the \emph{practical} overhead of
secure computation, and several works have perfected and implemented
protocols based on Yao garbled
circuits~\cite{DBLP:conf/uss/MalkhiNPS04,DBLP:conf/ccs/Ben-DavidNP08,DBLP:conf/scn/LindellPS08,DBLP:conf/icalp/KolesnikovS08,DBLP:conf/asiacrypt/PinkasSSW09,Henecka:2010:TTA:1866307.1866358,cryptoeprint:2010:584,cryptoeprint:2010:284,DBLP:conf/eurocrypt/ShelatS11,DBLP:conf/uss/HuangEKM11},
protocols based on homomorphic
encryption~\cite{DBLP:conf/tcc/IshaiPS09,DBLP:conf/crypto/DamgardO10,DBLP:conf/acns/JakobsenMN10,DBLP:conf/eurocrypt/BendlinDOZ11}
and protocols based on
OT~\cite{DBLP:conf/crypto/IshaiPS08,DBLP:conf/crypto/LindellOP11,cryptoeprint:2011:257}.

\begin{wraptable}[11]{l}{8.5cm}
\vspace{-0.7cm}
\begin{center}
\begin{small}

\begin{tabular}{|c|l|c|c|c|c|}
\hline
     &  & Security  & Model & Rounds & Time  \\ \hline
 (a) & DK~\cite{DBLP:conf/fc/DamgardK10} (3 parties) & Passive & SM & $O(d)$ & $1.5$s \\ \hline
 (b) & DK~\cite{DBLP:conf/fc/DamgardK10} (4 parties) & Active & SM & $O(d)$ & $4.5$s \\ \hline
 (c) & sS~\cite{DBLP:conf/eurocrypt/ShelatS11} & Active & SM & $O(1)$ & $192$s \\ \hline
 (d) & HEKM~\cite{DBLP:conf/uss/HuangEKM11} & Passive & ROM & $O(1)$ & $0.2$s \\ \hline
 (e) & IPS-LOP~\cite{DBLP:conf/crypto/IshaiPS08,DBLP:conf/crypto/LindellOP11} & Active & SM & $O(d)$ & $79$s \\ \hline
 (f) & This (single) & Active & ROM & $O(d)$ & $64$s  \\ \hline
 (g) & This (27, amortized) & Active & ROM & $O(d)$ & $2.5$s \\ \hline
\end{tabular}
\end{small}
\caption{Brief comparison with other implementations.}
\label{comparison}
\end{center}
\vspace{-1cm}
\end{wraptable}

A brief comparison of the time needed for oblivious AES
evaluation for the best known implementations are shown in
Table~\ref{comparison}.\footnote{Oblivious AES has become one of the most common
circuits to use for benchmarking generic MPC protocols, due to its reasonable 
size (about 30000 gates) and its relevance as a building block for constructing specific purpose protocols, like private set intersection~\cite{DBLP:conf/tcc/FreedmanIPR05}.}
The protocols in rows (a-b) are for 3 and 4
parties respectively, and are secure against at most one corrupted
party. One of the goals of the work in row (c) is how to efficiently
support different outputs for different parties: in our OT based
protocol this feature comes for free. The time in row (e) is an
estimate made by \cite{DBLP:conf/crypto/LindellOP11} on the running
time of their optimized version of the OT-based protocol
in~\cite{DBLP:conf/crypto/IshaiPS08}. The column \emph{Round}
indicates the round complexity of the protocols, $d$ being the depth
of the circuit while the column \emph{Model} indicates whether the
protocol was proven secure in the standard model (SM) or the random
oracle model (ROM).

The significance of this work is shown in row (g). The reason for the
dramatic drop between row (f) and (g) is that in (f), when we only
encrypt one block, our implementation preprocesses for many more gates
than is needed, for ease of implementation. In (g) we encrypt $27$
blocks, which is the minimum value which eats to up all the
preprocessed values. We consider these results positive: our
implementation is as fast or faster than any other 2PC protocol, even
when encrypting only one block. And more importantly, when running at
full capacity, the price to pay for active security is about a factor
$10$ against the passive-secure protocol in (d). We stress that this
is only a limited comparison, as the different experiments were run on
different hardware and network setups: when several options were
available, we selected the best time reported by the other
implementations. See~\secref{exp} for more timings and details of our
implementation.

%The best active-secure Yao-based protocol to date is
%arguably~\cite{cryptoeprint:2010:284}. With respect
%to~\cite{cryptoeprint:2010:284}, our protocol is asymptotically more
%efficient (both in the amount of public and symmetric key
%operations). On the other hand we cannot give a proper comparison for
%particular circuit sizes, as no implementation
%of~\cite{cryptoeprint:2010:284} is publicly available.

\Subsection{Overview of Our Approach}

We start from a classic textbook protocol for two-party
computation~\cite[Sec.~7.3]{Goldreich}. In this protocol, Alice holds
secret shares $x_A,y_A$ and Bob holds secret shares $x_B,y_B$ of some
bits $x,y$ s.t.~$x_A \oplus x_B = x$ and $y_A \oplus y_B = y$. Alice
and Bob want to compute secret shares of $z=g(x,y)$ where $g$ is some
Boolean gate, for instance the AND gate: Alice and Bob need to compute
a random sharing $z_A,z_B$ of $z=x y =x_A y_A \oplus x_A y_B \oplus
x_By_A \oplus x_B y_B$. The parties can compute the AND of their local
shares ($x_A y_A$ and $x_By_B$), while they can use oblivious transfer
(OT) to compute the cross products ($x_A y_B$ and $x_By_A$). Now the
parties can iterate for the next layer of the circuit, up to the end
where they will reconstruct the output values by revealing their
shares.

This protocol is secure against a semi-honest adversary: assuming the
OT protocol to be secure, Alice and Bob learn nothing about the
intermediate values of the computation. It is easy to see that if a
large circuit is evaluated, then the protocol is not secure against a
malicious adversary: any of the two parties could replace values on any
of the internal wires, leading to a possibly incorrect output and/or
leakage of information.

% \begin{wrapfigure}{r}{7cm}
% \begin{center}
% \begin{tikzpicture}[scale=0.60]
% % Draw the nodes
% \node (twopc) at (0,0) {\FMPC};
% \node (deal) at (0, -2) {\DEAL};
% \node (aot) at (0, -4) {\AOT};
% \node[anchor=west] (aand) at (1, -4) {\AAND};
% \node (abit) at (0, -6) {\ABIT};
% \node (wabit) at (0, -8) {\WABIT};
% \node (labit) at (0, -10) {\LABIT};
% \node[anchor=west] (ot) at (0.5, -11.5) {\OT};
% \node[anchor=east] (eq) at (-0.5, -11.5) {\EQ};

% % Draw the arrows
% \draw[->,>=stealth'] (deal) -- (twopc);
% \draw[->,>=stealth'] (aot) -- (deal);
% \draw[->,>=stealth'] (aand) -- (deal);
% \draw[->,>=stealth'] (abit) -- (aot);
% \draw[->,>=stealth'] (abit) -- (aand);
% \draw[->,>=stealth'] (wabit) -- (abit);
% \draw[->,>=stealth'] (labit) -- (wabit);
% \draw[->,>=stealth'] (ot) -- (labit);
% \draw[->,>=stealth'] (eq) -- (labit);
% \draw[->,>=stealth'] (abit) to [bend left=45] (deal);

% % Draw the braces
% \draw[decorate, decoration={brace}] (3,0) --
% (3,-2) node[right, midway] {\small \secref{twopcfromaot}};
% \draw[decorate, decoration={brace}] (3, -4) --
% (3, -5.95) node[right, midway] {\small \secref{aot} and \ref{sec:aand}};
% \draw[decorate, decoration={brace}] (3, -6.05) -- (3, -11.5) node[right, midway] {\small
%   \secref{abit}};

% % Very ugly hack to move picture to the right
% \node (balance) at (-5,0) {};

% \end{tikzpicture}
% \caption{Paper outline.}
% \figlab{outlinebitauth} 
% \end{center}
% \end{wrapfigure}

\begin{wrapfigure}[13]{r}{5cm}
\vspace{-1.5cm} % Even worse hack to move it up
\begin{center}
\begin{tikzpicture}[scale=0.60]
% Draw the nodes
\node (twopc) at (0,0) {\FMPC};
\node (deal) at (0, -2) {\DEAL};
\node (aot) at (0, -4) {\AOT};
\node[anchor=west] (aand) at (1, -4) {\AAND};
\node (abit) at (0, -6) {\ABIT};
\node[anchor=west] (eq) at (0.5, -7.5) {\EQ};
\node[anchor=east] (ot) at (-0.5, -7.5) {\OT};

% Draw the arrows
\draw[->,>=stealth'] (deal) -- (twopc);
\draw[-,>=stealth'] (aot) -- (deal);
\draw[-,>=stealth'] (aand) -- (deal);
\draw[->,>=stealth'] (abit) -- (aot);
\draw[->,>=stealth'] (abit) -- (aand);
\draw[->,>=stealth'] (ot) -- (abit);
\draw[->,>=stealth'] (eq) -- (abit);
\draw[-,>=stealth'] (abit) to [bend left=45] (deal);
\draw[->,>=stealth'] (eq) -- (aand);
\draw[->,>=stealth'] (eq) to [bend right=20] (aot);

% Draw the braces
\draw[decorate, decoration={brace}] (3.05,0) --
(3.05,-3.95) node[right, midway] {\small \secref{twopcfromaot}};
\draw[decorate, decoration={brace}] (3.05, -4) --
(3.05, -5.95) node[right, midway] {\small \secref{aot} and \ref{sec:aand}};
\draw[decorate, decoration={brace}] (3.05, -6.05) -- (3.05, -7.5) 
node[right, midway] {\small \secref{abit}};

% Very ugly hack to move picture to the right
\node (balance) at (-1,0) {};

\end{tikzpicture}
\vspace{-0.4cm}
\caption{Paper outline. This order of presentation is chosen to allow
  the best progression in introduction of our new techniques.}
\figlab{paperoutline}
\end{center}
\end{wrapfigure}
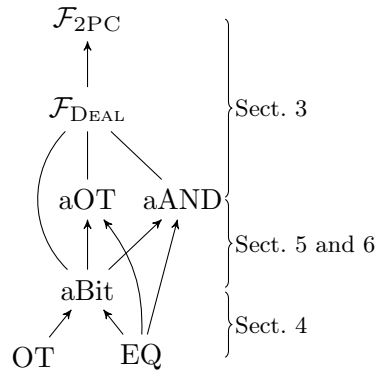

To cope with this, we put MACs on all bits. The starting point of our
protocol is \defterm{oblivious authentication} of bits. One party, the
\defterm{key holder}, holds a uniformly random \defterm{global key}
$\Delta \in \zo^\kappa$. The other party, the \defterm{MAC holder},
holds some secret bits ($x,y$, say). For each such bit the key holder
holds a corresponding uniformly random \defterm{local key} ($K_{x},
K_{y} \in \zo^\kappa$) and the MAC holder holds the corresponding
\defterm{MAC} ($M_{x} = K_{x} \oplus x \Delta$, $M_{y} = K_{y} \oplus
y \Delta$). The key holder does not know the bits and the MAC holder
does not know the keys. Note that $M_{x} \oplus M_{y} = (K_{x} \oplus
K_{y}) \oplus (x \oplus y)\Delta$. So, the MAC holder can locally
compute a MAC on $x \oplus y$ under the key $K_{x} \oplus K_{y}$ which
is non-interactively computable by the key holder. This homomorphic
property comes from fixing $\Delta$ and we exploit it throughout our
constructions.  From a bottom-up look, our protocol is constructed as
follows (see \figref{paperoutline} for the main structure):
\begin{description}
\item[Bit Authentication:] 

  We first implement oblivious authentication of bits (\ABIT). As
  detailed in~\secref{abit}, to construct authenticated bits we start
  by extending a few (say $\kappa = 640$) seed $\binom{2}{1}$-OTs into
  many (say $\ell = 2^{20}$) OTs, using OT extension.  Then, if
  $\Ali$ wants to get a bit $x$ authenticated, she can input it as the
  choice bit in an OT, while $\Bob$ can input $(K_{x},K_{x}\oplus
  \Delta)$, playing the sender in the OT. Now $\Ali$ receives $M_{x}=
  K_{x} \oplus x \Delta$. It should, of course, be ensured that even a
  corrupted $\Bob$ uses the same value $\Delta$ in all OTs. I.e., it
  should hold for all produced OTs that the XORs of the offered
  message pairs are constant---this constant value is then taken to be
  $\Delta$. It turns out, however, that when using the highly
  efficient \emph{passive-secure} OT extender
  in~\cite{DBLP:conf/crypto/IshaiKNP03} and starting from seed OTs
  where the XORs of message pairs are constant, one also produces OTs
  where the XORs of message pairs are constant, and we note that for
  this use the protocol in~\cite{DBLP:conf/crypto/IshaiKNP03} happens
  to be \emph{active-secure}! Using cut-and-choose we ensure that most
  of the XORs of message pairs offered in the seed OTs are constant,
  and with a new and inexpensive trick we offer privacy and
  correctness even if few of these XORs have different values. This
  cut-and-choose technique uses one call to a box $\EQ$ for checking
  equality.

\item[Authenticated local AND:] 

  From \ABITs we then construct \defterm{authenticated local AND}s
  (\AAND), where the MAC holder locally holds random authenticated
  bits $a, b, c$ with $c = ab$.  To create authenticated local ANDs,
  we let one party compute $c = ab$ for random $a$ and $b$ and get
  authentications on $a,b,c$ (when creating \AANDs, we assume the
  \ABITs are already available). The challenge is to ensure that $c =
  ab$. We construct an efficient proof for this fact, again using the
  box $\EQ$ once. This proof might, however, leak the bit $a$ with
  small but noticeable probability. We correct this using a combiner.

\item[Authenticated OT:] 

  From \ABITs we also construct \defterm{authenticated OT}s (\AOT),
  which are normal $\binom{2}{1}$-OTs of bits, but where all input
  bits and output bits are obliviously authenticated. This is done by
  letting the two parties generate \ABITs representing the sender
  messages $x_0,x_1$ and the receiver choice bit $c$.  To produce the
  receiver's output, first a random \ABIT is sampled. Then this bit is
  ``corrected'' in order to be consistent with the run of an OT
  protocol with input messages $x_0,x_1$ and choice bit $c$. This
  correction might, however, leak the bit $c$ with small but
  noticeable probability. We correct this using an OT combiner. One
  call to the box $\EQ$ is used.

\item[2PC:] 

  Given two \AANDs and two \AOTs one can evaluate in a very efficient
  way any Boolean gate: only $4$ bits per gate are communicated, as
  the MACs can be checked in an amortized manner.

\end{description}

That efficient 2PC is possible given enough \ABITs, \AANDs and \AOTs
is no surprise. In some sense, it is the standard way to base
passive-secure 2PC on passive-secure OT enhanced with a particular
flavor of committed OT (as in
\cite{DBLP:conf/crypto/CrepeauGT95,DBLP:conf/tcc/Garay04}).  What is
new is that we managed to find a particular committed OT-like
primitive which allows both a very efficient generation and a very
efficient use: while previous result based on committed OT require
hundreds of \emph{exponentiations} per gate, our cost per gate is in
the order of hundreds of \emph{hash functions}.  To the best of our
knowledge, we present the first practical approach to extending a few
seed OTs into a large number of committed OT-like primitives. Of more
specific technical contributions, the main is that we manage to do all
the proofs efficiently, thanks also to the preprocessing nature of our
protocol: Creating \ABITs, we get active security paying only a
constant overhead over the passive-secure protocol
in~\cite{DBLP:conf/crypto/IshaiKNP03}. In the generation of \AANDs and
\AOTs, we replace cut-and-choose with efficient, slightly leaky proofs
and then use a combiner to get rid of the leakage: When we preprocess
for $\ell$ gates and combine $B$ leaky objects to get each potentially
unleaky object, the probability of leaking is $(2\ell)^{-B} =
2^{-\log_2(\ell)(B-1)}$. As an example, if we preprocess for $2^{20}$
gates with an overhead of $B = 6$, then we get leakage probability
$2^{-100}$.

As a corollary to being able to generate any $\ell = \poly(\kappa)$
active-secure \ABIT{}s from $O(\kappa)$ seed OTs and $O(\ell)$ calls
to a hash-function, we get that we can generate any $\ell =
\poly(\kappa)$ active-secure $\binom{2}{1}$-\OT{}s of $\kappa$-bit
strings from $O(\kappa)$ seed OTs and $O(\ell)$ calls to a
hash-function, matching the asymptotic complexity
of~\cite{DBLP:conf/tcc/HarnikIKN08} while dramatically reducing their
hidden constants.

\Section{Preliminaries and Notation} \seclab{preliminaries}

We use $\kappa$ (and sometimes $\psi$) to denote the security
parameter. We require that a poly-time adversary break the protocol
with probability at most $\poly(\kappa) 2^{-\kappa}$. For a bit-string
$S \in \zo^*$ we define $0 S \Def 0^{|S|}$ and $1 S \Def S$. For a
finite set $S$ we use $s \inR S$ to denote that $s$ is chosen
uniformly at random in $S$. For a finite distribution $D$ we use $x
\leftarrow D$ to denote that $x$ is sampled according to $D$.

\Subsubsection{The UC Framework}\seclab{uc}

We prove our results static, active-secure in the UC
framework~\cite{DBLP:conf/focs/Canetti01}, and we assume the reader to
be familiar with it. We will idiosyncratically use the word \emph{box}
instead of the usual term \emph{ideal functionality}. To simplify the
statements of our results we use the following terminology:
\begin{definition}
  We say that a box $\func{A}$ is \emph{reducible} to a box $\func{B}$
  if there exist an actively secure implementation $\pi$ of $\func{A}$
  which uses only one call to $\func{B}$.  We say that $\func{A}$ is
  \emph{locally} reducible to $\func{B}$ if the parties of $\pi$ do
  not communicate (except through the one call to $\func{B}$).  We say
  that $\func{A}$ is \emph{linear} reducible to $\func{B}$ if the
  computing time of all parties of $\pi$ is linear in their inputs and
  outputs. We use \emph{equivalent} to denote reducibility in both
  directions.
\end{definition}

\longOnly{It is easy to see that if $\func{A}$ is (linear, locally)
reducible to $\func{B}$ and $\func{B}$ is (linear, locally) reducible
to $\func{C}$, then $\func{A}$ is (linear, locally) reducible to
$\func{C}$. }

\Subsubsection{Hash Functions}

We use a hash function $H: \zo^* \rightarrow \zo^\kappa$, which we
model as a random oracle (RO).  We sometimes use $H$ to mask a
message, as in $H(x) \oplus M$. If $\vert M \vert \ne \kappa$, this
denotes $\prg(H(x)) \oplus M$, where $\prg$ is a pseudo-random
generator $\prg: \zo^\kappa \rightarrow \zo^{\vert M \vert}$. We also
use a collision-resistant hash function $G: \zo^{2\kappa} \rightarrow
\zo^\kappa$.
  
As other 2PC protocols whose focus is
efficiency~\cite{DBLP:conf/icalp/KolesnikovS08,DBLP:conf/uss/HuangEKM11},
we are content with a proof in the random oracle model. What is the
exact assumption on the hash function that we need for our protocol to
be secure, as well as whether this can be implemented under standard
cryptographic assumption is an interesting theoretical question,
see~\cite{cryptoeprint:2010:544,cryptoeprint:2011:510}.

\Subsubsection{Oblivious Transfer} We use a box $\OT(\tau,\ell)$
which can be used to perform $\tau$ $\binom{2}{1}$-oblivious
transfers of strings of bit-length $\ell$. In each of the $\tau$ OTs
the sender \ant{S} has two inputs $x_0, x_1 \in \zo^\ell$, called the
\emph{messages}, and the receiver \ant{R} has an input $c \in \zo$,
called the \emph{choice bit}. The output to \ant{R} is $x_c = c(x_0
\oplus x_1) \oplus x_0$. No party learns any other information.

\Subsubsection{Equality Check} We use a box $\EQ(\ell)$ which
allows two parties to check that two strings of length $\ell$ are
equal. If they are different the box leaks both strings to the
adversary, which makes secure implementation easier.  We define and
use this box to simplify the exposition of our protocol. In practice
we implement the box by letting the parties compare exchanged
hash's of their values: this is a secure implementation of the box in
the random oracle model.

For completeness we give a protocol which securely implements $\EQ$ in
the RO model. Let $H :\zo^* \rightarrow \zo^\kappa$ be a hash
function, modeled as a RO. Let $\kappa$ be the security parameter.
\begin{enumerate}
\item \Ali chooses a random string $r \in_R \zo^\kappa$, computes
  $c=H(x||r)$ and sends it to \Bob.
\item \Bob sends $y$ to \Ali.
\item \Ali sends $x,r$ to \Bob. \Ali outputs $x \eqq y$.
\item \Bob outputs $(H(x||r) \eqq c) \wedge (x \eqq y)$.
\end{enumerate}

This is a secure implementation of the $\EQ(\ell)$ functionality in
the RO model. If \Ali is corrupted, the simulator extracts $x,r$ from
the simulated call to the RO, if the hash function was queried with an
input which yielded the $c$ sent by \Ali. Then, it inputs $x$ to $\EQ$
and receives $(x,y)$ from the ideal functionality (if $x \neq y$). If
the hash function was not queried with an input which yielded the $c$
sent by \Ali, then the simulator inputs a uniformly random $x$ to
$\EQ$ and receives $(x,y)$. It then sends $y$ to the corrupted
\Ali. On input $x',r'$ from \Ali, if $(x',r')\neq (x,r)$ the simulator
inputs ``abort'' to the \EQ functionality on behalf of \Ali, or
``deliver'' otherwise. If $(x',r')=(x,r)$, simulation is perfect. If
they are different, the only way that the environment can distinguish
is by finding $(x',r')\neq(x,r)$ s.t.~$H(x||r)=H(x'||r')$ or by
finding $(x',r')$ such that $c = H(x'||r')$ for a $c$ which did not
result from a previous query.  In the random oracle both events happen
with probability less than $\poly(\kappa)2^{-\kappa}$, as the
environment is only allowed a polynomial number of calls to the RO.

If \Bob is corrupted, then the simulator sends a random value $c\in_R
\zo^\ell$ to \Bob. Then, on input $y$ from \Bob it inputs this value
to the \EQ box and receives $(x,y)$. Now, it chooses a random $r\in_R
\zo^\kappa$ and programs the RO to output $c$ on input $x||r$, and
sends $x$ and $r$ to \Bob. Simulation is perfect, and the environment
can only distinguish if it had already queried the RO on input $x||r$,
and this happens with probability $\poly(\kappa)2^{-\kappa}$, as $r
\in \zo^{\kappa}$ is uniformly random, and the environment is only
allowed a polynomial number of calls to the RO.

\Subsubsection{Leakage Functions} \seclab{leak}

We use a notion of a class $\LL$ of \defterm{leakage functions on $\tau$
bits}. The context is that there is some uniformly random secret value
$\Delta \inR \zo^\tau$ and some adversary \Ali wants to guess
$\Delta$. To aid \Ali, she can do an attack which might leak some of
the bits of $\Delta$. The attack, however, might be detected. Each $L
\in \LL$ is a poly-time sampleable distribution on $(S,c) \in 2^{\{1,
  \ldots, \tau\}} \times \zo$. Here $c$ specifies if the attack was
detected, where $c=0$ signals detection, and $S$ specifies the bits
to be leaked if the attack was not detected.
We need a measure of how many bits a class $\LL$ leaks. We do this via
a game for an unbounded adversary \Ali.
\begin{enumerate}
\item

  The game picks a uniformly random $\Delta \inR \zo^\tau$.

\item

  \Ali inputs $L \in \LL$.

\item

  The game samples $(S,c) \leftarrow L$. If $c=0$, \Ali loses. If
  $c=1$, the game gives $\{ (i, \Delta_i) \}_{i \in S}$ to \Ali.

\item

  Let $\overline{S} = \{1,\ldots,\tau\} \setminus S$. \Ali inputs
  the guesses $\{ (i, g_i) \}_{i \in \overline{S}}$. If $g_i=\Delta_i$ for
  all $i \in \overline{S}$, \Ali wins, otherwise she loses.

\end{enumerate}

We say that an adversary \Ali is \emph{optimal} if she has the highest
possible probability of winning the game above. If there were no
leakage, i.e., $S = \emptyset$, then it is clear that the optimal \Ali
wins the game with probability exactly $2^{-\tau}$. If \Ali is always
given exactly $s$ bits and is never detected, then it is clear that
the optimal \Ali can win the game with probability exactly
$2^{s-\tau}$. This motivates defining the number of bits leaked by
$\LL$ to be $\leak_\LL \Def \log_2(\success_\LL) + \tau$, where
$\success_\LL$ is the probability that the optimal \Ali wins the
game. It is easy (details below) to see that if we take expectation
over random $(S,c)$ sampled from $L$, then $\leak_\LL = \max_{L \in \LL}
\log_2\left(\E{c 2^{\vert S \vert}}\right)$.

We say that $\LL$ is $\kappa$-secure if $\tau - \leak_\LL \ge \kappa$,
and it is clear that if $\LL$ is $\kappa$-secure, then no \Ali can win
the game with probability better than $2^{-\kappa}$.

We now rewrite the definition of $\leak_\LL$ to make it more workable.

It is clear that the optimal \Ali can guess all $\Delta_i$ for $i \in
\overline{S}$ with probability exactly $2^{\vert S \vert - \tau}$.
This means that the optimal \Ali wins with probability
$\sum_{s=0}^\tau \prob{(S,c) \leftarrow L: \vert S \vert = s \wedge
c=1} 2^{s-\tau}$. To simplify this expression we define index
variables $I_s, J_s \in \zo$ where $I_s$ is $1$ iff $c=1$ and $\vert S
\vert = s$ and $J_s$ is $1$ iff $\vert S \vert = s$. Note that $I_s =
c J_s$ and that $\sum_s J_s 2^s = 2^{\vert S \vert}$. So, if we take
expectation over $(S,c)$ sampled from $L$, then we get that
\begin{equation*}
  \begin{split}
    \sum_{s=0}^\tau \prob{(S,c) \leftarrow L: \vert S \vert = s \wedge
      c=1} 2^s 
    &= \sum_{s=0}^\tau \E{I_s} 2^s\\ 
    &= \E{\sum_{s=0}^\tau I_s 2^s} 
    = \E{\sum_{s=0}^\tau c J_s 2^s}\\
    &= \E{c \sum_{s=0}^\tau J_s 2^s}
    = \E{c 2^{\vert S \vert}}\ .
  \end{split}
\end{equation*}
Hence $\success_L = 2^{-\tau} \E{c 2^{\vert S \vert}}$ is the
probability of winning when using $L$ and playing optimal. Hence
$\success_\LL = \max_{L \in \LL} (2^{-\tau} \E{c 2^{\vert S \vert}})$
and $\log_2 (\success_\LL )= - \tau + \log_2 \max_{L \in \LL}\left( \E{c
2^{\vert S \vert}}\right)$, which shows that
$$\leak_\LL = \max_{L \in \LL} \log_2\left(\E{c 2^{\vert S \vert}}\right)\ ,$$
as claimed above.

\Section{The Two-Party Computation Protocol}\seclab{twopcfromaot}
\begin{wrapfigure}[6]{r}{4cm}
\vspace{-1.5cm} % Even worse hack to move it up
\begin{center}
\begin{tikzpicture}[scale=0.60]
% Draw the nodes
\node (twopc) at (0,0) {\FMPC};
\node (deal) at (0, -2) {\DEAL};
\node (aot) at (0, -4) {\AOT};
\node[anchor=west] (aand) at (1, -4) {\AAND};
\node[anchor=east] (abit) at (-1, -4) {\ABIT};

% Draw the arrows
\draw[->,>=stealth'] (deal) -- (twopc);
\draw[-,>=stealth'] (aot) -- (deal);
\draw[-,>=stealth'] (aand) -- (deal);
\draw[-,>=stealth'] (abit) -- (deal);
\end{tikzpicture}
\vspace{-0.3cm}
\caption{\secref{twopcfromaot} outline.}
\figlab{dealoutline}
\end{center}
\end{wrapfigure}
We want to implement the box \FMPC for Boolean two-party secure
computation as described in~\figref{FMPC}. We will implement this box
in the \DEAL-hybrid model of~\figref{FDEAL}. This box provides the
parties with \ABITs, \AANDs and \AOTs, and models the preprocessing
phase of our protocol. We introduce notation in~\figref{abitnotation}
for working with authenticated bits. The protocol implementing \FMPC
in the dealer model is described in \figref{PI2PC}. The dealer offers
random authenticated bits (to \Ali or \Bob), random authenticated
local AND triples and random authenticated OTs. Those are all the
ingredients that we need to build the 2PC protocol. \longOnly{Note
  that the dealer offers randomized versions of all commands: this is
  not a problem as the ``standard'' version of the commands (the one
  where the parties can specify their input bits instead of getting
  them at random from the box) are linearly reducible to the
  randomized version, as can be easily deduced from the protocol
  description.} The following result is proven in
\appref{proofoftwopc}:

\begin{theorem}\thmlab{twopc} 
  The protocol
  in~\figref{PI2PC} securely implements the box \FMPC in the
  \DEAL-hybrid model with security parameter $\kappa$.
\end{theorem}

\begin{boxfig}{Notation for authenticated and shared bits.}{abitnotation}
\begin{description}
\item[Global Key:] 

We call $\Delta_A,\Delta_B \in \zo^{\kappa}$ the two \defterm{global
keys}, held by $\Bob$ and $\Ali$ respectively.

\item[Authenticated Bit:] 

We write $\repa{x}$ to represent an \defterm{authenticated secret bit}
held by $\Ali$. Here $\Bob$ knows a key $K_x \in \zo^\kappa$ and
$\Ali$ knows a bit $x$ and a MAC $M_x = K_x \oplus x \Delta_A \in
\zo^{\kappa}$. Let $\repa{x}~\Def~(x,M_x,K_x)$.\footnote{Since
$\Delta_A$ is a global value we will not always write it
explicitly. Note that in $x \Delta_A$, $x$ represents a \emph{value},
$0$ or $1$, and that in $\repa{x}$, $K_x$ and $M_x$ it represents a
\emph{variable name}. I.e., there is only one key (MAC) per
authenticated bit, and for the bit named $x$, the key (MAC) is named
$K_x$ ($M_x$). If $x=0$, then $M_x = K_x$. If $x=1$, then $M_x = K_x
\oplus \Delta_A$.}  

If $\repa{x}=(x,M_x,K_x)$ and
$\repa{y}=(y,M_y,K_y)$ we write $\repa{z} = \repa{x}\oplus \repa{y}$
to indicate $\repa{z}=(z,M_z,K_z)\Def(x\oplus y, M_x\oplus M_y, K_x
\oplus K_y)$. Note that no communication is required to compute 
$\repa{z}$ from $\repa{x}$ and $\repa{y}$.

It is possible to authenticate a constant bit (a value known both to
$\Ali$ and $\Bob$) $b\in \zo$ as follows: $\Ali$ sets $M_b=
0^{\kappa}$, $\Bob$ sets $K_b=b\Delta_A$, now $\repa{b}\Def
(b,M_b,K_b)$. For a constant $b$ we let $\repa{x}\oplus b \Def
\repa{x} \oplus \repa{b}$, and we let $b\repa{x}$ be equal to
$\repa{0}$ if $b=0$ and $\repa{x}$ if $b=1$.

We say that $\Ali$ \defterm{reveals} $\repa{x}$ by sending ($x,M_x$)
to $\Bob$ who aborts if $M_x\neq K_x \oplus x\Delta_A$.  Alternatively
we say that $\Ali$ \defterm{announces} $x$ by sending $x$ to $\Bob$
without a MAC.

Authenticated bits belonging to \Bob are written as $\repb{y}$ and are
defined symmetrically, changing side of all the values and using the
global value $\Delta_B$ instead of $\Delta_A$. 

\item[Authenticated Share:] 

  We write $\reps{x}$ to represent the situation where $\Ali$ and
  $\Bob$ hold $\repa{x_A},\repb{x_B}$ and $x=x_A\oplus x_B$, and we
  write $\reps{x}=(\repa{x_A},\repb{x_B})$ or
  $\reps{x}=\repab{x_A}{x_B}$.

  If $\reps{x}=\repab{x_A}{x_B}$ and $\reps{y}=\repab{y_A}{y_B}$ we
  write $\reps{z}=\reps{x}\oplus \reps{y}$ to indicate
  $\reps{z}=(\repa{z_A},\repb{z_B})=(\repa{x_A}\oplus \repa{y_A},
  \repb{x_B}\oplus \repb{y_B})$. Note that no communication is
  required to compute $\reps{z}$ from $\reps{x}$ and $\reps{y}$.

  It is possible to create an authenticated share of a constant $b\in
  \zo$ as follows: $\Ali$ and $\Bob$ create
  $\reps{b}=\repab{b}{0}$. For a constant value $b\in \zo$, we define
  $b\reps{x}$ to be equal to $\reps{0}$ if $b=0$ and $\reps{x}$ if
  $b=1$.

  When an authenticated share is \defterm{revealed}, the parties
  reveal to each other their authenticated bits and abort if the MACs
  are not correct.

\end{description}
\end{boxfig}

\begin{boxfig}{The box \FMPC for Boolean Two-party Computation.}{FMPC}
\begin{description}
\item[Rand:] 

  On input $(\cmd{rand}, \varid)$ from $\Ali$ and $\Bob$, with
  $\varid$ a fresh identifier, the box picks $r
  \inR \zo$ and stores $(\mathit{\varid},r)$.

\item[Input:] 

  On input $(\cmd{input}, \ant{P}, \mathit{\varid}, x)$ from
  $\ant{P}\in\{\Ali,\Bob\}$ and $(\cmd{input}, \ant{P},
  \mathit{\varid}, ?)$ from the other party, with $\mathit{\varid}$ a
  fresh identifier, the box stores $(\mathit{\varid},x)$.

\item[XOR:] 

  On command $(\cmd{xor}, \mathit{\varid}_1, \mathit{\varid}_2,
  \mathit{\varid}_3)$ from both parties (if
  $\mathit{\varid}_1,\mathit{\varid}_2$ are defined and
  $\mathit{\varid}_3$ is fresh), the box retrieves
  $(\mathit{\varid}_1,x)$, $(\mathit{\varid}_2,y)$ and stores
  $(\mathit{\varid}_3, x\oplus y )$.

\item[AND:] 

  As \textbf{XOR}, but store $(\mathit{\varid}_3, x\cdot y)$.

\item[Output:] 

  On input $(\cmd{output}, \ant{P}, \mathit{\varid})$ from both
  parties, with $\ant{P}\in\{\Ali,\Bob\}$ (and $\mathit{\varid}$
  defined), the box retrieves $(\mathit{\varid},x)$ and
  outputs it to $\ant{P}$.

\end{description}
\Steal{} At each command the box leaks to the environment which
command is being executed (keeping the value $x$ in \textbf{Input} secret),
and delivers messages only when the environment says so.
\end{boxfig}

\begin{boxfig}{The box \DEAL for dealing preprocessed values.}{FDEAL}
\begin{description}
\item[Initialize:] 

  On input $(\cmd{init})$ from $\Ali$ and $(\cmd{init})$ from $\Bob$,
  the box samples $\Delta_A,\Delta_B\in \zo^\kappa$, stores them and
  outputs $\Delta_B$ to $\Ali$ and $\Delta_{A}$ to $\Bob$.~If $\Ali$
  (resp.~$\Bob$) is corrupted, she gets to choose $\Delta_B$
  (resp. $\Delta_A$).

\item[Authenticated Bit (\Ali):] 

  On input $(\cmd{aBIT}, \Ali)$ from $\Ali$ and $\Bob$, the box
  samples a random $\repa{x}=(x,M_x,K_x)$ with $M_x~=~K_x\oplus x
  \Delta_A$ and outputs it ($x,M_x$ to $\Ali$ and $K_x$ to $\Bob$).~
  If $\Bob$ is corrupted he gets to choose $K_x$.~If $\Ali$ is
  corrupted she gets to choose $(x,M_x)$, and the box sets
  $K_x~=~M_x\oplus x \Delta_A$.

\item[Authenticated Bit (\Bob):] 

  On input $(\cmd{aBIT}, \Bob)$ from $\Ali$ and $\Bob$, the box
  samples a random $\repb{x}=(x,M_x,K_x)$ with $M_x~=~K_x\oplus x
  \Delta_B$ and outputs it ($x,M_x$ to $\Bob$ and $K_x$ to $\Ali$). 
  As in \textbf{Authenticated Bit (\Ali)}, corrupted parties can
  choose their own randomness.

\item[Authenticated local AND (\Ali):] 

  On input $(\cmd{aAND}, \Ali)$ from $\Ali$ and $\Bob$, the box
  samples random $\repa{x}$,$\repa{y}$ and $\repa{z}$ with $z=x y$ and
  outputs them.~As in \textbf{Authenticated Bit (\Ali)}, corrupted
  parties can choose their own randomness.

\item[Authenticated local AND (\Bob)] Defined symmetrically.

\item[Authenticated OT (\Ali-\Bob):] 

  On input $(\cmd{aOT},\Ali,\Bob)$ from $\Ali$ and $\Bob$, the box
  samples random $\repa{x_0}$,$\repa{x_1}$,$\repb{c}$ and $\repb{z}$
  with $z~=~x_{c}=c(x_0\oplus x_1) \oplus x_0$ and outputs them.  As in
  \textbf{Authenticated Bit}, corrupted parties can choose their own
  randomness.

\item[Authenticated OT (\Bob-\Ali):] 

  Defined symmetrically.\longOnly{\footnote{The dealer offers \AOTs in
  both directions. Notice that the dealer could offer \AOT only in one
  direction and the parties could then ``turn'' them: as regular OT,
  \AOT is symmetric as well.}}

\item[Global Key Queries:]

  The adversary can at any point input $(\Ali, \Delta)$ and be
  told whether $\Delta = \Delta_B$. And it can at any point input
  $(\Bob, \Delta)$ and be told whether $\Delta = \Delta_A$.

\end{description}
\end{boxfig}

\begin{boxfig}{Protocol for \FMPC in the \DEAL-hybrid model}{PI2PC}
\begin{description}
\item[Initialize:] 

  When activated the first time, $\Ali$ and $\Bob$ activate \DEAL and
  receive $\Delta_B$ and $\Delta_A$ respectively. 

\item[Rand:] 

  $\Ali$ and $\Bob$ ask \DEAL for random authenticated bits
  $\repa{r_A},\repb{r_B}$ and stores $\reps{r}=\repab{r_A}{r_B}$ under
  $\mathit{\varid}$. 

\item[Input:] 

  If $\ant{P} = \Ali$, then $\Ali$ asks \DEAL for an authenticated bit
  $\repa{x_A}$ and announces (i.e., no MAC is sent together with the
  bit) $x_B=x\oplus x_A$, and the parties build $\repb{x_B}$ and
  define $\reps{x}=\repab{x_A}{x_B}$. The protocol is symmetric for
  $\Bob$.

\item[XOR:] 

  $\Ali$ and $\Bob$ retrieve $\reps{x},\reps{y}$ stored under
  $\mathit{\varid}_1,\mathit{\varid}_2$ and store
  $\reps{z}=\reps{x}\oplus\reps{y}$ under $\mathit{\varid}_3$. For
  brevity we drop explicit mentioning of variable identifiers below.

\item[AND:] 

  $\Ali$ and $\Bob$ retrieve $\reps{x},\reps{y}$ and compute
  $\reps{z}=\reps{xy}$ as follows:
  \begin{enumerate}
  \item 

    The parties ask \DEAL for a random AND triplet
    $\repa{u},\repa{v},\repa{w}$ with $w=uv$. \\ $\Ali$ reveals
    $\repa{f}=\repa{u}\oplus \repa{x_A}$ and $\repa{g}=\repa{v}\oplus
    \repa{y_A}$. \\ The parties compute $\repa{x_Ay_A} = f \repa{y_A}
    \oplus g \repa{x_A} \oplus \repa{w} \oplus fg$.

  \item 

    Symmetrically the parties compute $\repb{x_By_B}$.
    %The parties ask \DEAL for a random AND triplet
    %$\repb{u},\repb{v},\repb{w}$ with $w=uv$. \\ $\Bob$ reveals
    %$\repb{f}=\repb{u}\oplus \repb{x_B}, \repb{g}=\repb{v}\oplus
    %\repb{y_B}$. \\
    %The parties compute $\repb{x_By_B} = f \repb{y_B} \oplus g
    %\repb{x_B} \oplus \repb{w} \oplus fg$.

  \item 

    The parties ask \DEAL for a random authenticated OT
    $\repa{u_0},\repa{u_1},\repb{c},\repb{w}$ with $w=u_c$.\\ 
    They also
    ask for an authenticated bit $\repa{r_A}$. \\ Now $\Bob$ reveals
    $\repb{d}=\repb{c} \oplus \repb{y_B}$. \\ $\Ali$ reveals
    $\repa{f}=\repa{u_0} \oplus \repa{u_1} \oplus \repa{x_A}$ and
    $\repa{g}=\repa{r_A} \oplus \repa{u_0} \oplus d \repa{x_A}$. \\
    Compute $\repb{s_B}=\repb{w} \oplus f \repb{c} \oplus
    g$. Note that at this point $\repb{s_B}=\repb{r_A \oplus
    x_Ay_B}$.

  \item 

    Symmetrically the parties compute $\repa{s_A}=\repa{r_B \oplus
    x_By_A}$.
    %The parties ask \DEAL for a random authenticated OT
    %$\repb{u_0},\repb{u_1},\repa{c},\repa{w}$ with $w=u_c$.\\ They also
    %ask for an authenticated bit $\repb{r_B}$. \\ Now $\Ali$ reveals
    %$\repa{d}=\repa{c} \oplus \repa{y_A}$; \\ $\Bob$ reveals
    %$\repb{f}=\repb{u_0} \oplus \repb{u_1} \oplus \repb{x_B}$ and
    %$\repb{g}=\repb{r_B} \oplus \repb{u_0} \oplus d \repb{x_B}$. \\
    %Compute $\repa{s_A}=\repa{w} \oplus f \repa{c} \oplus
    %g$.\footnote{Note that at this point $\repa{s_A}=\repa{r_B \oplus
    %x_By_A}$.}

  \end{enumerate}
  $\Ali$ and $\Bob$ compute $\repa{z_A}=\repa{r_A}\oplus
  \repa{s_A}\oplus \repa{x_Ay_A}$ and $\repb{z_B}=\repb{r_B}\oplus
  \repb{s_B}\oplus \repb{x_By_B}$ and let
  $\reps{z}=\repab{z_A}{z_B}$.

\item[Output:] 

  The parties retrieve $\reps{x}=\repab{x_A}{x_B}$. If $\Ali$ is to
  learn $x$, $\Bob$ reveals $x_B$. If $\Bob$ is to learn $x$, $\Ali$
  reveals $x_A$.

\end{description}
\end{boxfig}

\Subsubsection{Why the global key queries?}
The \DEAL box (\figref{FDEAL}) allows the adversary to guess the value
of the global key, and it informs it if its guess is correct. This is
needed for technical reasons: When \DEAL is proven UC secure, the
environment has access to either \DEAL or the protocol implementing
\DEAL. In both cases the environment learns the global keys $\Delta_A$
and $\Delta_B$. In particular, the environment learns $\Delta_A$ even
if \Bob is honest. This requires us to prove the sub-protocol for
\DEAL secure to an adversary knowing $\Delta_A$ even if \Bob is
honest: to be be able to do this, the simulator needs to recognize
$\Delta_A$ if it sees it---hence the global key queries. Note,
however, that in the context where we use \DEAL (\figref{PI2PC}), the
environment does \emph{not} learn the global key $\Delta_A$ when \Bob
is honest: A corrupted \Ali only sees MACs on one bit using the same
local key, so all MACs are uniformly random in the view of a corrupted
\Ali, and \Bob never makes the local keys public.

\Subsubsection{Amortized MAC checks.}
In the protocol of \figref{PI2PC}, there is no need to send MACs and
check them every time we do a ``reveal''. In fact, it is
straightforward to verify that before an \textbf{Output} command is
executed, the protocol is perfectly secure even if the MACs are not
checked. Notice then that a keyholder checks a MAC $M_x$ on a bit $x$
by computing $M_x' = K_x \oplus x \Delta$ and comparing $M_x'$ to the
$M_x$ which was sent along with $x$. These equality checks can be
deferred and amortized. Initially the MAC holder, e.g. \Ali, sets $N =
0^\kappa$ and the key holder, e.g. \Bob, sets $N' = 0^\kappa$. As long
as no \textbf{Output} command is executed, when \Ali reveals $x$ she
updates $N \leftarrow G(N, H(M_x))$ for the MAC $M_x$ she should have
sent along with $x$, and $\Bob$ updates $N' \leftarrow
G(N',H(M_x'))$. Before executing an \textbf{Output}, $\Ali$ sends $N$
to \Bob who aborts if $N \ne N'$. Security of this check is easily
proved in the random oracle model. The optimization brings the
communication complexity of the protocol down from $O(\kappa |C|)$ to
$O(|C|+ o \kappa)$, where $o$ is the number of rounds in which outputs
are opened. For a circuit of depth $O(\vert C \vert / \kappa)$, the
communication is $O(|C|)$.

\Subsubsection{Implementing \DEAL.}
In the following sections we show how to implement \DEAL.  In
\secref{abit} we implement just the part with the commands
\textbf{Authenticated Bits}. In \secref{aot} we show how to extend
with the \textbf{Authenticated OT} commands, by showing how to
implement many \AOTs from many \ABITs. In \secref{aand} we then show
how to extend with the \textbf{Authenticated local AND} commands, by
showing how to implement many \AANDs from many \ABITs. We describe the
extensions separately, but since they both maintain the value of the
global keys, they will produce \AANDs and \AOTs with the same keys as
the \ABITs used, giving an implementation of \DEAL.

\Section{Bit Authentication}\seclab{abit}
\begin{wrapfigure}[7]{r}{4cm}
\vspace{-2.9cm}
\begin{center}
\begin{tikzpicture}[scale=0.60]
% Draw the nodes
\node (abit) at (0, 0) {\ABIT};
\node (wabit) at (0, -2) {\WABIT};
\node (labit) at (0, -4) {\LABIT};
\node[anchor=west] (ot) at (0.5, -5.5) {\OT};
\node[anchor=east] (eq) at (-0.5, -5.5) {\EQ};

% Draw the arrows
\draw[->,>=stealth'] (wabit) -- (abit);
\draw[->,>=stealth'] (labit) -- (wabit);
\draw[->,>=stealth'] (ot) -- (labit);
\draw[->,>=stealth'] (eq) -- (labit);
\end{tikzpicture}
\vspace{-0.3cm}
\caption{\secref{abit} outline.}
\figlab{abitoutline}
\end{center}
\end{wrapfigure}
In this section we show how to efficiently implement (oblivious) bit
authentication, i.e., we want to be in a situation where \Ali knows
some bits $x_1,\ldots,x_{\ell}$ together with MACs $M_1,\ldots,
M_{\ell}$, while \Bob holds a global key $\Delta_A$ and local keys
$K_1,\ldots,K_{\ell}$ s.t.~$M_i=K_i\oplus x_i \Delta_A$, as described
in \DEAL (\figref{FDEAL}). Given the complete symmetry of
\DEAL, we only describe the case where \Ali is MAC holder.

If the parties were honest, we could do the following: $\Ali$ and $\Bob$ run
an OT where \Bob inputs the two messages $(K_i,K_i\oplus \Delta_A)$
and \Ali chooses $x_i$, to receive $M_i=K_i\oplus
x_i\Delta_A$. However, if \Bob is dishonest he might not use the same
$\Delta_A$ in all OTs. The main ideas that make the protocol secure
against cheating parties are the following:
\begin{enumerate}
\item 

  For reasons that will be apparent later, we will actually start in
  the opposite direction and let \Bob receive some authenticated bits
  $y_i$ using an OT, where \Ali is supposed to always use the same
  global key $\Gamma_B$. Thus an honest \Ali inputs $(L_i,L_i\oplus
  \Gamma_B)$ in the OTs and \Bob receives $N_i=L_i\oplus y_i
  \Gamma_B$. To check that \Ali is playing honest in most OTs, the
  authenticated bits are randomly paired and a check is performed,
  which restricts \Ali to cheat in at most a few OTs.

\item 

  We then notice that what \Ali gains by using different $\Gamma_B$'s
  in a few OTs is no more than learning a few of \Bob's bits $y_i$. We
  call this a leaky \ABIT, or \LABIT.

\item 

  We show how to turn this situation into an equivalent one where \Ali
  (not \Bob) receives authenticated random bits $x_i$'s (none of which
  leaks to \Bob) under a ``slightly insecure'' global key
  $\Gamma_A$. The insecurity comes from the fact that the leakage of
  the $y_i$'s turns into the leakage of a few bits of the global key
  $\Gamma_A$ towards \Ali. We call this an \ABIT with weak global key,
  or \WABIT. 
  %The step is done using OT extension,
%  meaning that we can go from a few authenticated bits
%  for \Bob to a very large number of authenticated bits for \Ali. Thus
%  this step is particularly important for the efficiency of our
%  construction.

\item 

  Using privacy amplification, we amplify the previous setting to a new
  one where \Ali receives authenticated bits under a (shorter) fully
  secure global key $\Delta_A$, where no bits of $\Delta_A$ are known
  to \Ali, finally implementing the \ABIT command of the dealer box.

\end{enumerate}
\longOnly{We will proceed in reverse order and start with step 4 in
the previous description: we will start with showing how we can turn
authenticated bits under an ``insecure'' global key $\Gamma_A$ into
authenticated bits under a ``secure'' (but shorter) global key
$\Delta_A$.}
\shortOnly{We will proceed in reverse order and start with step 4 in
the previous description.}

\Subsection{Bit Authentication with Weak Global Key
  (\texorpdfstring{\WABIT}{WaBit})}\seclab{wabit}

We will first define the box providing bit authentication, but where
some of the bits of the global key might leak. We call this box \WABIT
(bit authentication with weak global key) and we formally describe it
in~\figref{WABIT}.  The box $\WABIT^{\LL}(\ell,\tau)$ outputs $\ell$
bits with keys of length $\tau$. The box is also parametrized by a
class $\LL$ of leakage functions on $\tau$ bits. The box $\ABIT(\ell,
\psi)$ is the box $\WABIT^\LL(\ell,\psi)$ where $\LL$ is the class of
leakage functions that never leak.

\begin{boxfig}{The box $\WABIT^\LL(\ell,\tau)$ for Bit Authentication
    with Weak Global Key}{WABIT}
\begin{description}
\item[Honest Parties:]  \
\begin{enumerate}
\item 

  The box samples $\Gamma_A \inR \zo^\tau$ and outputs it to $\Bob$.

\item

  The box samples and outputs $\repa{x_1},\ldots,\repa{x_{\ell}}$.
  Each $\repa{x_i}=(x_i,M'_{i},K'_{i}) \in \zo^{1 + 2\tau}$
  s.t.~$M'_i=K'_i\oplus x_i \Gamma_A$.

\end{enumerate}
\item[Corrupted Parties:] \
\begin{enumerate}
\item

  If \Ali is corrupted, then \Ali may choose a leakage function $L \in
  \LL$. Then the box samples $(S,c) \leftarrow L$. If $c=0$ the box
  outputs $\texttt{fail}$ to \Bob and terminates. If $c=1$, the box
  outputs $\{ (i, (\Gamma_A)_{i}) \}_{i \in S}$ to \Ali.

\item 

  If \Ali is corrupted, then \Ali chooses the $x_i$ and the $M_i'$ and
  then $K_i' = M_i' \oplus x_i \Gamma_A$.

\item 

  If \Bob is corrupted, then \Bob chooses $\Gamma_A$ and the $K_i'$.

\end{enumerate}
\item[Global Key Queries:] The adversary can input $\Gamma$ and will
  be told if $\Gamma = \Gamma_A$.
\end{description}
\end{boxfig}
\vspace{-0.7cm}
\begin{boxfig}{Subprotocol for reducing $\ABIT(\ell,\psi)$ to
    $\WABIT^\LL(\ell,\tau)$.}{SUBPROTOA}
\begin{enumerate}
\item

  The parties invoke $\WABIT^\LL(\ell,\tau)$ with $\tau = \frac{22}3
  \psi$. The output to \Ali is $((M'_1, x_1), \ldots,
  (M'_\ell,x_\ell))$. The output to \Bob is $(\Gamma_A, K'_1, \ldots,
  K'_\ell)$.

\item

  \Bob samples $\B{A} \inR \zo^{\psi \times \tau}$, a random binary
  matrix with $\psi$ rows and $\tau$ columns, and sends $\B{A}$ to
  \Ali.

\item

  \Ali computes $M_i = \B{A} M'_i \in \zo^{\psi}$ and outputs $((M_1,
  x_1), \ldots, (M_\ell,x_\ell))$.

\item

  \Bob computes $\Delta_A = \B{A} \Gamma_A$ and $K_i = \B{A} K'_i$ and
  outputs $(\Delta_A, K_1, \ldots, K_\ell)$.

\end{enumerate}
\end{boxfig}
\vspace{-0.7cm}

In~\figref{SUBPROTOA} we describe a protocol which takes a box \WABIT, where one
quarter of the bits of the global key might leak, and amplifies it to
a box \ABIT where the global key is perfectly secret.
The protocol is described for general
$\LL$ and it is parametrized by a desired security level $\psi$. The
proof of the following theorem can be found in
\appref{proofofTHMWABIT}.

\begin{theorem}\thmlab{THMWABIT}
  Let $\tau = \frac{22}{3} \psi$ and $\LL$ be a $\left(\frac34
    \tau\right)$-secure leakage function on $\tau$ bits. The protocol
  in \figref{SUBPROTOA} securely implements $\ABIT(\ell,\psi)$ in the
  $\WABIT^\LL(\ell,\tau)$-hybrid model with security parameter
  $\psi$. The communication is $O(\psi^2)$ and the work is $O(\psi^2
  \ell)$.
\end{theorem}

\Subsection{Bit Authentication with Leaking Bits
  (\texorpdfstring{\LABIT}{LaBit})}\seclab{labit}

We now show another insecure box for \ABIT. The new box is insecure in
the sense that a few of the bits to be authenticated might leak to the
other party. We call this box an \ABIT with leaking bits, or $\LABIT$
and formally describe it in \figref{LABIT}. The box
$\LABIT^\LL(\tau,\ell)$ outputs $\tau$ authenticated bits with keys of
length $\ell$, and is parametrized by a class of leakage functions
$\LL$ on $\tau$-bits. We show that $\WABIT^\LL$ can be reduced to
$\LABIT^\LL$. In the reduction, a \LABIT that outputs authenticated
bits $\repb{y_i}$ to \Bob can be turned into a \WABIT that outputs
authenticated bits $\repa{x_j}$ to \Ali, therefore we present the
\LABIT box that outputs bits to \Bob. The reduction is strongly
inspired by the OT extension techniques
in~\cite{DBLP:conf/crypto/IshaiKNP03}.

\begin{boxfig}{The box $\LABIT^\LL(\tau,\ell)$ 
for Bit Authentication with Leaking Bits}{LABIT}
\begin{description}
\item[Honest Parties:] \
  \begin{enumerate}
  \item 

    The box samples $\Gamma_B \inR \zo^\ell$ and outputs it to $\Ali$.

  \item

    The box samples and outputs $\repb{y_1},\ldots,\repb{y_{\tau}}$.
    Each $\repb{y_i}=(y_i,N_{i},L_{i}) \in \zo^{1+2\ell}$
    s.t.~$N_i=L_i\oplus y_i \Gamma_B$.

  \end{enumerate}

\item[Corrupted Parties:] \
  \begin{enumerate}
  \item

    If \Ali is corrupted, then \Ali may input a leakage function $L
    \in \LL$. Then the box samples $(S,c) \leftarrow L$. If $c=0$ the
    box outputs $\texttt{fail}$ to \Bob and terminates. If $c=1$, the
    box outputs $\{ (i, y_i) \}_{i \in S}$ to \Ali.

  \item

    Corrupted parties get to specify their outputs as in \figref{WABIT}.

  \end{enumerate}

\item[Choice Bit Queries:] The adversary can input $\Delta$ and will
  be told if $\Delta = (y_1, \ldots, y_\tau)$.
\end{description}
\end{boxfig}
%\vspace{-0.3cm}

\begin{boxfig}{Subprotocol for reducing $\WABIT^\LL(\ell,\tau)$ to 
    $\LABIT^\LL(\tau,\ell)$}{WABITLABIT}
  \begin{enumerate}
  \item

    \Ali and \Bob invoke $\LABIT^{\LL}(\tau, \ell)$. \Bob learns
    $((N_1, y_1), \ldots, (N_\tau,y_\tau))$ and \Ali learns
    $(\Gamma_B, L_1, \ldots, L_\tau)$.

  \item  
    
    $\Ali$ lets $x_j$ be the $j$-th bit of $\Gamma_B$ and $M_j$ the
    string consisting of the $j$-th bits from all the strings $L_i$,
    i.e.~$M_j~=~L_{1,j}||L_{2,j}||\ldots||L_{\ell,j}$.
   
  \item

    \Bob lets $\Gamma_A$ be the string consisting of all the bits
    $y_i$, i.e.~$\Gamma_A=y_1||y_2||\ldots||y_\ell$, and lets $K_j$
    be the string consisting of the $j$-th bits from all the strings
    $N_i$, i.e.~$K_j= N_{1,j}||N_{2,j}||\ldots||N_{\ell,j}$.
    
  \item 

    \Ali and \Bob now hold $\repa{x_j} = (x_j, M_j, K_j)$ for $j = 1,
    \ldots, \ell$.

\end{enumerate}
\end{boxfig}
\vspace{-0.3cm}

\begin{theorem}
  For all $\ell$, $\tau$ and $\LL$ the boxes $\WABIT^\LL(\ell,\tau)$
  and $\LABIT^\LL(\tau,\ell)$ are linear locally equivalent, i.e., can
  be implemented given the other in linear time without interaction.
\end{theorem}
\begin{proof}
  The first direction (reducing \WABIT to \LABIT) is shown
  in~\figref{WABITLABIT}. The other direction (\LABIT is linear
  locally reducible to \WABIT) will follow by the fact that the local
  transformations are reversible in linear time.  One can check that
  for all $j=1, \ldots,\tau$, $\repa{x_j}$ is a correct authenticated
  bit. Namely, from the box \LABIT we get that for all $i=1, \ldots,
  \ell$, $N_i = L_i \oplus y_i \Gamma_B$. In particular the $j$-th bit
  satisfies $N_{i,j} = L_{i,j} \oplus y_i (\Gamma_B)_j$, which can be
  rewritten (using the same renaming as in the protocol) as $K_{j,i} =
  M_{j,i} \oplus (\Gamma_A)_i x_j$, and therefore $M_j=K_j\oplus x_j
  \Gamma_A$, as we want.  It is easy so see (as the protocol only
  consists of renamings) that leakage on the choice bits is equivalent
  to leakage on the global key under this transformation, and guesses
  on $\Gamma_A$ are equivalent to guesses on $(y_1, \ldots, y_\tau)$,
  so giving a simulation argument is straight-forward when $\LL$ is
  the same for both boxes.\putbox
\end{proof}

Note that since we turn $\LABIT^\LL(\ell, \tau)$ into $\WABIT^\LL(\tau,
\ell)$, if we choose $\ell = \poly(\psi)$ we can turn a relatively
small number ($\tau = \frac{22}{3}\psi$) of authenticated bits towards one
player into a very larger number ($\ell$) of authenticated bits
towards the other player.

\Subsection{A Protocol For Bit Authentication With Leaking Bits}

In this section we show how to construct authenticated bits starting
from OTs. The protocol ensures that most of the authenticated bits
will be kept secret, as specified by the \LABIT box in \figref{LABIT}.

The main idea of the protocol, described in \figref{LABITPROT}, is the
following: many authenticated bits $\repb{y_i}$ for \Bob are created
using OTs, where \Ali is supposed to input messages
$(L_i,L_i\oplus\Gamma_B)$. To check that \Ali is using the same
$\Gamma_B$ in every OT, the authenticated bits are randomly
paired. Given a pair of authenticated bits $\repb{y_i},\repb{y_j}$,
\Ali and \Bob compute $\repb{z_i} = \repb{y_i} \oplus \repb{y_j}
\oplus d_i$ where $d_i = y_i\oplus y_j$ is announced by \Bob.  If \Ali
behaved honestly, she knows the MAC that $\Bob$ holds on $z_i$,
otherwise she has $1$ bit of entropy on this MAC, as shown below.  The
parties can check if \Ali knows the MAC using the \EQ box described
in App.~\ref{sec:uc}. As \Bob reveals $y_i \oplus y_j$, they waste
$\repb{y_j}$ and only use $\repb{y_i}$ as output from the
protocol---as $y_j$ is uniformly random $y_i \oplus y_j$ leaks no
information on $y_i$. Note that we cannot simply let \Ali reveal the
MAC on $z_i$, as a malicious \Bob could announce $1 \oplus z_i$: this
would allow \Bob to learn a MAC on $z_i$ and $1\oplus z_i$ at the same
time, thus leaking $\Gamma_B$. Using \EQ forces a thus cheating \Bob
to guess the MAC on a bit which he did not see, which he can do only
with negligible probability $2^{-\ell}$.

\newcommand{\Tau}{{\mathcal{T}}}

\begin{boxfig}{The protocol for reducing $\LABIT(\tau,\ell)$ to 
$\OT(2\tau,\ell)$ and $\EQ(\tau\ell)$.}{LABITPROT} 
\begin{enumerate}
\item
  
  \Ali samples $\Gamma_B \inR \zo^\ell$ and for $i = 1, \ldots, \Tau$
  samples $L_i \inR \zo^\ell$, where $\Tau = 2\tau$.

\item

  \Bob samples $(y_1, \ldots, y_\Tau) \inR \zo^{\Tau}$.

\item

  They run $\Tau$ OTs, where for $i = 1, \ldots, \Tau$ party \Ali
  offers $(Y_{i,0},Y_{i,1}) = (L_i, L_i \oplus \Gamma_B)$ and \Bob
  selects $y_i$ and receives $N_i=Y_{i,y_i} = L_i \oplus
  y_{i}\Gamma_B$. Let $\repb{y_1},\ldots,\repb{y_\Tau}$ be the
  candidate authenticated bits produced so far.

\item

  \Bob picks a uniformly random pairing $\pi$ (a permutation $\pi:
  \{1, \ldots, \Tau \} \rightarrow \{1, \ldots, \Tau \}$ where
  $\forall i, \pi(\pi(i))=i$), and sends $\pi$ to \Ali.  Given a
  pairing $\pi$, let $\SS(\pi) = \{ i \vert i \le \pi(i) \}$, i.e., for
  each pair, add the smallest index to $\SS(\pi)$.

\item

  For all $\tau$ indices $i \in \SS(\pi)$:
  \begin{enumerate}
  \item 

    \Bob announces $d_i = y_{i} \oplus y_{\pi(i)}$.
  
  \item 

    \Ali and \Bob compute $\repb{z_i}=\repb{y_{i}} \oplus
    \repb{y_{\pi(i)}} \oplus d_i$.

  \item 

    Let $Z_i$ and $W_i$ be the MAC and the local key for $z_i$ held by
    \Ali respectively \Bob. They compare these using \EQ and abort if
    they are different.

  \end{enumerate}
  The $\tau$ comparisons are done using \emph{one} call on the $\tau\ell$-bit
  strings $( Z_i )_{i \in \SS(\pi)}$ and $( W_i )_{i \in \SS(\pi)}$.

\item

  For all $i \in \SS(\pi)$ \Ali and \Bob output $\repb{y_{i}}$.

\end{enumerate}

\end{boxfig}

Note that if \Ali uses different $\Gamma_B$ in two paired instances,
$\Gamma_i$ and $\Gamma_j$ say, then the MAC held by \Bob on $y_{i}
\oplus y_{j}$ (and therefore also $z_i$) is $(L_i \oplus y_i \Gamma_i)
\oplus (L_j \oplus y_j \Gamma_j) = (L_i \oplus L_j) \oplus (y_i \oplus
y_j) \Gamma_j \oplus y_i (\Gamma_i \oplus \Gamma_j)$.  Since
$(\Gamma_i \oplus \Gamma_j) \ne 0^\ell$ and $y_i \oplus y_j$ is fixed
by announcing $d_i$, guessing this MAC is equivalent to guessing
$y_i$.  As \Ali only knows $L_i, L_j, \Gamma_i, \Gamma_j$ and $y_i
\oplus y_j$, she cannot guess $y_i$ with probability better than
$1/2$.  Therefore, if \Ali cheats in many OTs, she will get caught
with high probability. If she only cheats on a few instances she might
pass the test. Doing so confirms her guess on $y_i$ in the pairs where
she cheated. Now assume that she cheated in instance $i$ and offered
$(L_i,L_i\oplus\Gamma_B')$ instead of $(L_i,L_i\oplus\Gamma_B)$. After
getting her guess on $y_i$ confirmed she can explain the run as an
honest run: If $y_i = 0$, the run is equivalent to having offered
$(L_i,L_i\oplus\Gamma_B)$, as \Bob gets no information on the second
message when $y_i = 0$. If $y_i = 1$, then the run is equivalent to
having offered $(L_i',L_i'\oplus\Gamma_B)$ with $L_i' = L_i \oplus
(\Gamma_B \oplus \Gamma_B')$, as $L_i' \oplus \Gamma_B = L_i \oplus
\Gamma_B$ and \Bob gets no information on the first message when $y_i
= 1$.  So, any cheating strategy of \Ali can be simulated by letting
her honestly use the same $\Gamma_B$ in all pairs and then let her try
to guess some bits $y_i$. If she guesses wrong, the deviation is
reported to \Bob. If she guesses right, she is told so and the
deviation is not reported to \Bob. This, in turn, can be captured
using some appropriate class of leakage functions $\LL$. Nailing down
the exact $\LL$ needed to simulate a given behavior of \Ali, including
defining what is the ``right'' $\Gamma_B$, and showing that the needed
$\LL$ is always $\kappa$-secure is a relatively straight-forward but
very tedious business. The proof of the following theorem can be found
in \appref{proofofbitauth}.

\begin{theorem} \thmlab{bitauth} Let $\kappa=\frac34 \tau$, and let
  $\LL$ be a $\kappa$ secure leakage function on $\tau$ bits.  The
  protocol in \figref{LABITPROT} securely implements
  $\LABIT^{\LL}(\tau,\ell)$ in the $(\OT(2\tau,\ell),\EQ(\tau
  \ell))$-hybrid model. The communication is $O(\tau^2)$. The work is
  $O(\tau \ell)$.
\end{theorem}

\begin{corollary} \label{cor:total}
  Let $\psi$ denote the security parameter and let $\ell =
  \poly(\psi)$. The box $\ABIT(\ell,\psi)$ can be reduced to
  $(\OT(\frac{44}{3} \psi, \psi),\EQ(\psi))$. The communication is
  $O(\psi\ell + \psi^2)$ and the work is $O(\psi^2 \ell)$.
\end{corollary}
\begin{proof}
  Combining the above theorems we have that $\ABIT(\ell,\psi)$ can be
  reduced to $(\OT(\frac{44}{3} \psi,\ell),\EQ(\frac{22}3 \psi \ell))$
  with communication $O(\psi^2)$ and work $O(\psi^2 \ell)$. For any
  polynomial $\ell$, we can implement $\OT(\frac{44}{3}\psi, \ell)$
  given $\OT(\frac{44}{3}\psi,\psi)$ and a pseudo-random generator
  $\prg: \zo^\psi \rightarrow \zo^\ell$. Namely, seeds are sent using
  the OTs and the $\prg$ is used to one-time pad encrypt the
  messages. The communication is $2\ell$. If we use the RO to
  implement the pseudo-random generator and count the hashing of
  $\kappa$ bits as $O(\kappa)$ work, then the work is
  $O(\ell\psi)$. We can implement $\EQ(\frac{22}3\psi\ell)$ by
  comparing short hashes produced using the RO. 
  The work is $O(\psi\ell)$. \putbox
\end{proof}

Since the oracles $(\OT(\frac{44}{3} \psi, \psi),\EQ(\psi))$ are
independent of $\ell$, the cost of essentially any reasonable
implementation of them can be amortized away by picking $\ell$ large
enough. See \appref{complexity} for a more detailed complexity analysis.

\Subsubsection{Efficient OT Extension:} 

We notice that the \WABIT box resembles an intermediate step of the OT
extension protocol of \cite{DBLP:conf/crypto/IshaiKNP03}. Completing
their protocol (i.e., ``hashing away'' the fact that all messages
pairs have the same XOR), gives an efficient protocol for OT
extension, with the same asymptotic complexity
as~\cite{DBLP:conf/tcc/HarnikIKN08}, but with dramatically smaller
constants. See~\appref{otext} for details.

\Section{Authenticated Oblivious Transfer}\seclab{aot}
In this section we show how to implement \AOTs. We implemented
\ABIT{}s in \secref{abit}, so what remains is to show how to implement
\AOT{}s from \ABIT{}s i.e., to implement the \DEAL box when it outputs
$\repa{x_0},\repa{x_1},\repb{c},\repb{z}$ with $z=c(x_0\oplus
x_1)\oplus x_0 = x_{c}$. Because of symmetry we only show the
construction of \AOT{}s from \ABIT{}s with \Ali as sender and \Bob as
receiver.

\begin{wrapfigure}[7]{r}{4cm}
\vspace{-1.5cm}
\begin{center}
\begin{tikzpicture}[scale=0.60]
% Draw the nodes
\node (aot) at (0, -0) {\AOT};
\node (laot) at (0, -2) {\LAOT};
\node[anchor=east] (abit) at (-0.5, -4) {\ABIT};
\node[anchor=west] (eq) at (0.5, -4) {\EQ};

% Draw the arrows
\draw[->,>=stealth'] (abit) -- (laot);
\draw[->,>=stealth'] (eq) -- (laot);
\draw[->,>=stealth'] (laot) -- (aot);

\end{tikzpicture}
\vspace{-0.4cm}
\caption{\secref{aot} outline.}
\figlab{aotoutline}
\end{center}
\end{wrapfigure}

\begin{boxfig}{The Leaky Authenticated OT box $\LAOT(\ell)$ }{BOXLAOT}
  \begin{description}
  \item[Honest Parties:]

    For $i = 1, \ldots, \ell$, the box outputs random
    $\repa{x_0^i},\repa{x_1^i},\repb{c^i},\repb{z^i}$ with
    $z^i=c^i(x_0^i\oplus x_1^i)\oplus x_0^i$.

  \item[Corrupted Parties:] \
    \begin{enumerate}
    \item If \Bob is corrupted he gets to choose all his random values.
    \item If \Ali is corrupted she gets to choose all her random
      values. Also, she may, at any point before \Bob received his
      outputs, input $(i,g_i)$ to the box in order to try to guess
      $c_i$. If $c_i \ne g_i$ the box will output \texttt{fail} and
      terminate. Otherwise the box proceeds as if nothing has happened
      and \Ali will know the guess was correct. She may input as many
      guesses as she desires.
    \end{enumerate}

  \item[Global Key Queries:]

    The adversary can at any point input $(\Ali, \Delta)$ and will be
    returned whether $\Delta = \Delta_B$. And it can at any point input
    $(\Bob, \Delta)$ and will be returned whether $\Delta = \Delta_A$.
    
  \end{description}
\end{boxfig}

We go via a leaky version of authenticated OT, or \LAOT, described in
\figref{BOXLAOT}. The \LAOT box is leaky in the sense that choice bits
may leak when \Ali is corrupted: a corrupted \Ali is allowed to make
guesses on choice bits, but if the guess is wrong the box aborts
revealing that \Ali is cheating. This means that if the box does not
abort, with very high probability \Ali only tried to guess a few
choice bits.

The protocol to construct a leaky \AOT (described in
\figref{PROTLAOT}) proceeds as follows: First \Ali and \Bob get
$\repa{x_0},\repa{x_1}$ (\Ali's messages), $\repb{c}$ (\Bob's choice
bit) and $\repb{r}$. Then \Ali transfers the message $z=x_{c}$ to \Bob
in the following way: \Bob knows the MAC for his choice bit $M_c$,
while \Ali knows the two keys $K_c$ and $\Delta_B$. This allows \Ali
to compute the two possible MACs $(K_c, K_c\oplus \Delta_B)$
respectively for the case of $c=0$ and $c=1$. Hashing these values
leaves \Ali with two uncorrelated strings $H(K_c)$ and $H(K_c\oplus
\Delta_B)$, one of which \Bob can compute as $H(M_c)$. These values
can be used as a one-time pad for \Ali's bits $x_0,x_1$ (and some
other values as described later), and \Bob can retrieve $x_c$ and
announce the difference $d=x_c\oplus r$ and therefore compute the
output $\repb{z}=\repb{r}\oplus d$.

\begin{boxfig}{The protocol for authenticated OT with leaky choice
      bit}{PROTLAOT}
    %The protocol runs $\ell$ times in parallel, in the hybrid model
    %with access to $\ABIT^{\Ali}(2\ell)$ with $\Ali$ as MAC holder and
    %$\ABIT^{\Bob}(2\ell)$ with $\Bob$ as MAC holder and
    %$\EQ(\ell2\kappa)$.

    The protocol runs $\ell$ times in parallel, here described for a
    single leaky authenticated \OT.
\begin{enumerate}

\item 

  \Ali and \Bob get $\repa{x_0},\repa{x_1},\repb{c},\repb{r}$ from 
  the dealer. 

\item 

  Let $\repa{x_0}=(x_0,M_{x_0},K_{x_0}),\repa{x_1} =
  (x_1,M_{x_1},K_{x_1}),\repb{c}=(c,M_c,K_c),\repb{r}=(r,M_r,K_r)$.

\item 
  
  \Ali chooses random strings $T_0,T_1 \in \zo^{\kappa}$.

\item 
  
  \Ali sends $(X_0,X_1)$ to \Bob where $X_0 = H(K_c) \oplus
  (x_0||M_{x_0}||T_{x_0})$ and $X_1 = H(K_c\oplus \Delta_B) \oplus
  (x_1||M_{x_1}||T_{x_1})$.

\item 
  
  \Bob computes $(x_c||M_{x_c}||T_{x_c})= X_c \oplus H(M_c)$. \Bob
  aborts if $M_{x_c} \neq K_{x_c}\oplus x_c \Delta_A$. Otherwise, let
  $z=x_c$.

\item 

  \Bob announces $d=z\oplus r$ to \Ali and the parties compute
  $\repb{z}=\repb{r}\oplus d$. Let $\repb{z}=(z,M_z,K_z)$.

\item 

  \Ali sends $(I_0,I_1)$ to \Bob where $I_0 = H(K_z) \oplus T_1$ and
  $I_1 = H(K_z\oplus \Delta_B) \oplus T_0$.

\item 

  \Bob computes $T_{1\oplus z} = I_z \oplus H(M_z)$. \emph{Notice
    that now \Bob has both $(T_0,T_1)$}.

\item 

  \Ali and \Bob both input $(T_0,T_1)$ to $\EQ$.
  The comparisons are done using one call to $\EQ(\ell2\kappa)$.

\item 

  If the values are the same, they output
  $\repa{x_0},\repa{x_1},\repb{c},\repb{z}$.

\end{enumerate}
\end{boxfig}

In order to check if \Ali is transmitting the correct bits $x_0,x_1$,
she will transfer the respective MACs together with the bits: as \Bob
is supposed to learn $x_c$, revealing the MAC on this bit does not
introduce any insecurity. However, \Ali can now mount a selective
failure attack: \Ali can check if \Bob's choice bit $c$ is equal to,
e.g., $0$ by sending $x_0$ with the right MAC and $x_1$ together with
a random string. Now if $c=0$ \Bob only sees the valid MAC and
continues the protocol, while if $c=1$ \Bob aborts because of the
wrong MAC. A similar attack can be mounted to check if $c=1$. We will
fix this later by randomly partitioning and combining a few \LAOTs
together.

On the other hand, if \Bob is corrupted, he could be announcing the
wrong value $d$. In particular, \Ali needs to check that the
authenticated bit $\repb{z}$ is equal to $x_c$ without learning
$c$. In order to do this, we have \Ali choosing two random strings
$T_0,T_1$, and append them, respectively, to $x_0,x_1$ and the MACs on
those bits, so that \Bob learns $T_c$ together with $x_c$. After
\Bob announces $d$, we can again use the MAC and the keys for $z$ to
perform a new transfer: \Ali uses $H(K_z)$ as a one-time pad for $T_1$
and $H(K_z\oplus \Delta_B)$ as a one-time pad for $T_0$. Using $M_z$,
the MAC on $z$, \Bob can retrieve $T_{1\oplus z}$. This means that an
honest \Bob, that sets $z=x_c$, will know both $T_0$ and $T_1$, while
a dishonest \Bob will not be able to know both values except with
negligible probability. Using the \EQ~box \Ali can check that
\Bob knows both values $T_0,T_1$. Note that we cannot simply have \Bob
openly announce these values, as this would open the possibility for
new attacks on \Ali's side. The proof of the following theorem can be
found in \appref{proofoflaot}.
\begin{theorem} \thmlab{laot} 
  The protocol in \figref{PROTLAOT} securely implements $\LAOT(\ell)$
  in the  $(\ABIT(4\ell,\kappa),\EQ(2\ell\kappa))$-hybrid model.
\end{theorem}

To deal with the leakage of the \LAOT box, we let \Bob randomly
partition the \LAOTs in small buckets: all the \LAOTs in a bucket will
be combined using an OT combiner (as shown in~\figref{PROTAOT}), in
such a way that if at least one choice bit in every bucket is unknown
to \Ali, then the resulting \AOT will not be leaky.  The overall
protocol is secure because of the OT combiner and the probability that
any bucket is filled only with OTs where the choice bit leaked is
negligible, as shown in~\appref{proofthmaot}.

\begin{boxfig}{From Leaky Authenticated OTs to Authenticated
    OTs}{PROTAOT}
\begin{enumerate}
\item

  \Ali and \Bob generate $\ell'=B\ell$ authenticated OTs using
  $\LAOT(\ell')$. If the box does not abort, name the outputs $\{
  \repa{x^i_{0}},\repa{x^i_{1}},\repb{c^i},\repb{z^i}\}^{\ell'}_{i=1}$.

\item 

  \Bob sends a $B$-wise independent permutation $\pi$ on
  $\{1,\ldots,\ell'\}$ to \Ali.  For $j = 0, \ldots, \ell - 1$, the
  $B$ quadruples $\{ \repa{x^{\pi(i)}_{0}},\repa{x^{\pi(i)}_{1}},
  \repb{c^{\pi(i)}},\repb{z^{\pi(i)}}\}^{j B + B}_{i=j B + 1}$ are
  defined to be in the $j$'th bucket.

\item

  We describe how to combine two OTs from a bucket, call them
  $\repa{x^1_0},\repa{x^1_1},\repb{c^1},\repb{z^1}$ and
  $\repa{x^2_0},\repa{x^2_1},\repb{c^2},\repb{z^2}$. Call the result
  $\repa{x_0},\repa{x_1},\repb{c},\repb{z}$.
  To combine more than two, just iterate by taking the result and
  combine it with the next leaky OT.
  \begin{enumerate}
  \item 

    \Ali reveals $d = x^1_0 \oplus x^1_1 \oplus x^2_0 \oplus x^2_1$.

  \item 

    Compute: $\repb{c}=\repb{c^1}\oplus \repb{c^2}$,
    $\repb{z}=\repb{z^1}\oplus\repb{z^2}\oplus d\repb{c^1}$,
    $\repa{x_0}=\repa{x^1_0}\oplus \repa{x^2_0}$,
    $\repa{x_1}=\repa{x^1_0}\oplus \repa{x^2_1}$.

  \end{enumerate}

\end{enumerate}
\end{boxfig}

\begin{theorem} 
  \thmlab{thmAOT} Let $\AOT(\ell)$ denote the box which outputs $\ell$
  \AOT{}s as in \DEAL. If $(\log_2(\ell) + 1)(B-1) \ge \psi$, then the
  protocol in~\figref{PROTAOT} securely implements $\AOT(\ell)$ in the
  $\LAOT(B \ell)$-hybrid model with security parameter $\psi$.

\end{theorem}

\Section{Authenticated local AND}\seclab{aand}

\begin{wrapfigure}[6]{r}{4cm}
\vspace{-2.5cm}
\begin{center}
\begin{tikzpicture}[scale=0.60]
% Draw the nodes
\node (aand) at (0, 0) {\AAND};
\node (laand) at (0, -2) {\LAAND};
\node[anchor=east] (abit) at (-0.5, -4) {\ABIT};
\node[anchor=west] (eq) at (0.5, -4) {\EQ};

% Draw the arrows
\draw[->,>=stealth'] (abit) -- (laand);
\draw[->,>=stealth'] (eq) -- (laand);
\draw[->,>=stealth'] (laand) -- (aand);

\end{tikzpicture}
\vspace{-0.3cm}
\caption{\secref{aand} outline.}
\figlab{aandoutline}
\end{center}
\end{wrapfigure}

In this section we show how to generate \AAND, i.e., how to implement
the dealer box when it outputs $\repa{x},\repa{y},\repa{z}$ with
$z=xy$. As usual, as \AAND for \Bob is symmetric, we only present how
to construct \AAND for \Ali.

We first construct a leaky version of \AAND, or \LAAND, described in
\figref{BOXLAAND}. Similar to the \LAOT box the \LAAND box may leak
the value $x$ to \Bob, at the price for \Bob of being detected. The
intuition behind the protocol for \LAAND, described in \figref{LAAND},
is to let \Ali compute the AND locally and then authenticate the
result. \Ali and \Bob then perform some computation on the keys and
MACs, in a way so that \Ali will be able to guess \Bob's result only
if she behaved honestly during the protocol: \Ali behaved honestly
(sent $d = z \oplus r$) iff she knows $W_0 = (K_x||K_z)$ or $W_1 =
(K_x \oplus \Delta_A||K_y \oplus K_z)$. In fact, she knows $W_x$. As
an example, if $x = 0$ and \Ali is honest, then $z = 0$, so she knows
$M_x = K_x$ and $M_z = K_z$. Had she cheated, she would know $M_z =
K_z \oplus \Delta_A$ instead of $K_z$. \Bob checks that \Ali knows
$W_0$ or $W_1$ by sending her $H(W_0) \oplus H(W_1)$ and ask her to
return $H(W_0)$. This, however, allows \Bob to send $H(W_0) \oplus
H(W_1) \oplus E$ for an error term $E \ne 0^\kappa$. The returned
value would be $H(W_0) \oplus x E$. To prevent this attack, they use
the \EQ box to compare the values instead. If \Bob uses $E \ne
0^\kappa$, he must now guess $x$ to pass the protocol. However, \Bob
still may use this technique to guess a few $x$ bits. We fix this
leakage later in a way similar to the way we fixed leakage of the
\LAOT box in \secref{aot}. The proof of the following theorem can be
found in~\appref{proofoflaand}.

\begin{theorem}\thmlab{laand}
  The protocol in \figref{LAAND} securely implements $\LAAND(\ell)$ in
  the $(\ABIT(3\ell,\kappa),\EQ(\ell\kappa))$-hybrid model.
\end{theorem}

\begin{boxfig}{The box $\LAAND(\ell)$ for 
    $\ell$ Leaky Authenticated local AND.}{BOXLAAND}
\begin{description}
\item[Honest Parties:] For $i = 1, \ldots, \ell$, the box outputs
random $\repa{x_i},\repa{y_i},\repa{z_i}$ with $z_i=x_iy_i$.
\item[Corrupted Parties:] \
\begin{enumerate}
\item If \Ali is corrupted she gets to choose all her random values.
\item If \Bob is corrupted he gets to choose all his random values,
  including the global key $\Delta_A$.
  Also, he may, at any point prior to output being delivered to \Ali,
  input $(i, g_i)$ to the box in order to try to guess $x_i$. If $g_i
  \not= x_i$ the box will output \texttt{fail} to \Ali and
  terminate. Otherwise the box proceeds as if nothing has happened and
  \Bob will know the guess was correct. He may make as many guesses as
  he desires.
\end{enumerate}
\item[Global Key Queries:] The adversary can input $\Delta$ and will
  be told if $\Delta = \Delta_A$.
\end{description}
\end{boxfig}

%Details can be found in \appref{aand}.

%\begin{boxfig}{Protocol for authenticated AND with leaking bit}{LAAND}
%Here described for a single triple:
%\begin{enumerate}
%\item
%
%  \Ali and \Bob ask the dealer for $\repa{x},\repa{y},\repa{r}$. (The
%  global key is $\Delta_A$).
%
%\item
%
%  \Ali computes $z=xy$ and announces $d=z\oplus r$.
%
%\item 
%
%  The parties compute $\repa{z}=\repa{r}\oplus d$.
%
%\item
%  
%  \Bob sends $U = H(K_x,K_z) \oplus H(K_x \oplus \Delta_A,K_y \oplus K_z)$ to
%  \Ali.
%
%\item
%
%  If $x=0$, then \Ali lets $V = H(M_x,M_z)$.  If $x=1$, then \Ali
%  lets $V = U \oplus H(M_x,M_y\oplus M_z)$.
%
%\item
%
%  \Ali and \Bob call the \EQ box, with inputs
%  $V$ and $H(K_x,K_z)$ respectively.\footnote{All the $\ell$ calls to
%  \EQ are handles using a single call to
%  $\EQ(\ell \kappa)$.}
%
%\item
%
%  If the strings were not different, the parties output
%  $\repa{x},\repa{y},\repa{z}$.
%
%\end{enumerate}
%
%\end{boxfig}

\begin{boxfig}{Protocol for authenticated local AND with leaking bit}{LAAND}
The protocol runs $\ell$ times in parallel. Here described for a
single leaky authenticated local AND:
\begin{enumerate}
\item

  \Ali and \Bob ask the dealer for $\repa{x},\repa{y},\repa{r}$. (The
  global key is $\Delta_A$).

\item

  \Ali computes $z=xy$ and announces $d=z\oplus r$.

\item 

  The parties compute $\repa{z}=\repa{r}\oplus d$.

\item
  
  \Bob sends $U = H(K_x||K_z) \oplus H(K_x \oplus \Delta_A||K_y \oplus K_z)$ to
  \Ali.

\item

  If $x=0$, then \Ali lets $V = H(M_x||M_z)$.  If $x=1$, then \Ali
  lets $V = U \oplus H(M_x||M_y\oplus M_z)$.

\item

  \Ali and \Bob call the \EQ box, with inputs
  $V$ and $H(K_x||K_z)$ respectively. All the $\ell$ calls to
  \EQ are handled using a single call to
  $\EQ(\ell \kappa)$.

\item

  If the strings were not different, the parties output
  $\repa{x},\repa{y},\repa{z}$.

\end{enumerate}
\end{boxfig}

%Unfortunately, now \Bob can mount a selective failure attack. If \Bob
%sends $U \oplus E$, where $U$ is the correct value according to the
%protocol and $E \ne 0^{\kappa}$ is an error term, then \Ali inputs
%$H(K_x,K_z)$ to the comparison when $x = 0$ and $H(K_x,K_z) \oplus E$
%when $x = 1$. This allows \Bob to make a guess at \Ali's bit $x$ and
%make sure that the protocol aborts only if the guess was wrong.
%Therefore, if the protocol does not abort, \Bob learns \Ali's
%bit. Thus \Bob can undetected learn $e$ of the bits of \Ali with
%probability $2^{-e}$.  We fix this problem below by generating many
%authenticated ANDs with possibly leaking $x$ (or leaky \AAND from now
%on) and then amplify their security. 

We now handle a few guessed $x$ bits by random bucketing and a
straight-forward combiner. In doing this efficiently, it is central
that the protocol was constructed such that only $x$ could leak. Had
\Bob been able to get information on both $x$ and $y$ we would have
had to do the amplification twice.

\begin{boxfig}{From Leaky Authenticated local ANDs to Authenticated
    local ANDs}{PROTAND}
The protocol is parametrized by positive integers $B$ and $\ell$.
\begin{enumerate}
\item

  \Ali and \Bob call $\LAAND(\ell')$ with $\ell' = B \ell$. If the
  call to \LAAND aborts, this protocol aborts. Otherwise, let
  $\{\repa{x_i},\repa{y_i},\repa{z_i}\}^{\ell'}_{i=1}$ be the outputs.

\item 

  \Ali picks a $B$-wise independent permutation $\pi$ on
  $\{1,\ldots,\ell'\}$ and sends it to \Bob. For $j = 0, \ldots,
  \ell-1$, the $B$ triples
  $\{\repa{x_{\pi(i)}},\repa{y_{\pi(i)}},\repa{z_{\pi(i)}}\}^{j B +
  B}_{i=j B + 1}$ are defined to be in the $j$'th bucket.

\item

  The parties combine the $B$ \LAANDs in the same bucket. We describe
  how to combine two \LAANDs, call them
  $\repa{x^1},\repa{y^1},\repa{z^1}$ and
  $\repa{x^2},\repa{y^2},\repa{z^2}$ into one, call the result
  $\repa{x},\repa{y},\repa{z}$:
  \begin{enumerate}
  \item 

    \Ali reveals $d = y^1 \oplus y^2$.

  \item 

    Compute $\repa{x} = \repa{x^1}\oplus \repa{x^2}$,
    $\repa{y}=\repa{y^1}$ and $\repa{z}= \repa{z^1}\oplus \repa{z^2}
    \oplus d \repa{x^2}$.

  \end{enumerate}
  To combine all $B$ \LAANDs in a bucket, just iterate by taking the
  result and combine it with the next element in the bucket.
    
\end{enumerate}
\end{boxfig}

Similar to the the way we removed leakage in \secref{aot} we start by
producing $B \ell$ \LAANDs. Then we randomly distribute the $B \ell$
\LAANDs into $\ell$ buckets of size $B$. Finally we combine the
\LAANDs in each bucket into one \AAND which is secure if at least one
\LAAND in the bucket was not leaky. The protocol is described
in~\figref{PROTAND}. 
The proof of~\thmref{thmAND} can be found in~\appref{proofofthmAND}.

\begin{theorem} \thmlab{thmAND} Let $\AAND(\ell)$ denote the box which
  outputs $\ell$ \AAND{}s as in \DEAL. If $(\log_2(\ell) + 1)(B-1) \ge
  \psi$, then the protocol in~\figref{PROTAND} securely implements
  $\AAND(\ell)$ in the $\LAAND(B \ell)$-hybrid model with security
  parameter $\psi$.
\end{theorem}

%% In the proof it is shown that the probability that a bucket ends up
%% with only leaky triples is upper bounded by $(2 \ell)^{1-B}$ and that
%% the protocol is perfectly secure when all buckets contain a non-leaky
%% triple. As an example, if we generate $\ell = 2^{19}$ \AANDs with a
%% bucket size of $B = 6$, we get an insecurity of $2^{-100}$, which is
%% tolerable. Doing so would require $6 \cdot 2^{19}$ \LAANDs which in
%% turn would require $3 \cdot 6 \cdot 2^{19}$ \ABITs. Generating $3
%% \cdot 6 \cdot 2^{19}$ \ABITs can be done extremely fast as the
%% dominating cost is around $3 \cdot 6 \cdot 2^{19}$ calls to a hash
%% function and a few seed OTs. For timings we refer to \appref{exp}.

This completes the description of our protocol. For the interested
reader, a diagrammatic recap of the construction is given in
\appref{fod}.

\Section{Experimental Results}\seclab{exp}

We did a proof-of-concept implementation in Java. The hash function in
our protocol was implemented using Java's standard implementation of
SHA256. The implementation consists of a circuit-independent protocol
for preprocessing all the random values output by \DEAL, a framework
for constructing circuits for a given computation, and a run-time
system which takes preprocessed values, circuits and inputs and carry
out the secure computation. 

We will not dwell on the details of the implementation, except for one
detail regarding the generation of the circuits. In our
implementation, we do not compile the function to be evaluated into a
circuit in a separate step. The reason is that this would involve
storing a huge, often highly redundant, circuit on the disk, and
reading it back. This heavy disk access turned out to constitute a
significant part of the running time in an earlier of our prototype
implementations which we discarded. Instead, in the current prototype,
circuits are generated on the fly, in chunks which are large enough
that their evaluation generate large enough network packages that we
can amortize away communication latency, but small enough that the
circuit chunks can be kept in memory during their evaluation. A
circuit compiled is hence replaced by a succinct program which
generates the circuit in a streaming manner. This circuit stream is
then sent through the runtime machine, which receives a separate
stream of preprocessed \DEAL-values from the disk and then evaluates
the circuit chunk by chunk in concert with the runtime machine at the
other party in the protocol. The stream of preprocessed \DEAL-values
from the disk is still expensive, but we currently see no way to avoid
this disk access, as the random nature of the preprocessed values
seems to rule out a succinct representation.

For timing we did oblivious ECB-AES encryption. (Both parties input a
secret $128$-bit key $K_A$ respectively $K_B$, defining an AES key $K
= K_A \oplus K_B$. \Ali inputs a secret $\ell$-block message $(m_1,
\ldots, m_\ell) \in \zo^{128\ell}$. \Bob learns $(E_K(m_1), \ldots,
E_K(m_\ell))$.) We used the AES circuit
from~\cite{DBLP:conf/asiacrypt/PinkasSSW09} and we thank Benny Pinkas,
Thomas Schneider, Nigel P.~Smart and Stephen C.~Williams for providing
us with this circuit.

The reason for using AES is that it provides a reasonable sized
circuit which is also reasonably complex in terms of the structure of
the circuit and the depth, as opposed to just running a lot of AND
gates in parallel. Also, AES has been used for benchmark in previous
implementations, like~\cite{DBLP:conf/asiacrypt/PinkasSSW09}, which
allows us to do a crude comparison to previous implementations. The
comparison can only become crude, as the experiments were run in
different experimental setups.

\begin{boxfig}{Timings. Left table is average over $5$ runs. Right
    table is from single runs. Units are as follows: $\ell$ is number
    of $128$-bit blocks encrypted, $G$ is Boolean gates,
    $\sigma$ is bits of security, $\Tpre, \Tonl, \Ttot$ are
    seconds.}{B}
\begin{center}
\begin{tabular}{|r|r|r|r|r|r|r|}
\hline
\hline
~ $\ell$       & 
~ $G$          & 
~ $\sigma$     &  
~ $\Tpre$      & 
~ $\Tonl$      & 
~ $\Ttot/\ell$ &
~ $G/\Ttot$    \\
\hline
$1$         & 
$34\mathord{,}520$    & 
$55$         &
$38$        &
$4$         &
$44$        &
$822$       \\
$27$         & 
$922\mathord{,}056$    & 
$55$         &
$38$        &
$5$         &
$1.6$        &
$21\mathord{,}545$       \\
$54$         & 
$1\mathord{,}842\mathord{,}728$  & 
$58$         &
$79$        &
$6$         &
$1.6$        &
$21\mathord{,}623$       \\
$81$         & 
$2\mathord{,}765\mathord{,}400$  & 
$60$         &
$126$        &
$10$         &
$1.7$        &
$20\mathord{,}405$       \\
$108$         & 
$3\mathord{,}721\mathord{,}208$  & 
$61$         &
$170$        &
$12$         &
$1.7$        &
$20\mathord{,}541$       \\
$135$         & 
$4\mathord{,}642\mathord{,}880$  & 
$62$         &
$210$        &
$15$         &
$1.7$        &
$20\mathord{,}637$       \\
\hline
\hline
\end{tabular} ~\ ~
\begin{tabular}{|r|r|r|r|r|r|r|}
\hline
\hline
~ $\ell$       & 
~ $G$          & 
~ $\sigma$     &  
~ $\Tpre$      & 
~ $\Tonl$      & 
~ $\Ttot/\ell$ &
~ $G/\Ttot$    \\
\hline
$256$         & 
$8\mathord{,}739\mathord{,}200$  & 
$65$         &
$406$        &
$16$         &
$1.7$        &
$20\mathord{,}709$       \\
$512$         & 
$17\mathord{,}478\mathord{,}016$  & 
$68$         &
$907$        &
$26$         &
$1.8$        &
$18\mathord{,}733$       \\
$1\mathord{,}024$         & 
$34\mathord{,}955\mathord{,}648$  & 
$71$         &
$2\mathord{,}303$        &
$52$         &
$2.3$        &
$14\mathord{,}843$       \\
$2\mathord{,}048$         & 
$69\mathord{,}910\mathord{,}912$  & 
$74$         &
$5\mathord{,}324$        &
$143$         &
$2.7$        &
$12\mathord{,}788$       \\
$4\mathord{,}096$         & 
$139\mathord{,}821\mathord{,}440$  & 
$77$         &
$11\mathord{,}238$        &
$194$         &
$2.8$        &
$12\mathord{,}231$       \\
$8\mathord{,}192$         & 
$279\mathord{,}642\mathord{,}496$  & 
$80$         &
$22\mathord{,}720$        &
$258$         &
$2.8$         &
$12\mathord{,}170$       \\
$16\mathord{,}384$         & 
$559\mathord{,}284\mathord{,}608$  & 
$83$         &
$46\mathord{,}584$        &
$517$         &
$2.9$         &
$11\mathord{,}874$       \\
\hline
\hline
\end{tabular}

~

\end{center}
\end{boxfig}

In the timings we ran \Ali and \Bob on two different machines on
Anonymous University's intranet (using two Intel Xeon E3430 2.40GHz
cores on each machine). We recorded the number of Boolean gates
evaluated ($G$), the time spent in preprocessing ($\Tpre$) and the
time spent by the run-time system ($\Tonl$). In the table in
\figref{B} we also give the amortized time per AES encryption
($\Ttot/\ell$ with $\Ttot \Def \Tpre + \Tonl$) and the number of gates
handled per second ($G/\Ttot$). The time $\Tpre$ covers the time spent
on computing and communicating during the generation of the values
preprocessed by \DEAL, and the time spent storing these value to a
local disk.  The time $\Tonl$ covers the time spent on generating the
circuit and the computation and communication involved in evaluating
the circuit given the values preprocessed by \DEAL.

We work with two security parameters. The computational security
parameter $\kappa$ specifies that a poly-time adversary should have
probability at most $\poly(\kappa) 2^{-\kappa}$ in breaking the
protocol. The statistical security parameter $\sigma$ specifies that
we allow the protocol to break with probability $2^{-\sigma}$
independent of the computational power of the adversary. As an example
of the use of $\kappa$, our keys and therefore MACs have length
$\kappa$. This is needed as the adversary learns $H(K_i)$ and $H(K_i
\oplus \Delta)$ in our protocols and can break the protocol given
$\Delta$. As an example of the use of $\sigma$, when we generate
$\ell$ gates with bucket size $B$, then $\sigma \le
(\log_2(\ell)+1)(B-1)$ due to the probability $(2 \ell)^{1-B}$ that a
bucket might end up containing only leaky components. This probability
is independent of the computational power of the adversary, as the
components are being bucketed by the honest party after it is
determined which of them are leaky.

In the timings, the computational security parameter has been set to
$120$.  Since our implementation has a fixed bucket size of $4$, the
statistical security level depends on $\ell$. In the table, we specify
the statistical security level attained ($\sigma$ means insecurity
$2^{-\sigma}$). At computational security level $120$, the
implementation needs to do $640$ seed OTs. The timings do not include
the time needed to do these, as that would depend on the
implementation of the seed OTs, which is not the focus here. We note,
however, that using, e.g., the implementation in
\cite{DBLP:conf/asiacrypt/PinkasSSW09}, the seed OTs could be done in
around $20$ seconds, so they would not significantly affect the
amortized times reported.

The dramatic drop in amortized time from $\ell = 1$ to $\ell = 27$ is
due to the fact that the preprocessor, due to implementation choices,
has a smallest unit of gates it can preprocess for. The largest number
of AES circuits needing only one, two, three, four and five units is
$27$, $54$, $81$, $108$ and $135$, respectively. Hence we preprocess
equally many gates when $\ell = 1$ and $\ell = 27$.

As for total time, we found the best amortized behavior at $\ell =
54$, where oblivious AES encryption of one block takes amortized $1.6$
seconds, and we handle $21\mathord{,}623$ gates per second. As for
online time, we found the best amortized behavior at $\ell = 2048$,
where handling one AES block online takes amortized $32$ milliseconds,
and online we handle $1\mathord{,}083\mathord{,}885$ gates per
second. We find these timings encouraging and we plan an
implementation in a more machine-near language, exploiting some of the
findings from implementing the prototype.

\clearpage
\appendix 

\longOnly{\bibliographystyle{alpha}}
\shortOnly{\bibliographystyle{plain}}
\bibliography{AOT}

\clearpage

\Section{Complexity Analysis}
\applab{complexity}

We report here on the complexity analysis of our protocol. As showed
in Corollary~\ref{cor:total}, the protocol requires an initial call to
an ideal functionality for $(\OT(\frac{44}{3} \psi,
\psi),\EQ(\psi))$. After this, the cost per gate is only a number of
invocations to a cryptographic hash function $H$. In this section we
give the exact number of hash functions that we use in the
construction of the different primitives. As the final protocol is
completely symmetric, we count the total number of calls to $H$ made
by both parties.

\begin{description}
\item[Equality \EQ:] The \EQ box can be securely implemented with $2$
  calls to a hash function $H$.
\item[Authenticated OT \AOT:] Every \AOT costs $4B$ calls to \ABIT,
  $2B$ calls to \EQ, and $6B$ calls to $H$, where $B$ is the ``bucket
  size''.
\item[Authenticated AND \AAND:] Every \AAND costs $3B$ calls to \ABIT,
  $B$ calls to \EQ, and $3B$ calls to $H$, where $B$ is the ``bucket
  size''.
\item[2PC Protocol, Input Gate:] Input gates cost $1$ \ABIT.
\item[2PC Protocol, AND Gate:] AND gates cost $2$ \AOT, $2$ \AAND, $2$ \ABIT.
\item[2PC Protocol, XOR Gate:] XOR gates require no calls to $H$.
\end{description}

The cost per \ABIT, in the protocol described in the paper, requires
$59$ calls to $H$. However, using some further optimizations (that are
not described in the paper, as they undermine the modularity of our
constructions) we can take this number down to $8$.

By plugging in these values we get that the cost per input gate is
$59$ calls to $H$ ($8$ with optimizations), and the cost per AND gate
is $856B+118$ calls to $H$ ($142B+16$ with optimizations). The
implementation described in \secref{exp} uses the optimized version of
the protocol and buckets of fixed size $4$, and therefore the total
cost per AND gate is $584$ calls to $H$.

As described in \secref{twopcfromaot} we can greatly reduce communication
complexity of our protocol by deferring the MAC checks. However, this trick
comes at cost of two calls to $H$ (one for each player) every time we do a
``reveal''. This adds $2B$ hashes for each $\AOT$ and
$\AAND$ and in total adds $8B + 20$ hashes to the cost each AND
gate. This added cost is not affected by the optimization mentioned above.

\newcommand{\Adv}{\ant{A}}

\Section{Proof of \thmref{twopc}}\applab{proofoftwopc}
  The simulator can be built in a standard way, incorporating the
  \DEAL box and learning all the shares, keys and MACs
  that the adversary was supposed to use in the protocol.

  In a little more detail, knowing all outputs from \DEAL to the
  corrupted parties allows the simulator to extract inputs used by
  corrupted parties and input these to the box \FMPC
  on behalf of the corrupted parties. As an example, if \Ali is
  corrupted, then learn the $x_A$ sent to \Ali by \DEAL in
  \textbf{Input} and observe the value $x_B$ sent by \Ali to
  \Bob. Then input $x = x_A \oplus x_B$ to \FMPC. This is the same
  value as shared by $\reps{x}=\repab{x_A}{x_B}$ in the protocol.

  Honest parties are run on uniformly random inputs, and when a honest
  party (\Ali say) is supposed to help open $\reps{x}$, then the
  simulator learns from \FMPC the value $x'$ that $\reps{x}$ should be
  opened to. Then the simulator computes the share $x_B$ that \Bob
  holds, which is possible from the outputs of \DEAL to \Bob. Then the
  simulator learns the key $K_{x_A}$ that \Bob uses to authenticate
  $x_A$, which can also be computed from the outputs of \DEAL to
  \Bob. Then the simulator lets $x_A = x' \oplus x_B$ and and lets
  $M_{x_A} = K_{x_A} \oplus x_A K_{x_A}$ and sends $(x_A,M_{x_A})$ to
  \Bob.

  The simulator aborts if the adversary ever successfully sends some
  inconsistent bit, i.e., a bit different from the bit it should send
  according to the protocol and its outputs from \DEAL. 

  It is easy to see that the protocol is passively secure and that if
  the adversary never sends an inconsistent bit, then it is perfectly
  following the protocol up to input substitution. So, to prove
  security it is enough to prove that the adversary manages to send an
  inconsistent bit with negligible probability. However, sending an
  inconsistent bit turns out to be equivalent to guessing the global
  key $\Delta$.

  We now formalize the last claim. Consider the following
  game $\Game_{I,I}$ played by an attacker \Ali:
  \begin{description}
  \item[Global key:]

    A global key $\Delta \leftarrow \zo^\kappa$ is sampled with some
    distribution and \Ali might get side information on $\Delta$.

  \item[MAC query I:]

    If \Ali outputs a query $(\texttt{mac}, b, l)$, where $b \in \zo$
    and $l$ is a label which $\Ali$ did not use before, sample a
    fresh local key $K \inR \zo^\kappa$, give $M = K\oplus b \Delta$ to
    \Ali and store $(l,K,b)$.

  \item[Break query I:]

    If \Ali outputs a query $(\texttt{break}, a_1, l_1, \ldots, a_p,
    l_p, M')$, where $p$ is some positive integer and values
    $(l_1,K_1,b_1), \ldots, (l_p,K_p,b_p)$ are stored, then let $K =
    \oplus_{i=1}^p a_i K_i$ and $b = \oplus_{i=1}^p a_i b_i$. If $M' =
    K\oplus (1\oplus b) \Delta$, then $\Ali$ wins the game. This query
    can be used only once.
    
  \end{description}

  We want to prove that if any \Ali can win the game with probability
  $q$, then there exist an adversary \Bob which does not use more
  resources than \Ali and which guesses $\Delta$ with probability $q$
  without doing any MAC queries. Informally this argues that
  breaking the scheme is linear equivalent to guessing $\Delta$
  without seeing any MAC values.

  For this purpose, consider the following modified game $\Game_{II,II}$
  played by an attacker \Ali:
  \begin{description}
  \item[Global key:]

    \emph{No change}.

  \item[MAC query II:]

    If \Ali outputs a query $(\texttt{mac}, b, l, M)$, where $b \in
    \zo$ and $l$ is a label which $\Ali$ did not use before and $M
    \in \zo^\kappa$, let $K = M \oplus b \Delta$ and store $(l,K,b)$.

  \item[Break query II:]

    If \Ali outputs a query $(\texttt{break}, \Delta')$ where
    $\Delta' = \Delta$, then $\Ali$ wins the game. This query can be
    used only once.

  \end{description}

  We let $\Game_{II,I}$ be the hybrid game with \textbf{MAC query II} and
  \textbf{Break query I}.

  We say that an adversary $\Ali$ is no stronger than adversary
  $\Bob$ if $\Ali$ does not perform more queries than
  $\Bob$ does and the running time of $\Ali$ is asymptotically
  linear in the running time of $\Bob$.
  
  \begin{lemma}
    For any adversary $\Ali_{I,I}$ for $\Game_{I,I}$ there exists an
    adversary $\Ali_{II,I}$ for $\Game_{II,I}$ which is no stronger
    than $\Ali_{I,I}$ and which wins the game with the same
    probability as $\Ali_{I,I}$.
  \end{lemma}
  \begin{proof}
    Given an adversary $\Ali_{I,I}$ for $\Game_{I,I}$, consider the
    following adversary $\Ali_{II,I}$ for $\Game_{II,I}$. The
    adversary $\Ali_{II,I}$ passes all side information on $\Delta$ to
    $\Ali_{I,I}$. If $\Ali_{I,I}$ outputs $(\texttt{mac}, b, l)$, then
    $\Ali_{II,I}$ samples $M \inR \zo^{\kappa}$, outputs
    $(\texttt{mac}, b, l, M)$ to $\Game_{II,I}$ and returns $M$ to
    $\Ali_{I,I}$. If $\Ali_{I,I}$ outputs $(\texttt{break}, a_1, l_1,
    \ldots, a_p, l_p, M')$, then $\Ali_{II,I}$ outputs
    $(\texttt{break}, a_1, l_1, \ldots, a_p, l_p, M')$ to
    $\Game_{II,I}$. It is easy to see that $\Ali_{II,I}$ makes the
    same number of queries as $\Ali_{I,I}$ and has a running time
    which is linear in that of $\Ali_{I,I}$, and that $\Ali_{II,I}$
    wins with the same probability as $\Ali_{I,I}$. Namely, in
    $\Game_{I,I}$ the value $K$ is uniform and $M = K \oplus b
    \Delta$. In $\Game_{II,I}$ the value $M$ is uniform and $K = M
    \oplus b \Delta$. This gives the exact same distribution on
    $(K,M)$. \putbox
  \end{proof}

  \begin{lemma}
    For any adversary $\Ali_{II,I}$ for $\Game_{II,I}$ there exists
    an adversary $\Ali_{II,II}$ for $\Game_{II,II}$ which is no
    stronger than $\Ali_{II,I}$ and which wins the game with the same
    probability as $\Ali_{II,I}$.
  \end{lemma}
  \begin{proof}
    Given an adversary $\Ali_{II,I}$ for $\Game_{II,I}$, consider the
    following adversary $\Ali_{II,II}$ for $\Game_{II,II}$. The
    adversary $\Ali_{II,II}$ passes any side information on $\Delta$
    to $\Ali_{II,I}$. If $\Ali_{II,I}$ outputs $(\texttt{mac}, b, l,
    M)$, then $\Ali_{II,II}$ outputs $(\texttt{mac}, b, l, M)$ to
    $\Game_{II,II}$ and stores $(l,M,b)$.  If $\Ali_{II,I}$ outputs
    $(\texttt{break}, a_1, l_1, \ldots, a_p, l_p, M')$, where values
    $(l_1,M_1,b_1), \ldots, (l_p,M_p,b_p)$ are stored, then let $M =
    \oplus_{i=1}^p a_i M_i$ and $b = \oplus_{i=1}^p a_i b_i$ and
    output $(\texttt{break}, M \oplus M')$. For each $(l_i,M_i,b_i)$
    let $K_i$ be the corresponding key stored by $\Game_{II,II}$. We
    have that $M_i = K_i \oplus b_i \oplus \Delta$, so if we let $K =
    \oplus_{i=1}^p a_i K_i$, then $M = K \oplus b \Delta$.  Assume
    that $\Ali_{II,I}$ would win $\Game_{II,I}$, i.e., $M' = K\oplus
    (1\oplus b) \Delta$.  This implies that $M \oplus M' = K \oplus b
    \Delta \oplus K \oplus (1\oplus b) \Delta = \Delta$, which means
    that $\Ali_{II,II}$ wins $\Game_{II,II}$. \putbox
  \end{proof}

Consider then the following game $\Game_{II}$ played by an attacker
\Ali:
\begin{description}
\item[Global key:]

  \emph{No change}.
\item[MAC query:]

  \emph{No MAC queries are allowed.}

\item[Break query II:]

  \emph{No change}.

  %If \Ali outputs a query $(\texttt{break}, \Delta')$ where
  %$\Delta' = \Delta$, then $\Ali$ wins the game. This query can be
  %used only once.

\end{description}

\begin{lemma}
  For any adversary $\Ali_{II,II}$ for $\Game_{II,II}$ there exists
  an adversary $\Ali_{II}$ for $\Game_{II}$ which is no stronger
  than $\Ali_{II,II}$ and which wins the game with the same
  probability as $\Ali_{II,II}$.
\end{lemma}
\begin{proof}
  Let $\Ali_{II} =\Ali_{II,II}$. The game $\Game_{II}$ simply ignores
  the MAC queries, and it can easily be seen that they have no effect
  on the winning probability, so the winning probability stays the
  same. \putbox
\end{proof}

\begin{corollary}
  For any adversary $\Ali_{I,I}$ for $\Game_{I,I}$ there exists an
  adversary $\Ali_{II}$ for $\Game_{II}$ which is no stronger than
  $\Ali_{I,I}$ and which wins the game with the same probability as
  $\Ali_{I,I}$.
\end{corollary}

This formalizes the claim that the only way to break the scheme is to
guess $\Delta$.

\Section{Proof of \thmref{THMWABIT}}\applab{proofofTHMWABIT}

The simulator answers a global key query $\Gamma$ to $\WABIT$ by doing
the global key query $\B{A} \Gamma$ on the ideal functionality $\ABIT$ and
returning the reply. This gives a perfect simulation of these queries,
and we ignore them below.

  Correctness of the protocol is straightforward: We have that $M'_i =
  K'_i \oplus x_i \Gamma_A$, so $M_i = \B{A} M'_i = \B{A} K'_i \oplus
  x_i \B{A} \Gamma_A = K_i \oplus x_i \Delta_A$. Clearly the protocol
  leaks no information on the $x_i$'s as there is only communication
  from \Bob to \Ali. It is therefore sufficient to look at the case
  where \Ali is corrupted. We are not going to give a simulation
  argument but just show that $\Delta_A$ is uniformly random in the
  view of \Ali except with probability $2^{2-\psi}$. Turning this
  argument into a simulation argument is straight forward.

  We start by proving three technical lemmas.

  Assume that $\LL$ is a class of leakage functions on $\tau$
  bits which is $\kappa$-secure. Consider the following game.
  \begin{enumerate}
  \item

    Sample $\Gamma_A \inR \zo^\tau$.

  \item
  
    Get $L \in \LL$ from \Ali and sample $(S,c) \leftarrow L$.

  \item

    Give $\{ (j,(\Gamma_A)_j) \}_{j \in S}$ to \Ali.

  \item

    Sample $\B{A} \inR \zo^{\psi \times \tau}$ and give $\B{A}$ to \Ali.

  \item

    Let $\Delta_A = \B{A} \Gamma_A$.

  \end{enumerate}

  We want to show that $\Delta_A$ is uniform to \Ali except with
  probability $2^{2-\psi}$.  When we say that \emph{$\Delta_A$ is
  uniform to \Ali} we mean that $\Delta_A$ is uniformly random in
  $\zo^\psi$ and independent of the view of \Ali. When we say
  \emph{except with probability $2^{2-\psi}$} we mean that there
  exists a failure event $F$ for which it holds that
  \begin{enumerate}
  \item

    $F$ occurs with probability at most $2^{2-\psi}$ and 

  \item

    when $F$ does not occur, then $\Delta_A$ is uniform to \Ali.

  \end{enumerate}

  For a subset $S \subset \{ 1, \ldots, \tau\}$ of the column indices,
  let $\B{A}^S$ be the matrix where column $j$ is equal to $\B{A}^j$
  if $j \in S$ and column $j$ is the $0$ vector if $j \not\in S$. We
  say that we blind out column $j$ with $0$'s if $j \not\in
  S$. Similarly, for a column vector $\vec{v}$ we use the notation
  $\vec{v}_S$ to mean that we set all indices $v_i$ where $i \not\in S$
  to be $0$. Note that $\B{A} \vec{v}_S = \B{A}^S \vec{v}$. Let
  $\overline{S} = \{ 1, \ldots, \tau \} \setminus S$.

  \begin{lemma}\lemmlab{lemmaone}
    Let $S$ be the indices of the bits learned by \Ali and let $\B{A}$
    be the matrix in the game above. If $\B{A}^{\overline{S}}$ spans
    $\zo^{\psi}$, then $\Delta_A$ is uniform to \Ali.
  \end{lemma}
  \begin{proof}
    We start by making two simple observations. First of all, if \Ali
    learns $(\Gamma_A)_j$ for $j \in S$, then it learns
    $(\Gamma_A)_S$\footnote{Here we are looking at the string
    $\Gamma_A$ as a column vector of bits.}, so it knows $\B{A}
    (\Gamma_A)_S = \B{A}^S \Gamma_A$.  The second observation is that
    $\B{A}\Gamma_A = \B{A}^S \Gamma_A + \B{A}^{\overline{S}}
    \Gamma_A$, as $\B{A} = \B{A}^S + \B{A}^{\overline{S}}$.  The lemma
    follows directly from these observations and the premise: We
    have that $\B{A}^{\overline{S}} \Gamma_A$ is uniformly
    random in $\zo^\psi$ when the columns of $\B{A}^{\overline{S}}$
    span $\zo^{\psi}$. Since $\B{A}^{\overline{S}} \Gamma_A = \B{A}
    (\Gamma_A)_{\overline{S}}$ and $(\Gamma_A)_{\overline{S}}$ is
    uniformly random and independent of the view of \Ali it follows
    that $\B{A}^{\overline{S}} \Gamma_A$ is uniformly random and
    independent of the view of \Ali. Since $\B{A}^S \Gamma_A$ is known
    by \Ali it follows that $\B{A}^S \Gamma_A + \B{A}^{\overline{S}}
    \Gamma_A$ is uniform to \Ali. The proof concludes by using that
    $\Delta_A = \B{A}^S \Gamma_A + \B{A}^{\overline{S}} \Gamma_A$.
    \putbox
  \end{proof}

  \begin{lemma}\lemmlab{lemmatwo}
    Let $W$ be the event that $\vert S \vert \ge \tau - n$ and
    $c=1$. Then $\prob{W} \le 2^{-\psi}$.
  \end{lemma}
  \begin{proof}
    We use that
    \begin{itemize}
    \item $\kappa = \frac34 \tau$,
    \item $\tau = \alpha n$ for $\alpha = \frac{44}{27}$,
    \item $n = \frac92 \psi$,
    \item $\LL$ is $\kappa$-secure on $\tau$ bits.
    \end{itemize}
    Without loss of generality we can assume that \Ali plays an
    optimal $L \in \LL$, i.e., $\log_2(\E{c 2^{\vert S \vert}}) =
    \leak_\LL$. Since $\LL$ is $\kappa$ secure on $\tau$ bits, it
    follows that $\leak_\LL \le \tau - \kappa = \frac14 \tau$.
    This gives that
    \begin{equation}\eqlab{tobefrred}
      \E{c 2^{\vert S \vert}} \le 2^{\frac14 \tau}\ ,
    \end{equation}
    which we use later.

    Now let $\overline{W}$ be the event that $W$ does not happen. By
    the properties of conditional expected value we have that
    $$\E{c 2^{\vert S \vert}} = \prob{W} \E{c 2^{\vert S \vert} \vert
    W} + \prob{\overline{W}} \E{c 2^{\vert S \vert} \vert
    \overline{W}}\ .$$ When $W$ happens, then $\vert S \vert \ge \tau
    - n = (\alpha - 1)n$ and $c = 1$, so $c 2^{\vert S \vert} =
    2^{\vert S \vert} \ge 2^{(\alpha - 1)n}$. This gives that
    $$\E{c 2^{\vert S \vert} \vert W} \ge 2^{(\alpha - 1)n}\ .$$ Hence 
    $$\E{c 2^{\vert S \vert}} \ge \prob{W} 2^{(\alpha - 1)n}\ .$$ 
    Combining with \eqref{tobefrred} we get that
    $$\prob{W} \le 2^{\frac14 \tau - (\alpha - 1)n}\ .$$ It is,
    therefore, sufficient to show that $\frac14 \tau - (\alpha - 1)n =
    - \psi$, which can be checked to be the case by definition of
    $\tau, \alpha, n$ and $\psi$. 
\putbox
  \end{proof}

  \begin{lemma}\lemmlab{lemmathree}
    Let $x_1, \ldots, x_n \inR \zo^\psi$. Then $x_1, \ldots, x_n$ span
    $\zo^\psi$ except with probability $2^{1-\psi}$.
  \end{lemma}
  \begin{proof}
    We only use that
    \begin{itemize}
    \item $n=\frac{9}{2}\psi$. 
    \end{itemize}
    Define random variables $Y_1, \ldots, Y_n$ where $Y_i = 0$ if
    $x_1, \ldots, x_{i-1}$ spans $\zo^\psi$ or the span of $x_1,
    \ldots, x_{i-1}$ does not include $x_i$. Let $Y_i = 1$ in all
    other cases. Note that if $x_1, \ldots, x_{i-1}$ spans $\zo^\psi$,
    then $\prob{Y_i = 1} = 0 \le \oh$ and that if $x_1, \ldots,
    x_{i-1}$ does not span $\zo^\psi$, then they span at most half of
    the vectors in $\zo^{\psi}$ and hence again $\prob{Y_i = 1} \le
    \oh$. This means that it holds for all $Y_i$ that $\prob{Y_i = 1}
    \le \oh$ independently of the values of $Y_j$ for $j \ne i$. This
    implies that if we let $Y = \sum_{i=1}^n Y_i$, then
    $$\prob{Y \ge \oh(a + n) } \le 2 e^{-a^2/2 n}\ ,$$ using
    the random walk bound. Namely, let $X_i = 2 Y_i - 1$. Then $X_i
    \in \{ -1, 1 \}$ and it holds for all $i$ that $\prob{X_i = 1} \le
    \oh$ independently of the other $X_j$. If the $X_i$ had been
    independent and $\prob{X_i = 1} = \prob{X_i = -1} = \oh$, and $X =
    \sum_{i=1}^n X_i$, then the random walk bound gives that
    $$\prob{X \ge a} \le 2 e^{-a^2/2 n}\ .$$ Since we have
    that $\prob{X_i = 1} \le \oh$ independently of the other $X_j$,
    the upper bound applies also to our setting. Then use that $X = 2
    Y - n$.  

    If we let $a = \frac52 \psi$, then $\oh(a + n) = \frac72 \psi = n
    - \psi$ and $2 e^{-a^2/2 n} = 2 e^{-\left(\frac52 \psi\right)^2/2
      \frac92 \psi} = 2 e^{- \frac{25}{36} \psi}$, and $e^{-
      \frac{25}{36}} < \oh$. It follows that $\prob{Y \ge n - \psi}
    \le 2^{1-\psi}$. When $Y \le n - \psi$, then $Y_i = 0$ for at
    least $\psi$ values of $i$. This is easily seen to imply that
    $x_1, \ldots, x_n$ contains at least $\psi$ linear independent
    vectors.  \putbox
  \end{proof}

  Recall that $W$ is the event that $\vert S \vert \ge \tau - n$ and
  $c=1$. By~\lemmref{lemmatwo} we have that $\prob{W} \le 2^{-n} \le
  2^{-\psi}$. For the rest of the analysis we assume that $W$ does not
  happen, i.e., $\vert S \vert \le \tau -
    n$ and hence $\vert \overline{S}
  \vert \ge \tau = \frac92 \psi$.  
  Since $\B{A}$ is picked uniformly at
  random and independent of $S$ it follows that $\frac92 \psi$ of the
  columns in $\B{A}^{\overline{S}}$ are uniformly random and
  independent. Hence, by~\lemmref{lemmathree}, they span $\zo^\psi$
  except with probability $2^{1-\psi}$. We let $D$ be the event that
  they do not span. If we assume that $D$ does not happen, then
  by~\lemmref{lemmaone} $\Delta_A$ is uniform to $\Ali$. I.e., if the
  event $F = W \cup D$ does not happen, then $\Delta_A$ is uniform to
  $\Ali$. And, $\prob{F} \le \prob{W} + \prob{D} \le 2^{-\psi} +
  2^{1-\psi} \le 2^{2-\psi} $.
  
\Section{Proof of \thmref{bitauth}}\applab{proofofbitauth}

Notice that since we have to prove that we implement $\LABIT$, which
has the global key queries, it would be stronger to show that we
implement a version of $\LABIT'$ which does not have these global key
queries. This is what we do below, as we let $\LABIT$ denote this
stronger box.

Given a pairing $\pi$, let $\SS(\pi) = \{ i \vert i < \pi(i) \}$, i.e.,
for each pair we add the smallest indexed to $\SS(\pi)$.

The cases where no party is corrupted and where \Bob is corrupted is
straight forward, so we will focus on the case that \Ali is corrupted.

The proof goes via a number of intermediary boxes, and for each we
show linear reducibility.

\Subsubsection{Approximating \LABIT, Version 1}
 
This box captures the fact that the only thing a malicious
\Ali can manage is to use different $\Gamma$'s in a few bit
authentications.

\begin{boxfig}{The First Intermediate Box IB1}{BOXBI1}
\begin{description}
\item[Honest-Parties:] As in \LABIT.
\item[Corrupted Parties:] \ 
\begin{enumerate}

\item
  
  If \Bob is corrupted: As in \LABIT.

\item
\begin{enumerate} 
\item
 If \Ali is corrupted, then \Ali inputs a functions $\col: \{ 1,
  \ldots, \Tau \} \rightarrow \{ 1, \ldots, \Tau \}$. We think of $\col$
  as assigning colors from $\{ 1, \ldots, \Tau \}$ to $\Tau$ balls
  named $1, \ldots, \Tau$. In addition \Ali inputs $\Lambda_1,
  \ldots, \Lambda_\Tau \in \zo^{\ell}$ and $L_1, \ldots, L_\Tau \in
  \zo^\ell$.  
\item
  Then the box samples a uniformly random pairing $\pi : \{
  1, \ldots, \Tau \} \rightarrow \{ 1, \ldots, \Tau \}$ and outputs
  $\pi$ to \Ali. We think of $\pi$ as pairing the $\Tau$ balls. Let
  $\SS = \SS(\pi)$ and let $\M = \{ i \in \SS \vert \col(i) \ne \col(\pi(i))
  \}$. We call $i \in \M$ a \emph{mismatched} ball.
\item Now \Ali inputs the guesses
  $\{ (i, g_i) \}_{i \in \M}$. 
\item The box samples $(y_1, \ldots, y_\Tau)
  \inR \zo^\Tau$. Then the box lets $c=1$ if $g_i = y_i$
  for $i \in \M$, otherwise it lets $c=0$. If $c=0$ the box outputs
  $\texttt{fail}$ to \Bob and terminates. Otherwise, for $i \in
  \SS$ it computes $N_i = L_i \oplus y_i \Lambda_{\col(i)}$ and outputs
  $\{ ((N_i, y_i) \}_{i \in \SS}$ to \Bob.
\end{enumerate}
\end{enumerate}
\end{description}
\end{boxfig}

\begin{lemma}
  \func{IB1} is linear reducible to
  $(\OT(2\tau,\ell),\EQ(\tau\ell))$.
\end{lemma}
\begin{proof}
  By
  observing \Ali's inputs to the OTs, the simulator learns all
  $(Y_{i,0},Y_{i,1})$. Let $L_i = Y_{i,0}$ and $\Gamma_i = Y_{i,0}
  \oplus Y_{i,1}$.

  Let $f = \vert \{ \Gamma_i \}_{i=1}^\Tau \vert$ and pick distinct
  $\Lambda_1, \ldots, \Lambda_f$ and $\col:
  \{1,\ldots,\Tau\}\rightarrow\{1,\ldots,\Tau\}$ such that $\Gamma_i =
  \Lambda_{\col(i)}$. By construction
  \begin{equation*}
    \begin{split}
      Y_{i,1} &= Y_{i,0} \oplus (Y_{i,0} \oplus Y_{i,1})\\
      &= L_i \oplus \Gamma_i\\
      &= L_i \oplus \Lambda_{\col(i)}\ .
    \end{split}
  \end{equation*}

  Input $\col$ and $\Lambda_1,\ldots, \Lambda_f$ and $L_1, \ldots,
  L_\Tau$ to \func{IB1} on behalf of \Ali and receive $\pi$.  Send $\pi$ to
  \Ali as if coming from \Bob along with uniformly random $\{ d_i
  \}_{i \in \SS}$.

  Then observe the inputs $Z_i$ from \Ali to the \EQ box.

  The simulator must now pick the guesses $g_i$ for $i \in \M$. Note
  that $i \in \M$ implies that $\Lambda_{\col(i)} \ne
  \Lambda_{\col(\pi(i))}$, which implies that $\Gamma_i \ne
  \Gamma_{\pi(i)}$.  We use this to pick $g_i$, as follows: after
  seeing $d_i$, \Ali knows that either $(y_i,y_{\pi(i)})=(0,d_i)$ or
  $(y_i,y_{\pi(i)})=(1,1\oplus d_i)$.  Hence an honest \Bob would
  input to the comparison the following value depending on $y_i$
  $$
  W_i(y_i) = (L_i \oplus L_{\pi(i)}\oplus d_i \Lambda_{\col(\pi(i))})
  \oplus y_i (\Lambda_{\col(i)}\oplus \Lambda_{\col(\pi(i))})\ .
  $$

  As $i \in \M$, the mismatched set, $\Lambda_{\col(i)} \neq
  \Lambda_{\col(\pi(i))}$ and therefore $W_i(0)\neq W_i(1)$. Therefore
  if \Ali's input to the \EQ box $Z_i$ is equal to $W_i(0)$
  (resp. $W_i(1)$), the simulator inputs a guess $g_i=0$
  (resp. $g_i=1$). In any other case, the simulator outputs
  \texttt{fail} and aborts.

  Notice that in the real-life protocol, if
  $g_i = y_i$, then $N_i = W_i(y_i) = Z_i$ and \Ali passes the
  test. If $g_i \ne y_i$, then $N_i = W_i(1 \oplus c_i) \ne Z_i$ and
  \Ali fails the test. So, the protocol and the simulation fails on
  the same event.  Note then that when the box does not fail, then it
  outputs
  \begin{equation*}
    \begin{split}
      N_i &= L_i \oplus y_i \Lambda_{\col(i)} \\
      &= Y_{i,0} \oplus y_i \Gamma_i\\
      &= Y_{i,0} \oplus y_i (Y_{i,0} \oplus Y_{i,1})\\
      &= Y_{i,y_i}\ ,
    \end{split}
  \end{equation*}
  exactly as the protocol. Hence the simulation is perfect. \putbox

%the input $W_i$ to the comparison will be the following
%  value depending $y_i$:
%  \begin{equation*}
%    \begin{split}
%      Z_i(c_i) &= X_{i,c_i} \oplus X_{\pi(i), c_{\pi(i)}}\\
%      &= (X_{i,0} \oplus c_i \Delta_i) \oplus (X_{\pi(i),0} \oplus
%      c_{\pi(i)} \Delta_{\pi(i)})\\  
%      &= (X_{i,0} \oplus X_{\pi(i),0}) \oplus (c_i \Delta_i \oplus
%      c_{\pi(i)} \Delta_{\pi(i)})\\  
%      &= (X_{i,0} \oplus X_{\pi(i),0}) \oplus 
%      c_i (\Delta_i \oplus \Delta_{\pi(i)}) \oplus z_i \Delta_{\pi(i)}\ .
%    \end{split}
%  \end{equation*}
%  Since $\Delta_i \oplus \Delta_{\pi(i)} \ne 0^\ell$ it follows that
%  $Z_i(0) \ne Z_i(1)$. If $N_i \ne Z_i(0)$ and $N_i \ne Z_i(1)$, then
%  let $d_i = \bot$. Otherwise, pick $d_i$ to be the unique bit such
%  that $N_i = Z_i(d_i)$. Notice that in the real-life protocol, if
%  $d_i = c_i$, then $N_i = Z_i(c_i) = Z_i$ and \Ali passes the
%  test. If $d_i \ne c_i$, then $N_i = Z_i(1 \oplus c_i) \ne Z_i$ and
%  \Ali fails the test. So, the protocol and the simulation fails on
%  the same event.  Note then that when the box does not fail, then it
%  outputs
%  \begin{equation*}
%    \begin{split}
%      N_i &= L_i \oplus c_i \Lambda_{c(i)} \\
%      &= X_{i,0} \oplus c_i \Delta_i\\
%      &= X_{i,0} \oplus c_i (X_{i,0} \oplus X_{i,1})\\
%      &= X_{i,c_i}\ ,
%    \end{split}
%  \end{equation*}
%  exactly as the protocol. Hence the simulation is perfect. \putbox
\end{proof}

\Subsubsection{Approximating \LABIT, Version 2}

We now formalize the idea that a wrong $\Gamma$-value is no worse that
a leaked bit.

We first need a preliminary definition of the most common color called
$\col_0$. If several colors are most common, then arbitrarily pick the
numerically largest one. To be more precise, for each color $c$, let
$C(c) = \{ j \in \{ 1, \ldots, \Tau \} \vert \col(j) = c \}$, let $a_0
= \max_{c} \vert C(c) \vert$ and let $\col_0 = \max \{ c \vert C(c) =
a_0 \vert \}$.

Consider the following box \func{IB2} in~\figref{BOXBI2} for
formalizing the second idea.

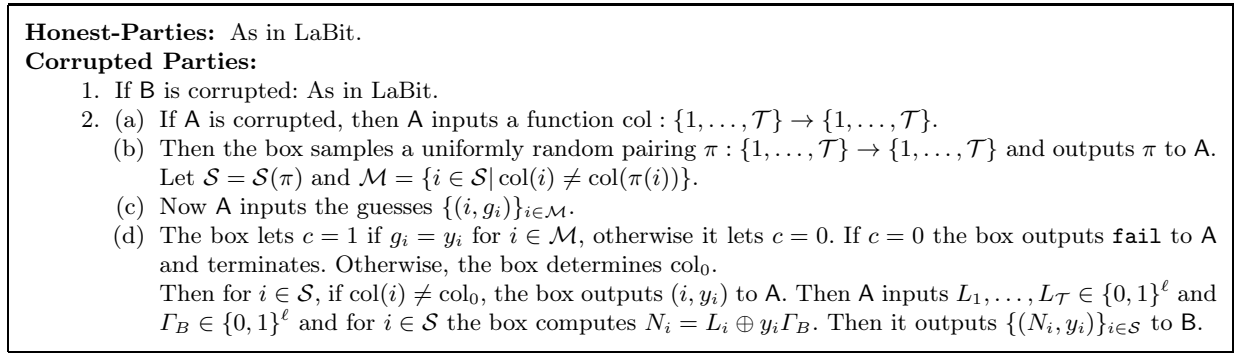
\begin{boxfig}{The Second Intermediate Box IB2}{BOXBI2}
\begin{description}
\item[Honest-Parties:] As in \LABIT.
\item[Corrupted Parties:] \ 
\begin{enumerate}

\item
  If \Bob is corrupted: As in \LABIT.
\item

\begin{enumerate}
\item
  If \Ali is corrupted, then \Ali inputs a function $\col: \{ 1,
  \ldots, \Tau \} \rightarrow \{ 1, \ldots, \Tau \}$. 
\item
 Then the box
  samples a uniformly random pairing $\pi : \{ 1, \ldots, \Tau \}
  \rightarrow \{ 1, \ldots, \Tau \}$ and outputs $\pi$ to \Ali. Let
  $\SS = \SS(\pi)$ and $\M = \{ i \in \SS \vert \col(i) \ne \col(\pi(i)) \}$.
\item Now \Ali inputs the guesses
  $\{ (i, g_i) \}_{i \in \M}$. 

\item The box lets $c=1$ if $g_i = y_i$ for $i \in \M$,
  otherwise it lets $c=0$. If $c=0$ the box outputs $\texttt{fail}$ to
  \Ali and terminates. Otherwise, the box determines $\col_0$.

  Then for $i \in \SS$, if $\col(i) \ne \col_0$, the box outputs $(i,
  y_i)$ to \Ali. Then \Ali inputs $L_1, \ldots, L_\Tau \in \zo^\ell$
  and $\Gamma_B \in \zo^\ell$ and for $i \in \SS$ the box computes
  $N_i = L_i \oplus y_i \Gamma_B$. Then it outputs $\{ (N_i, y_i)
  \}_{i\in \SS }$ to \Bob.
\end{enumerate}
\end{enumerate}
\end{description}
\end{boxfig}

\begin{lemma}
  \func{IB2} is linear locally reducible to \func{IB1}.
\end{lemma}
\begin{proof}
  The implementation of \func{IB2} consist simply of calling
  \func{IB1}.

  The case where \Bob or no party is corrupted is trivial, so
  assume that \Ali is corrupted. Note that the simulator must
  simulate \func{IB2} to the environment and is the one simulating
  \func{IB1} to the corrupted \Ali.

  First the simulator observes the inputs $\col$, $\Lambda_1, \ldots,
  \Lambda_\Tau \in \zo^{\ell}$ and $L_1, \ldots, L_\Tau \in \zo^\ell$
  of \Aliadv to \func{IB1} and inputs $\col$ to \func{IB2}.

  Then \func{IB2} outputs $\pi$ and the simulator inputs $\pi$ to
  \Aliadv as if coming from \func{IB1}, and computes $\M$ as
  \func{IB1} and \func{IB2} would have done.

  Then the simulator observes the guesses $\{(i,g_i)\}_{i\in \M}$ from
  \Aliadv to \func{IB1} and inputs $\{(i,g_i)\}_{i\in \M}$ to
  \func{IB2}. If \func{IB2} outputs \texttt{fail} to \Bob the
  simulation is over, and it is perfect as \func{IB1} and \func{IB2}
  fail based on the same event. If \func{IB2} does not fail it
  determines $\col_0$ and for $i \in \M$, if $\col(i) \ne \col_0$, the
  box outputs $(i, y_i)$ to the simulator. The simulator can also
  determine $\col_0$.

  Now let $\Gamma_B = \Lambda_{\col_0}$ and for $i \in \M$, if
  $\col(i) = \col_0$, let $L_i' = L_i$. Then for $i\in\M$, if
  $\col(i) \ne \col_0$, let $L_i' = (L_i \oplus y_i \Lambda_{\col(i)}) \oplus
  y_i \Gamma_B$. Then input $L_i', \ldots, L_\Tau'$ and $\Gamma_B$ to
  \func{IB2}.

  As a result \func{IB2} will  for $i \in \SS$ where
  $\col(i) = \col_0$, output $L_i' \oplus y_i \Gamma_B = L_i \oplus y_i
  \Lambda_{\col_0}$, and for for $i\in \SS$ where $\col(i) \ne
  \col_0$ it will output $L_i' \oplus y_i \Gamma_B = L_i \oplus y_i
  \Lambda_{\col(i)}$. Hence \func{IB2} gives exactly the outputs that
  \func{IB1} would have given after interacting with \Aliadv,
  giving a perfect simulation. \putbox
\end{proof}

\Subsubsection{Approximate \LABIT, Version 3}

We now massage \func{IB2} a bit to make it look like \LABIT. As a step
towards this, consider the box \func{IB3} in \figref{BOXTRAN}.

\begin{boxfig}{The third Intermediate Box, IB3}{BOXTRAN}
\begin{description}
\item[Honest-Parties:] As in \LABIT.
\item[Corrupted Parties:] \ 
\begin{enumerate}
\item Corrupted \Bob: As in \LABIT.
\item
\begin{enumerate}
\item
  If \Ali is corrupted, then \Ali inputs a function $\col: \{ 1,
  \ldots, \Tau \} \rightarrow \{ 1, \ldots, \Tau \}$. 
\item
  Then the box samples a uniformly random pairing $\pi : \{ 1, \ldots,
  \Tau \} \rightarrow \{ 1, \ldots, \Tau \}$ and outputs $\pi$ to
  \Ali. Let $\M = \{ i \in \SS \vert \col(i) \ne \col(\pi(i)) \}$. The
  box flips a coin $c \in \zo$ with $c=1$ with probability $2^{- \vert
  \M \vert}$.  If $c=0$ the box outputs $\texttt{fail}$ to \Bob and
  terminates. Otherwise, the box outputs \texttt{success} and the game
  proceeds.
\item Now \Ali inputs the guesses
  $\{ (i, g_i) \}_{i \in \M}$.

\item The box updates $y_i \leftarrow g_i$ for $i \in \M$.
  Then the box determines $\col_0$. Then for $i = \SS \setminus \M$,
  if $\col(i) \ne \col_0$, the box outputs $i$ to \Ali who inputs $g_i
  \in \zo$ and the box updates $y_i \leftarrow g_i$. 
\item  Then \Ali
  inputs $L_1, \ldots, L_{\Tau} \in \zo^\ell$ and $\Gamma_B \in \zo^\ell$
  and for $i\in\SS$ the box computes $N_i = L_i \oplus y_i
  \Gamma_B$. Then it outputs $\{(N_i, y_i)\}_{i\in\SS}$ to
  \Bob.
\end{enumerate}

\end{enumerate}
\end{description}
\end{boxfig}

\begin{lemma}
  \func{IB3} is linear locally reducible to \func{IB2}.
\end{lemma}

\begin{proof}

It is easy to see that \func{IB3} is linear locally reducible to
\func{IB2}---again the implementation consist simply of calling
\func{IB2}. To see this, consider first the change in how the box
fails and how the $y_i$ for $i \in \M$ are set. In \func{IB2} the box
fails exactly with probability $2^{- \vert \M \vert}$ as the
probability that $g_i = y_i$ for $i \in \M$ is exactly $2^{- \vert \M
\vert}$. Furthermore, if \func{IB2} does not fail, then $y_i = g_i$
for $i \in \M$. So, this is exactly the same behavior as
\func{IB3}, hence this change is really just another way to
implement the same box. As for the second change, the simulator will
input a uniformly random $g_i \inR \zo$ to \func{IB3} when
\func{IB3} outputs $i$ and will then show $(i, y_i)$ to the corrupted
\Aliadv expecting to interact with \func{IB2}.\putbox

\end{proof}

We then argue that we can define a class $\LL$ such that $\LABIT^\LL$
is linear locally reducible to \func{IB3}. Let $\LL$ be the following
class.
\begin{itemize}
\item

  A leakage function is specified by $L = \col$, where $\col: \{ 1,
  \ldots, \Tau \} \rightarrow \{ 1, \ldots, \Tau \}$.

\item

  To sample a leakage function $L = \col$, sample a uniformly random
  pairing $\pi : \{ 1, \ldots, \Tau \} \rightarrow \{ 1, \ldots, \Tau
  \}$, let $\SS = \SS(\pi)$, let $\Pi: \SS(\pi) \rightarrow \{ 1,
  \ldots, \tau \}$ be the order preserving permutation, let $\M = \{ j
  \in \SS \vert \col(j) \ne \col(\pi(j)) \}$, let $c=1$ with
  probability $2^{-\vert \M \vert}$ and $c=0$ otherwise, let $\col_0$
  be the most common color as defined before, let $S' = \M \cup \{ j \in
  \SS \vert \col(j) \ne \col_0 \}$, $S = \pi(S')$ and output $(c,S)$.

\end{itemize}

Playing with \func{IB3} and $\LABIT^\LL$ will give the same failure
probability and will allow to specify the same bits. The only
difference is that when playing with $\LABIT^\LL$, the corrupted
\Aliadv does not get to see $\pi$, as $\LABIT^\LL$ does not leak the
randomness used to sample the leakage function $L$. Below we argue
that given $c$ and $S$ one can efficiently sample a uniformly random
pairing $\pi$ which would lead to $S$ given $c$. Turning this into a
simulation argument is easy: the simulator will know $c$ and $S$ and
will sample $\pi$ from these and show this $\pi$ to \Aliadv, hence
perfectly simulating \func{IB3}. This gives the following lemma.

\begin{lemma}
  $\LABIT^\LL(\tau,\ell)$ is linear locally reducible to \func{IB3}.
\end{lemma}

The simulator knows $\col$ and $S$ and it can determine $\col_0$. From
$\col_0$ the simulator can also compute
\begin{equation*}
\begin{split}
T &= 
S \setminus \{ j \in \{1,\ldots,\Tau \} \vert \col(j) \ne \col_0 \}\\ &= 
\M \cap \{ j \in \{1,\ldots,\Tau \} \vert \col(j) = \col_0 \}\\ &= 
\{ j \in \{1,\ldots,\Tau \} \vert \col(j) = \col_0 \wedge \col(j) \ne
\col(\pi(j)) \}\\ &=
\{ j \in \{1,\ldots,\Tau \} \vert \col(j) = \col_0 \wedge \col(\pi(j))
\ne \col_0 \}\ .
\end{split}
\end{equation*}
This restriction is meet iff $\pi$ has the property that $\col(\pi(j))
\ne \col_0$ for $j \in T$ and $\col(\pi(j)) = \col_0$ for $j \in C_0
\setminus T$, where $C_0 = \{ j \vert \col(j) = \col_0
\}$. Furthermore, any $\pi$ meeting this restrictions would lead to
the observed value of $\pi$. It is hence sufficient to show that we
can sample a uniformly random $\pi$ meeting these restrictions.

Let $\overline{C_0} = \{ 1, \ldots, \Tau \} \setminus C_0$.  Pick
$\pi_0: T \rightarrow \overline{C_0}$ to be a uniformly random
injection on the specified domains.  Pick $\pi_1: C_0 \setminus T
\rightarrow C_0$ similarly.  Let $\pi_2: T \cup C_0 \rightarrow \{ 1,
\ldots, \tau \}$ be defined by $\pi_2(j) = \pi_0(j)$ for $j \in T$ and
$\pi_2(j) = \pi_1(j)$ for $j \in C_0 \cup T$. Since $\pi_0$ and
$\pi_1$ map into disjoint sets, this is again an injection.  Now let
$\pi_3: \{ 1, \ldots, \tau \} \setminus (C_0 \cup T) \rightarrow \{ 1,
\ldots, \tau \} \setminus \img(\pi_2)$ be a random permutation on the
specified domains. Define $\pi$ from $\pi_2$ and $\pi_3$ as we defined
$\pi_2$ from $\pi_0$ and $\pi_1$. Then it is easy to see that $\pi$ is
a uniformly random permutation meeting the restrictions. The
definition of $\pi$ shows how to sample it efficiently.

\Subsubsection{Concluding the Proof}

Using the above theorem and lemmata and the fact that linear
reducibility is transitive, we now have the following theorem.

\begin{corollary}
  $\LABIT^\LL(\tau,\ell)$ is linear reducible to
  $(\OT(2\tau,\ell),\EQ(\tau\ell))$.
\end{corollary}

We now show that if we set $\kappa = \frac34 \tau$, then $\LL$ is
$\kappa$-secure. For this purpose we assign a price to each ball $j
\in \SS(\pi)$.
\begin{enumerate}
\item

  If $\col(j) \ne \col(\pi(j))$, then let $\price_{\col,\pi}(j) = 1$.

\item

  If $\col(j) = \col(\pi(j)) = \col_0$, then let $\price_{\col,\pi}(j) = 1$.

\item

  If $\col(j) = \col(\pi(j)) \ne \col_0$, then let $\price_{\col,\pi}(j) = 0$.

\end{enumerate}
Let $\price_{\col,\pi} = \sum_{j \in \SS} \price_{\col,\pi}(j)$.

\begin{lemma}
  Consider an adversary \Ali playing the game against $\LL$ and
  assume that it submits $L = \col$. Assume that the game uses $\pi$.
  Then the success probability of \Ali is at most
  $2^{-\price_{\col,\pi}}$.
\end{lemma}
\begin{proof}
  Define $\price_{\col,\pi}^1(j)$ as $\price_{\col,\pi}(j)$ except
  that if $\col(j) = \col(\pi(j)) = \col_0$, then
  $\price^1_{\col,\pi}(j) = 0$.  Define $\price_{\col,\pi}^2(j)$ as
  $\price_{\col,\pi}(j)$ except that if $\col(j) \ne \col(\pi(j))$,
  then $\price_{\col,\pi}(j) = 0$. Then $\price_{\col,\pi}(j) =
  \price_{\col,\pi}^1(j) + \price_{\col,\pi}^2(j)$.  Define
  $\price_{\col,\pi}^1$ and $\price_{\col,\pi}^2$ by summing over $j
  \in \SS$. Then $\price_{\col,\pi} = \price_{\col,\pi}^1 +
  \price_{\col,\pi}^2$.  Note that $\vert \M \vert =
  \price_{\col,\pi}(j)$ and note that $\vert S' \vert = \tau -
  \price_{\col,\pi}^2(j)$,\footnote{Recall that $S'$ is defined during
  the definition of $\LL$ above.} as the only balls $j \in \SS$ which
  do not enter $S'$ are those for which $\col(j) = \col(\pi(j)) =
  \col_0$. We have that \Ali wins if $c=1$ and he guesses $y_{\Pi(j)}$
  for $j \in \SS \setminus S'$. The probability that $c=1$ is
  $2^{-\vert \M \vert} = 2^{-\price_{\col,\pi}(j)}$. We have that
  $\vert \SS \setminus S' \vert = \vert \SS \vert - \vert S' \vert =
  \tau - (\tau - \price_{\col,\pi}^2(j)) =
  \price_{\col,\pi}^2(j)$. So, the probability that \Ali guesses
  correctly is $2^{-\price_{\col,\pi}^2(j)}$. So, the overall success
  probability is
  $2^{-\price_{\col,\pi}^1(j)}2^{-\price_{\col,\pi}^2(j)} =
  2^{-\price_{\col,\pi}(j)}$. \putbox
\end{proof}

Now let $\pi$ be chosen uniformly at random and let $\price_{\col(j)}$ be
the random variable describing $\price_{\col,\pi}(j)$. Let $\price_{\col} =
\sum_{j \in \SS} \price_{\col(j)}$. It is then easy to see that the
probability of winning the game on $L = \col$ is at most
$$\success_{\col} = \sum_{p=0}^\tau \prob{\price_{\col} = p} 2^{-p}\ .$$

For each price $p$, let $P_p$ be an index variable which is $1$ if
$\price_{\col} = p$ and which is $0$ otherwise. Note that $\E{P_p} =
\prob{\price_{\col} = p}$, and note that $\sum_{p=0}^\tau P_p 2^{-p} =
2^{-\price_{\col}}$ as $P_p = 0$ for $p \ne \price_{\col}$ and $P_p = 1$
for $p = \price_{\col}$. Then
\begin{equation*}
  \begin{split}
    \success_{\col} &= \sum_{p=0}^\tau \prob{\price_{\col} = p} 2^{-p}\\ 
    &= \sum_{p=0}^\tau \E{P_p} 2^{-p}\\ 
    &= \E{\sum_{p=0}^\tau P_p 2^{-p}}\\ 
    &= \E{2^{-\price_c}}\\ 
    &= \E{2^{-\sum_{j \in \SS} \price_c(j)}}\ .
  \end{split}
\end{equation*}
Now let $\phi(x) = 2^{-x}$, and we have that
$$\success_{\col} =  
\E{\phi(\sum_{j \in \SS} \price_{\col}(j))}\ .$$
Since $\phi(x)$ is concave it follows from Jensen's inequality that     
$$\E{\phi(\sum_{j \in \SS} \price_{\col}(j))} \le
\phi\left(\E{\sum_{j \in \SS} \price_{\col}(j)}\right)\ .$$ 
Hence
$$\success_{\col} \le 2^{-\E{\sum_{j=1}^\tau \price_{\col}(\Pi^{-1}(j))}} =
2^{-\sum_{j=1}^\tau \E{\price_{\col}(\Pi^{-1}(j))}} = 2^{-\sum_{j \in \SS}
\E{\price_{\col}(j)}}\ .$$ It follows that if we can compute $m_0 = \min_{\col}
\sum_{j \in \SS} \E{\price_{\col}(j)}$, then $2^{-m_0}$ is an upper bound on
the best success rate.

We say that $L = \col$ is optimal if $\sum_{j \in \SS} \E{\price_L(j)}
= m_0$, and now find an optimal $L$.

We first show that there is no reason to use balls of color $\col_0$ in
the optimal strategy.
\begin{lemma}\lemmlab{polycrome}
  Let $L = \col$ be an optimal leakage function and let $\col_0 =
  \col_0(\col)$. Then there exist $j$ such that $\col(j) \ne \col_0$.
\end{lemma}
\begin{proof}
  Assume for the sake of contradiction that $\col(j) = \col_0$ for $j= 1,
  \ldots, \Tau$. Then clearly $\sum_{j \in \SS} \E{\price_{\col}(j)} =
  \tau$, and it is easy to see that there are strategies which do
  better than $2^{-\tau}$, so $L$ cannot be optimal. \putbox
\end{proof}

Let $\col_1, \ldots, \col_\Tau$ be an enumeration of the colors
different from $\col_0$, i.e., $\{ \col_0, \col_1, \ldots, \col_\Tau
\} = \{ 1, \ldots, \Tau \}$. Let $C_i$ be the balls with color
$\col_i$, i.e., $C_i = \{ j \in \{1,\ldots,\Tau\} \vert \col(j) =
\col_i \}$.  Note that $\{ 1, \ldots, \Tau \}$ is a disjoint union of
$C_1, \ldots, C_\Tau$.  Let $a_i$ be the number of balls of color
$\col_i$, i.e., $a_i = \vert C_i \vert$. Note that $\Tau =
\sum_{i=1}^\Tau a_i$.

With these definitions we have that
$$\sum_{j=1}^\Tau \E{\price_{\col}(j)} = \sum_{i=1}^\Tau \sum_{j \in C_i}
  \E{\price_{\col}(j)}\ .$$

For a ball $j \in C_0$ of color $\col_0$ we always have $\price_{\col}(j) =
\oh$, by definition of the price, so 
$$\sum_{j \in C_0} \E{\price_{\col}(j)} = \sum_{j \in C_0} \oh = \oh a_0\
.$$

For a ball $j \in C_i$ for $i > 0$ we have $\price_{\col}(j) = 0$ if
$\col(\pi(j)) = \col_i$ and $\price_{\col}(j) = \oh$ if $\col(\pi(j))
\ne \col_i$. We have that $\pi(j)$ is uniform on $\{ 1, \ldots, \Tau
\} \setminus \{ j \}$. Since $\col(j) = \col_i$ there are $a_i - 1$
balls $k \in \{ 1, \ldots, \Tau \} \setminus \{ j \}$ for which
$\col(k) = \col_i$. So,
\begin{equation*}
\begin{split}
\E{\price_{\col}(j)} &= 
\oh \frac{(\Tau-1) - (a_i -1)}{\Tau-1}\\ 
&= \oh \frac{\Tau - a_i}{\Tau-1}\ ,
\end{split}
\end{equation*}
which implies that
\begin{equation*}
\begin{split}
\sum_{j \in C_i} \E{\price_{\col}(j)} &= 
a_i \oh \frac{\Tau - a_i}{\Tau-1}\\ 
&= \oh \frac{1}{\Tau-1}(\Tau a_i - a_i^2)\ .
\end{split}
\end{equation*}
It follows that
\begin{equation*}
\begin{split}
\sum_{i=1}^{\Tau-1} \sum_{j \in C_i} \E{\price_{\col}(j)} 
&= 
\oh \frac{1}{\Tau-1}
\sum_{i=1}^{\Tau-1} (\Tau a_i - a_i^2)\\
&= 
\oh \frac{1}{\Tau-1}(
\Tau \sum_{i=1}^{\Tau-1}  a_i - \sum_{i=1}^{\Tau-1} a_i^2)\\ 
&= 
\oh \frac{1}{\Tau-1}(
\Tau (\Tau - a_0) - \sum_{i=1}^{\Tau-1} a_i^2)\ .
\end{split}
\end{equation*}
All in all we now have that
\begin{equation*}
\begin{split}
\sum_{i=0}^{\Tau-1} \sum_{j \in C_i} \E{\price_{\col}(j)} 
&= 
\oh a_0 + 
\oh \frac{1}{\Tau-1}(
\Tau (\Tau - a_0) - \sum_{i=1}^{\Tau-1} a_i^2)\\
&= 
\oh a_0 
- \oh \frac{\Tau}{\Tau-1} a_0 +
\oh \frac{1}{\Tau-1}(
 \Tau \Tau - \sum_{i=1}^{\Tau-1} a_i^2)\\
&= 
\oh (\frac{-a_0}{\Tau-1})  +
\oh \frac{1}{\Tau-1}(
 \Tau^2  - \sum_{i=1}^{\Tau-1} a_i^2)\\
&= 
\oh \frac{\Tau^2}{\Tau-1} 
- \oh \frac{1}{\Tau-1} (a_0 + \sum_{i=1}^{\Tau-1} a_i^2)\ .
\end{split}
\end{equation*}
To minimize this expression we have to maximize $a_0 +
\sum_{i=1}^{\Tau-1} a_i^2$. Recall that $\col_0$ is defined to be the
most common color, so we must adhere to $a_0 \ge a_i$ for $i >
0$. Under this restriction it is easy to see that 
$a_0 + \sum_{i=1}^{\Tau-1} a_i^2$ is maximal when
$a_0 = a_1 = \Tau/2$ and $a_2 = \cdots a_\Tau = 0$, in which case it
has the value $\Tau/2 + (\Tau/2)^2$. So,
\begin{equation*}
\begin{split}
\E{\price_{\col}} &= 
\oh \frac{\Tau^2}{\Tau-1} 
- \oh \frac{1}{\Tau-1} (\Tau/2 + (\Tau/2)^2)\\
&= 
\oh \frac{\Tau^2 - \Tau/2 + (\Tau/2)^2}{\Tau-1}\\
&= 
\oh \frac{4 \tau^2 - \tau - \tau^2}{2\tau-1}
= 
\oh \frac{3 \tau^2 - \tau}{2\tau-1}
= 
\oh \tau \frac{3 \tau - 1}{2\tau-1}
>
\oh \tau \frac{3 \tau}{2\tau}
=
\frac34 \tau = \kappa\ .
\end{split}
\end{equation*}

\Section{Efficient OT Extension}\applab{otext}

In this section we show how we can produce a virtually unbounded
number of OTs from a small number of seed OTs. The amortized work per
produced OT is linear in $\kappa$, the security parameter.

A similar result was proved in~\cite{DBLP:conf/tcc/HarnikIKN08}.
In~\cite{DBLP:conf/tcc/HarnikIKN08} the amortized work is linear in
$\kappa$ too, but our constants are much better than those
of~\cite{DBLP:conf/tcc/HarnikIKN08}. In fact, our constants are small
enough to make the protocol very practical.\footnote{As an example,
  our test run (see \secref{exp}) with $\ell = 54$ involved generating
  $44\mathord{,}826\mathord{,}624$ \ABIT{}s, each of which can be
  turned into one \OT using two applications of a hash function. The
  generation took $85$ seconds. Using these numbers, gives an estimate
  of $527\mathord{,}372$ actively secure OTs per second. Note,
  however, that the generation involved many other things than
  generating the \ABIT{}s, like combining them to \AOT{}s and
  \AAND{}s.}  Since~\cite{DBLP:conf/tcc/HarnikIKN08} does not
attempt to analyze the exact complexity of the result, it is hard to
give a concrete comparison, but since the result
in~\cite{DBLP:conf/tcc/HarnikIKN08} goes over generic secure
multiparty computation of non-trivial functionalities, the constants
are expected to be huge compared to ours.

Let $\kappa$ be the security parameter. We show that
$\OT({\ell},\kappa)$ is linear reducible to ($\OT(\frac{8}{3}
\kappa,\kappa)$, $\EQ(\frac43 \kappa^2)$) for any $\ell =
\poly(\kappa)$, i.e., given $\frac83 \kappa$ active-secure OTs of
$\kappa$-bit strings we can produce an essentially unbounded number of
active-secure OTs of $\kappa$-bit strings. The amortized work involved
in each of these $\ell$ OTs is linear in $\kappa$, which is optimal.

The approach is as follows.
\begin{enumerate}
\item

  Use $\OT(\frac83 \kappa,\kappa)$ and a pseudo-random generator to
  implement $\OT(\frac83 \kappa,\ell)$.

\item

  Use $\OT(\frac83 \kappa,\ell)$ and $\EQ(\frac43 \kappa \ell)$ to
  implement $\WABIT^\LL$ for $\ell$ authentications with $\frac43
  \kappa$-bit keys and MACs and with $\LL$ being $\kappa$-secure.

\item

  Use a random oracle $H: \zo^{\frac43 \kappa} \rightarrow
  \zo^{\kappa} $ and $\WABIT^\LL$ for $\ell$ authentications with
  $\frac43 \kappa$-bit keys to implement $\OT(\ell,\kappa)$, as
  described below.
  
\end{enumerate}

Here, as in~\cite{DBLP:conf/tcc/HarnikIKN08}, we consider a hashing of
$O(\kappa)$ bits to be linear work. The pseudo-random generator can be
implemented with linear work using $H$.

\longOnly{\paragraph{From \WABIT to \OT.}}

As a first step, we notice that the \ABIT box described
in~\secref{labit} resembles an intermediate step of the passive-secure
OT extension protocol of \cite{DBLP:conf/crypto/IshaiKNP03}: an \ABIT
can be seen as a random OT, where all the sender's messages are
correlated, in the sense that the XOR of the messages in any OT is a
constant (the global key of the \ABIT). This correlation can be easily
broken using the random oracle. In fact, even if few bits
of the global difference $\Delta$ leak to the adversary, the same
reduction is still going to work (for an appropriate choice of the
parameters). Therefore, we are able to start directly from the box for
authenticated bits with weak key, or \WABIT described
in~\secref{wabit}.

\begin{boxfig}{The Random OT box $\func{ROT}(\ell,\kappa)$}{BOXROT}
\begin{enumerate}

\item 

  For the sender \ant{S} the box samples $X_{i,0}, X_{i,1} \inR
  \zo^\kappa$ for $i = 1, \ldots, \ell$.  If \ant{S} is corrupted,
  then it gets to specify these inputs.

\item

  For the receiver \ant{R} the box samples $\vec{b} = (b_1, \ldots,
  b_\ell) \inR \zo^\ell$. If \ant{R} is corrupted, then it gets to
  specify these inputs.

\item

  The box outputs $((X_{1,b_1}, b_1), \ldots, (X_{\ell,b_l},b_\ell))$
  to \ant{R} and $((X_{1,0},X_{1,1}), \ldots,
  (X_{\ell,0},X_{\ell,1}))$ to \ant{S}.

\end{enumerate}
\end{boxfig}

\begin{boxfig}{The protocol for reducing $\func{ROT}(\ell,\kappa)$ to
  $\WABIT^\LL(\ell,\frac43 \kappa)$}{PROTROT}
\begin{enumerate}
\item

  Call $\WABIT^\LL(\ell,\frac43 \kappa)$.  The output to \ant{R} is
  $((M_1, b_1), \ldots, (M_\ell,b_\ell))$. The output to \ant{S} is
  $(\Delta, K_1, \ldots, K_\ell)$.

\item

  \ant{R} computes $Y_i = H(M_i)$ and outputs $((Y_1, b_1), \ldots,
  (Y_\ell,b_\ell))$.

\item

  \ant{S} computes $X_{i,0} = H(K_i)$ and $X_{i,1} = H(K_i \oplus
  \Delta)$ and outputs $((X_{1,0},X_{1,1}), \ldots,
  (X_{\ell,0},X_{\ell,1}))$.

\end{enumerate}
\end{boxfig}

Here $\kappa$ is the security level, i.e., we want to implement $\OT$
with insecurity $\poly(\kappa) 2^{- \kappa}$.  We are to use an
instance of $\WABIT^\LL$ with slightly larger keys. Specifically, let
$\tau = \frac43 \kappa$, as we know how to implement a box
$\WABIT^\LL$ with $\tau$-bit keys and where $\LL$ is $\kappa$-secure
for $\kappa = \frac34 \tau$. We implemented such a box
in~\secref{wabit}. The protocol is given in \figref{PROTROT}. It
implements the box for random OT given in \figref{BOXROT}.

We have that $M_i = K_i \oplus b_i \Delta$, so $Y_i = H(M_i) =
H(K_i \oplus b_i \Delta) = X_{i,b_i}$. Clearly the protocol leaks no
information on the $b_i$ as there is no communication from \ant{R} to
\ant{S}. It is therefore sufficient to look at the case where \ant{R}
is corrupted. We are not going to give a simulation argument but just
show that $X_{i,1 \oplus b_i}$ is uniformly random to \ant{R} except
with probability $\poly(\kappa) 2^{-\kappa}$.

Since $X_{i,1 \oplus b_i} = H(K_i \oplus (1 \oplus b_i) \Delta)$ and
$H$ is a random oracle, it is clear that $X_{i,1 \oplus b_i}$ is
uniformly random to \ant{R} until \ant{R} queries $H$ on $Q = K_i
\oplus (1 \oplus b_i) \Delta$. Since $M_i = K_i \oplus b_i \Delta$ we
have that $Q = K_i \oplus (1 \oplus b_i) \Delta$ would imply that $M_i
\oplus Q = \Delta$. So, if we let \ant{R} query $H$, say, on $Q
\oplus M_i$ each time it queries $H$ on some $Q$, which would not
change its asymptotic running time, then we have that all $X_{i,1
\oplus b_i}$ are uniformly random to \ant{R} until it queries $H$ on
$\Delta$. It is not hard to show that the probability with which an
adversary running in time $t = \poly(\kappa)$ can ensure that
$\WABIT^\LL$ does not fail and then query $H$ on $\Delta$ is
$\poly(\kappa) 2^{-\kappa}$. This follows from the $\kappa$-security
of $\LL$.

\Section{Proof of \thmref{laot}}\applab{proofoflaot}

The simulator answers global key queries to the dealer by doing
the identical global key queries on the ideal functionality
$\LAOT(\ell)$ and returning the reply from $\LAOT(\ell)$. This gives
a perfect simulation of these queries, and we ignore them below.

For honest sender and receiver correctness of the protocol follows
immediately from correctness of the \ABIT box and the \EQ box.

\begin{lemma} \label{lem:S}
  The protocol in \figref{PROTLAOT} securely implements $\LAOT(\ell)$
  against corrupted \Ali.
\end{lemma}
\begin{proof}
  We consider the case of a corrupt sender \Aliadv running the above
  protocol against a simulator \Sim. We show how to simulate one instance.
  \begin{enumerate}

  \item \label{lem:S_OA1}

    First \Sim receives \Aliadv's input $(M_{x_0}, x_{0})$, $(M_{x_1},
    x_{1})$, $K_c, K_r$ and $\Delta_B$ to the dealer. Then \Sim
    samples a bit $y \inR \zo$, sets $K_z = K_r \oplus y \Delta_B$ and
    inputs $(M_{x_0}, x_{0})$, $(M_{x_1}, x_{1})$, $K_c, K_z$ and
    $\Delta_B$ to a \LAOT box. The box outputs $\Delta_A$, $(M_c, c)$,
    $(M_z, z)$, $K_{x_0}$ and $K_{x_1}$ to the honest \Bob as
    described in the protocol.

  \item \label{lem:S_transfer}

    \Aliadv outputs the message $(X_{0}, X_{1})$. The simulator knows
    $\Delta_B$ and $K_c$ and can therefore compute
    $$
    X_0 \oplus H(K_c) = (\overline{x}_0||\overline{M}_{x_{0}}||T'_{x_0})
    $$
    and
    $$
    X_1 \oplus  H(K_c\oplus \Delta_B) =
    (\overline{x}_{1}||\overline{M}_{x_{1}}||T'_{x_1})\ .
    $$

    For all $j \in \zo$ \Sim tests if $(\overline{M}_{x_{j}},
    \overline{x}_j) = (M_{x_j}, x_{j})$. If, for some $j$, this is not
    the case \Sim inputs a guess to the \LAOT box guessing that $c =
    (1 - j)$ to the \LAOT box. If the box outputs \texttt{fail} \Sim
    does the same and aborts the protocol. Otherwise \Sim proceeds by
    sending $y$ to \Aliadv.  Notice that if \Sim does not abort but
    does guess the choice bit $c$ it can perfectly simulate the
    remaining protocol. In the following we therefore assume this is
    not the case.

  \item \label{lem:S_re-auth1}

    Similarly \Sim gets $(I_{0}, I_{1})$ from
    \Aliadv and computes 
    $$
    I_0 \oplus H(K_z) = T''_1
    $$
    and 
    $$
    I_1 \oplus H(K_z\oplus \Delta_B) = T''_0\ .
    $$

  \item \label{lem:S_re-auth2}

    When \Sim receives \Aliadv's input 
    $(T_{ 0}, T_{1})$ for the \EQ box it first tests if  
    $(T'_{j}, T''_{1 \oplus \overline{x}_j}) = 
    (T_{ \overline{x}_j}, T_{1 \oplus \overline{x}_j})$ for all $j \in
    \zo$. If, for some $j$, this is not the case \Sim inputs a
    guess to the \LAOT box guessing that $c = (1 - j)$. If the box
    outputs \texttt{fail}, \Sim outputs \texttt{fail} and aborts.
    If not, the simulation is over.
    
  \end{enumerate}
  For analysis of the simulation we denote by $F$ the event that
  for some $j \in \zo$ \Aliadv computes values $M^{*}_{x_j} \in
  \zo^{\kappa}$ and $x^{*}_{j} \in \zo$ so that $(M^{*}_{x_j},
  x^{*}_{j}) \not= (M_{x_j}, x_{j})$ and $M^{*}_{x_j} = K_{x_j} \oplus
  x^{*}_{j} \Delta_A$. In other words, $F$ is the event that \Aliadv
  computes a MAC on a message bit it was not supposed to know. We will
  now show that, assuming $F$ does not occur, the simulation is
  perfectly indistinguishable from the real protocol. We then show
  that $F$ only occurs with negligible probability and therefore that
  simulation and the real protocol are indistinguishable. 

  From the definition of the \LAOT box we have
  that $(\overline{M}_{x_ j}, \overline{x}_{j}) = (M_{x_j}, x_{j})$
  implies $\overline{M}_{x_j} = K_{x_j} \oplus x_j \Delta_A$. Given
  the assumption that $F$ does not occur clearly we have that
  $(\overline{M}_{x_j}, \overline{x}_{j}) \neq (M_{x_j}, x_{j})$ also
  implies $\overline{M}_{x_j} \neq K_{x_j} \oplus \overline{x}_j
  \Delta_A$.  This means that \Sim aborts in step \ref{lem:S_transfer}
  with exactly the same probability as the honest receiver would in
  the real protocol. Also, in the real protocol we have $y = z\oplus r$
  for $r \inR \zo$ thus both in the real protocol and the simulation
  $y$ is distributed uniformly at random in the view of \Aliadv.

  Next in step \ref{lem:S_re-auth2} of the simulation notice that in
  the real protocol, if $c = j \in \zo$, an honest \Bob would input
  $T'_{j}$ and $T''_{1 \oplus \overline{x}_j}$ to \EQ (sorted in the
  correct order). The protocol would then continue if and only
  if $(T'_{j}, T''_{1 \oplus \overline{x}_j}) = (T_{\overline{x}_j},
  T_{1 \oplus \overline{x}_j})$ and abort otherwise. I.e., the real
  protocol would continue if and only if $(T'_{j}, T''_{1
    \oplus \overline{x}_j}) = (T_{\overline{x}_j}, T_{1 \oplus
    \overline{x}_j})$ and $c = j$, which is exactly what happens in
  the simulation. Thus we have that given $F$ does not occur, all
  input to \Aliadv during the simulation is distributed exactly as in
  real protocol. In other words the two are perfectly
  indistinguishable.

  Now assume $F$ does occur, that is for some $j \in \zo$ \Aliadv
  computes values $M^{*}_{x_j}$ and $x^{*}_{j}$ as described above. In
  that case \Aliadv could compute the global key of the honest
  receiver as $M^{*}_{j}\oplus M_{x_j} = \Delta_A$. However, since all
  inputs to \Aliadv are independent from $\Delta_A$ (during the
  protocol), \Aliadv can only guess $\Delta_A$ with negligible
  probability (during the protocol) and thus $F$ can only occur with
  negligible probability (during the protocol). After the protocol
  \Aliadv, or rather the environment, will receive outputs and learn
  $\Delta_A$, but this does not change the fact that guessing $\Delta_A$
  during the protocol can be done only with negligible probability. \putbox
\end{proof}

\begin{lemma} \label{lem:R}
  The protocol in \figref{PROTLAOT} securely implements $\LAOT(\ell)$
  against corrupted \Bob.
\end{lemma}
\begin{proof}
  We consider the case of a corrupt receiver \Bobadv running the
  above protocol against a simulator \Sim. The simulation runs as
  follows.
  \begin{enumerate}
  
  \item \label{lem:R_setup}
    
    The simulation starts by \Sim getting \Bobadv's input to dealer
    $\Delta_A$, $(M_c, c)$, $(M_r, r)$, $K_{x_0}$ and $K_{x_1}$. Then
    \Sim simply inputs $\Delta_A$, $(M_c, c)$, $M_z=M_r$, $K_{x_0}$
    and $K_{x_1}$ to the \LAOT box. The box outputs $z$ to \Sim and
    $\Delta_B$, $(M_{x_0}, x_{0})$, $(M_{x_1}, x_{1})$, $K_c$ and
    $K_z$ to the sender as described above.

  \item \label{lem:R_trans}

    Like the honest sender \Sim samples random keys $T_{0},
    T_{1} \inR \zo^{\kappa}$. Since \Sim knows $M_c,
    K_{x_0}, K_{x_1}, \Delta_A, c$ and $z = x_{c}$ it
    can compute
    $X_{c} = H(M_c) \oplus (z||M_{z}||T_{z})$
    exactly as the honest sender would. It then samples $X_{1 \oplus
      c} \inR \zo^{2\kappa + 1}$ and inputs $(X_{0}, X_{1})$ to
   \Bobadv. 

  \item \label{lem:R_re-auth1}

    The corrupt receiver \Bobadv replies by sending some
    $\overline{y} \in \zo$. 

  \item \label{lem:R_re-auth2}

    \Sim sets $\overline{z} = r \oplus \overline{y}$,
    computes $I_{\overline{z}} = H(M_z)\oplus T_{1\oplus\overline{z}}$ and
    samples $I_{1\oplus\overline{z}} \inR \zo^{\kappa}$. It then inputs
    $(I_{0}, I_{1})$ to \Bobadv. 

  \item  \label{lem:R_re-auth3}

   \Bobadv outputs some 
    $(\overline{T}_{0}, \overline{T}_{1})$ for the \EQ
    box and \Sim continues or aborts as the honest \Ali would
    in the real protocol, depending on whether
    or not 
    $(T_{0}, T_{1}) = (\overline{T}_{0}, \overline{T}_{1})$.

  \end{enumerate}
  For the analysis we denote by $F$ the event that \Bobadv queries the
  RO on $K_c \oplus (1 \oplus c)\Delta_B$ or $K_z \oplus
  (1 \oplus z)\Delta_B$. We first show that assuming $F$ does not occur, the
  simulation is perfect. We then show that $F$ only occurs with
  negligible probability (during the protocol) and thus the
  simulation is indistinguishable from the real protocol (during the
  protocol). We then discuss how to simulate the RO after outputs have
  been delivered.

  First in the view of \Bobadv step \ref{lem:R_setup} of the
  simulation is clearly identical to the real protocol. Thus the first
  deviation from the real protocol appears in step \ref{lem:R_trans}
  of the simulation where the $X_{1 \oplus c}$ is chosen uniformly at
  random.  However, assuming $F$ does not occur, \Bobadv has no
  information on $H(K_c \oplus (1 \oplus c)\Delta_B)$ thus in the view of
  \Bobadv, $X_{1 \oplus c}$ in the real protocol is a one-time pad
  encryption of $(x_{1 \oplus c}||M_{x_{1 \oplus c}}||T_{x_{1 \oplus c}})$. In other words,
  assuming $F$ does not occur, to \Bobadv, $X_{1 \oplus c}$ is uniformly random
  in both the simulation and the real protocol, and thus all input to
  \Bobadv up to step \ref{lem:R_trans} is distributed identically in
  the two cases.

  For steps \ref{lem:R_re-auth1} to \ref{lem:R_re-auth3} notice that in
  the real protocol an honest sender would set $K_z = K_r \oplus
  \overline{y} \Delta_B$ and we would have
%  \[
%  \mac^{\Delta}_{Z_{i}}(\overline{z}_{i}) 
%  = (\tilde{Z}_{i} \oplus \overline{y}_{i}\Delta) \oplus \overline{z}_{i}\Delta 
%  = \tilde{Z}_{i} \oplus \tilde{z}_{i}\Delta 
%  = O_{i}.
%  \]
  $$
  (K_r \oplus \overline{y}\Delta_B) \oplus \overline{z}\Delta_B 
  = K_r \oplus r \Delta_B 
  = M_r\ .
  $$ 
  Thus we have that the simulation generates $I_{\overline{z}}$
  exactly as in the real protocol. An argument similar to the one
  above for step \ref{lem:R_trans} then gives us that the simulation
  is perfect given the assumption that $F$ does not occur.

  We now show that \Bobadv can be modified so that if $F$ does occur,
  then 
  \Bobadv can find $\Delta_B$. However, since all input to \Bobadv are
  independent of $\Delta_B$ (during the protocol), \Bobadv only has
  negligible probability of guessing $\Delta_B$ and thus we can
  conclude that $F$ only occurs with negligible probability.

  The modified \Bobadv keeps a list $Q =
  (Q_{1}, \ldots, Q_{q})$ of all \Bobadv's queries to $H$. Since
  \Bobadv is efficient we have that $q$ is a polynomial in $\kappa$.
  To find $\Delta_B$ the modified \Bobadv then goes over all
  $Q_{k} \inR Q$ and computes $Q_{k} \oplus M_z = \Delta'$ and $Q_k
  \oplus M_c = \Delta''$. Assuming that $F$ does occur there will be
  some $Q_{k'} \in Q$ s.t.~$\Delta' = \Delta_B$ or
  $\Delta''=\Delta_B$. The simulator can therefore use global key
  queries to find $\Delta_B$ if $F$ occurs.

  We then have the issue that after outputs are delivered to the
  environment, the environment learns $\Delta_B$, and we have to keep
  simulating $H$ to the environment after outputs are delivered. This
  is handled exactly as in the proof of \thmref{laand} in
  \appref{proofofthmAND} using the programability of the RO.  \putbox
\end{proof}

\Section{Proof of \thmref{thmAOT}}\applab{proofthmaot}

We want to show that the protocol in~\figref{PROTAOT} produces secure
\AOT{}s, having access to a box that produces leaky \AOT{}s. Remember
that a leaky \AOT or \LAOT, is insecure in the sense that a corrupted
sender can make guesses at any of the choice bits: if the guess is
correct, the box does nothing and therefore the adversary knows that
the guess was correct. If the guess is wrong, the box alerts the
honest receiver about the cheating attempt and aborts.

In the protocol the receiver randomly partitions $\ell B$ leaky OTs in
$\ell$ buckets of size $B$. First we want to argue that the
probability that every bucket contains at least one OT where the
choice bit is unknown to the adversary is overwhelming. Repeating the
same calculations as in the proof of~\thmref{thmAND} it turns out that
this happens with probability bigger than $1-(2\ell)^{(1-B)}$.

Once we know that (with overwhelming probability) at least one OT in
every bucket is secure for the receiver (i.e., at least one choice bit
is uniformly random in the view of the adversary), the security of the
protocol follows from the fact that we use a standard OT
combiner~\cite{DBLP:conf/eurocrypt/HarnikKNRR05}. Turning this into a
simulation proof can be easily done in a way similar to the proof
of~\thmref{thmAND} in~\appref{proofoflaand}.

\Section{Proof of \thmref{laand}}\applab{proofoflaand}

\begin{proof}
The simulator answers global key queries to the dealer  by doing
the identical global key queries on the ideal functionality
$\LAAND(\ell)$ and returning the reply from $\LAAND(\ell)$. This gives
a perfect simulation of these queries, and we ignore them below.

Notice that for honest sender and receiver correctness of the
protocol follows immediately from correctness of the \ABIT box.
\begin{lemma} \lemmlab{AANDAliAdv}
  The protocol in~\figref{LAAND} securely implements the \LAAND
  box against corrupted \Ali.
\end{lemma}
\begin{proof}
  We first focus on the simulation of the protocol before outputs are
  given to the environment. Notice that before outputs are given to
  the environment, the global key $\Delta_A$ is uniformly random to the
  environment, as long as \Bob is honest.

  We consider the case of a corrupt sender \Aliadv running the above
  protocol against a simulator \Sim for honest \Bob.
  \begin{enumerate}
  \item \label{lem:A_AND1}

    First \Sim receives \Aliadv's input $(M_{x}, x),
    (M_{y}, y), (M_r,r)$ for the
    dealer.
    
    Then \Sim receives the bit
    $d \inR \zo$.
    
  \item \label{lem:A_AND2}
    
    \Sim samples a random $U\in_R \zo^{2\kappa}$ and sends it to
    \Aliadv. Then \Sim reads $\overline{V}$, \Aliadv's input to the
    \EQ box. If $\overline{V} \neq (1-x)H(M_x,M_z) \oplus
    x (U\oplus H(M_x,M_y \oplus M_z))$ or $d\oplus y \neq xy$, \Sim
    outputs abort, otherwise, it inputs $(x,y,z,M_x,M_y,M_z=M_r)$ to
    the \LAAND box.

\end{enumerate}

The first difference between the real protocol and the simulation is
that $U = H(K_x,K_z) \oplus H(K_x \oplus \Delta_A,K_y \oplus K_z)$ in
the real protocol and $U$ is uniformly random in the
simulation. Since $H$ is a random oracle, this is perfectly
indistinguishable to the adversary until it queries on both $(K_x,K_z)$
and $(K_x \oplus \Delta_A,K_y \oplus K_z)$. Since $\Delta_A$ is
uniformly random to the environment and the adversary during the
protocol, this will happen with negligible probability during the
protocol. We later return to how we simulate after outputs are given
to the environment.

The other difference between the protocol and the simulation is that
the simulation always aborts if $z \neq xy$. Assume now that \Aliadv
manages, in the real protocol, to make the protocol continue with
$z=xy\oplus 1$. If $x=0$, this means that \Aliadv queried the oracle
on $(K_x,K_z)=(M_x,M_z\oplus \Delta_A)$, and since \Sim knows the
outputs of corrupted \Ali, which include $M_z$, and see the input
$M_z\oplus \Delta_A$ to the RO $H$, if \Aliadv queries the oracle on
$(K_x,K_z)=(M_x,M_z\oplus \Delta_A)$, \Sim can compute $\Delta_A$.  If
$x=1$ then \Aliadv must have queried the oracle on $(K_x\oplus
\Delta_A, K_y\oplus K_z)=(M_x, M_y \oplus M_z \oplus \Delta_A)$, which
again would allow \Sim to compute $\Delta_A$. Therefore, in both cases
we can use such an \Aliadv to compute the global key $\Delta_A$ and,
given that all of \Aliadv's inputs are independent of $\Delta_A$
during the protocol, this happens only with negligible probability.

Consider now the case after the environment is given outputs. These
outputs include $\Delta_A$. It might seem that there is nothing more
to simulate after outputs are given, but recall that $H$ is a random
oracle simulated by \Sim and that the environment might keep querying
$H$. Our concern is that $U$ is uniformly random in the simulation and
$U = H(K_x,K_z) \oplus H(K_x \oplus \Delta_A,K_y \oplus K_z)$ in the
real protocol. We handle this as follows. Each time the environment
queries $H$ on an input of the form $(Q_1,Q_2) \in \zo^{2\kappa}$, go
over all previous queries $(Q_3, Q_4)$ of this form and let $\Delta =
Q_1 \oplus Q_3$. Then do a global key query to $\ABIT(3\ell,\kappa)$ to
determine if $\Delta = \Delta_A$. If \Sim learns $\Delta_A$ this way,
she proceeds as described now. Note that since \Ali is corrupted, \Sim
knows all outputs to \Ali, i.e., \Sim knows all MACs $M$ and all bits
$b$.  If $b = 0$, then \Sim also knows the key, as $K = M$ when
$b=0$. If $b=1$, \Sim computes the key as $K = M \oplus \Delta_A$. So,
when \Sim learns $\Delta_A$, she at the same time learns all
keys. Then for each $U$ she simply programs the RO such that $U =
H(K_x,K_z) \oplus H(K_x \oplus \Delta_A,K_y \oplus K_z)$. This is
possible as \Sim learns $\Delta_A$ no later than when the environment
queries on two pairs of inputs of the form $(Q_1,Q_2) = (K_x,K_z)$ and
$(Q_3,Q_4) = (K_x \oplus \Delta_A, K_y \oplus K_z)$. So, when \Sim
learns $\Delta_A$, either $H(K_x,K_z)$ or $H(K_x \oplus \Delta_A,K_y
\oplus K_z)$ is still undefined. If it is $H(K_x,K_z)$, say, which is
undefined, \Sim simply set $H(K_x,K_z) \leftarrow U \oplus H(K_x \oplus
\Delta_A, K_y \oplus K_z)$. \putbox
\end{proof}

\begin{lemma} \lemmlab{AANDBobAdv}
  The protocol described  in~\figref{LAAND} securely implements the \LAAND box
  against corrupted \Bob.
\end{lemma}
\begin{proof}
  We consider the case of a corrupt \Bobadv running the above protocol
  against a simulator \Sim. The simulation runs as follows.
  \begin{enumerate}
  
  \item \label{lem:B_AND1}
    
    The simulation starts by \Sim getting \Bobadv's input to
    the dealer $K_x,K_y,K_r$ and $\Delta_A$. 

  \item \label{lem:B_AND2}

    The simulator samples a random $d\in_R \zo$,  sends it to
    \Bobadv and computes $K_z=K_r\oplus d\Delta_A$.

 \item \label{lem:B_AND3}

   \Sim receives $\overline{U}$ from \Bobadv, and reads $\overline{V}$,
   \Bobadv's input to the equality box.

 \item \label{lem:B_AND4}

   If $\overline{U}=H(K_x,K_z)\oplus H(K_x\oplus \Delta_A, K_y\oplus
   K_z)$ and $\overline{V}= H(K_x,K_z)$, input $(K_x,K_y,K_z)$ to the
   box for \LAAND and complete the protocol (this is the case where
   \Bobadv is behaving as an honest player). Otherwise, if $\overline{U}
   \neq H(K_x,K_z)\oplus H(K_x\oplus \Delta_A, K_y\oplus K_z)$ and
   $\overline{V} = H(K_x,K_z)$ or $\overline{V} = \overline{U}\oplus
   H(K_x\oplus\Delta_A,K_z\oplus K_z)$, input $g=0$ or $g=1$ resp.~
   into the \LAAND box as a guess for the bit $x$. If the box output
   \texttt{fail}, output \texttt{fail} and abort, and otherwise
   complete the protocol.
   
  \end{enumerate}

The simulation is perfect: the view of \Bobadv consists only of the
bit $d$, that is uniformly distributed both in the real game and in
the simulation, and in the aborting condition, that is the same 
in the real and in the simulated game. \putbox

\end{proof}

\end{proof}

\Section{Proof of \thmref{thmAND}} \applab{proofofthmAND}

\begin{proof}
The simulator answers global key queries to $\LAAND(B \ell)$ by doing
the identical global key queries on the ideal functionality
$\AAND(\ell)$ and returning the reply. This gives a perfect simulation
of these queries, and we ignore them below.

It is easy to check that the protocol is correct and secure if both
parties are honest or if \Ali is corrupted.

What remains is to show that, even if \Bob is corrupted and tries to
guess some $x$'s from the \LAAND box, the overall protocol is secure.

We argue this in two steps. We first argue that the probability that
\Bob learns the $x$-bit for all triples in the same bucket is
negligible. We then argue that when all buckets contain at least one
triple for which $x$ is unknown to \Bob, then the protocol can be
simulated given $\LAAND(B \ell)$.

Call each of the triples a \defterm{ball} and call a ball
\defterm{leaky} if \Bob learned the $x$ bit of the ball in the call to
$\LAAND(\ell')$. Let $\gamma$ denote the number of leaky balls.

For $B$ of the leaky balls to end up in the same bucket, there must be
a subset $S$ of balls with $\vert S \vert = B$ consisting of only
leaky balls and a bucket $i$ such that all the balls in $S$ end up in
$i$.

We first fix $S$ and $i$ and compute the probability that all balls in
$S$ end up in $i$. The probability that the first ball ends up in $i$
is $\frac{B}{B \ell}$. The probability that the second balls ends up
in $i$ given that the first ball is in $i$ is $\frac{B-1}{B \ell-1}$,
and so on. We get a probability of
\begin{align*}
  \frac{B}{B \ell}\cdot \frac{B-1}{B \ell-1} \cdots \frac{1}{B \ell-B+1} =
  \binom{B \ell}{B}^{-1}
\end{align*}
that $S$ ends up in $i$.

There are $\binom{\gamma}{B}$ subsets $S$ of size $B$ consisting of
only leaky balls and there are $\ell$ buckets, so by a union bound the
probability that any bucket is filled by leaky balls is upper bounded
by
$$\binom{\gamma}{B} \ell \binom{B \ell}{B}^{-1}\ .$$ This is
assuming that there are exactly $\gamma$ leaky balls. Note then that
the probability of the protocol not aborting when there are $\gamma$
leaky balls is $2^{- \gamma}$. Namely, for each bit $x$ that \Bob
tries to guess, he is caught with probability $\oh$. So, the
probability that \Bob undetected can introduce $\gamma$ leaky balls
and have them end up in the same bucket is upper bounded by
$$ 
\alpha(\gamma,\ell,B) = 2^{-\gamma} \binom{\gamma}{B} \ell \binom{B
  \ell}{B}^{-1}\ .
$$

It is easy to see that 
\begin{align*}
  \frac{\alpha(\gamma+1,\ell,B)}{\alpha(\gamma,\ell,B)} =
  \frac{\gamma+1}{2(\gamma+1-B)}.
\end{align*}

So, $\alpha(\gamma+1,\ell,B)/\alpha(\gamma,\ell,B) > 1$ iff $\gamma <
2B-1$, hence $\alpha(\gamma,\ell,B)$ is maximized in $\gamma$ at
$\gamma = 2B-1$. If we let $\alpha'(B,\ell) = \alpha(2B-1,\ell,B)$ it
follows that the success probability of the adversary is at most
$$\alpha'(B,\ell) = 2^{-2B+1} \ell \frac{(2B-1)!(B \ell -
  B)!}{(B-1)!(B\ell)!}\ .
$$
  
Writing out the product $\frac{(2B-1)!(B \ell - B)!}{(B-1)!(B\ell)!}$
it is fairly easy to see that for $2 \leq B < \ell$ we have that
\begin{align*}
  \frac{(2B-1)!(B \ell - B)!}{(B-1)!(B\ell)!} <
\frac{(2B)^B}{(B\ell)^B},
\end{align*}

so
\begin{align*}
  \alpha'(B,\ell) \leq 2^{-2B+1} \ell \frac{(2B)^B}{(B\ell)^B} = (2
  \ell)^{1-B}.
\end{align*}
   
We now prove that assuming each bucket has one non-leaky triple the
protocol is secure even for a corrupted \Bobadv.

We look only at the case of two triples,
$\repa{x^1},\repa{y^1},\repa{z^1}$ and
$\repa{x^2},\repa{y^2},\repa{z^2}$, being combined into
$\repa{x},\repa{y},\repa{z}$. It is easy to see why this is
sufficient: Consider the iterative way we combine the $B$ triples of a
bucket. At each step we combine two triples where one may be the
result of previous combinations. Thus if a combination of two triples,
involving a non-leaky triple, results in a non-leaky triple, the
subsequent combinations involving that result will all result in a
non-leaky triple.

In the real world a corrupted \Bobadv will input keys $K_{x^1},
K_{y^1}, K_{z^1}$ and $K_{x^2}, K_{y^2}, K_{z^2}$ and $\Delta_{A}$,
and possibly some guesses at the $x$-bits to the \LAAND box.  Then
\Bobadv will see $d = y^1 \oplus y^2$ and $M_{d} = (K_{y^1} \oplus
K_{y^2}) \oplus d\Delta_{A}$ and \Ali will output $x = x^1 \oplus
x^2$, $y = y^1$ , $z = z^1\oplus z^2 \oplus d x^2$ and $M_{x} =
(K_{x^{1}} \oplus K_{x^{2}}) \oplus x\Delta_{A}$, $M_{y} = K_{y^1}
\oplus y\Delta_{A}$, $M_{z} = (K_{z^1} \oplus K_{z^2} \oplus
dK_{x^{2}})\oplus z\Delta_{A}$ to the environment.

Consider then a simulator \ant{Sim} running against \Bobadv and using
an \AAND box. In the first step \ant{Sim} gets all \Bobadv's keys like
in the real world. If \Bobadv submits a guess $(i,g_i)$ \ant{Sim}
simply outputs \texttt{fail} and terminates with probability
$\frac{1}{2}$. To simulate revealing $d$, \ant{Sim} samples $d \inR
\zo$, sets $M_{d} = K_{y^1} \oplus K_{y^2} \oplus d \Delta_{A}$ and
sends $d$ and $M_{d}$ to \Bobadv.  \ant{Sim} then forms the keys
$K_{x} = K_{x^1} \oplus K_{x^2}$, $K_{y} = K_{y^1}$ and $K_{z} =
K_{z^1} \oplus K_{z^2} \oplus dK_{x^{2}}$ and inputs them to the \AAND
box on behalf of \Bobadv.  Finally the \AAND box will output random
$x$, $y$ and $z = xy$ and $M_{x} = K_{x} \oplus x\Delta_{A}$, $M_{y} =
K_{y} \oplus y\Delta_{A}$, $M_{z} = K_{z} \oplus z\Delta_{A}$.

We have already argued that the probability of \Bobadv guessing one of
the $x$-bits is exactly $\frac{1}{2}$, so \ant{Sim} terminates the
protocol with the exact same probability as the \LAAND box in the real
world. Notice then that, given the assumption that \Bobadv at most
guesses one of the $x$-bits, all bits $d$, $x$ and $y$ are uniformly
random to the environment both in the real world and in the
simulation. Thus because \ant{Sim} can form the keys $K_{x}$, $K_{y}$
and $K_{z}$ to the \AAND box exactly as they would be in the real
world the simulation will be perfect.

\putbox
\end{proof}

\Section{Full Overview Diagram}\applab{fod}

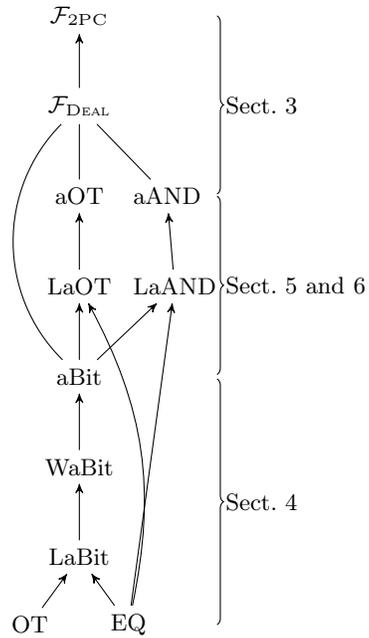
\begin{figure}
\begin{center}
\begin{tikzpicture}[scale=0.6]
% Draw the nodes
\node (twopc) at (0,0) {\FMPC};
\node (deal) at (0, -2) {\DEAL};

\node (aot) at (0, -4) {\AOT};
\node (laot) at (0, -6) {\LAOT};

\node[anchor=west] (aand) at (1, -4) {\AAND};
\node[anchor=west] (laand) at (1, -6) {\LAAND};

\node (abit) at (0, -8) {\ABIT};
\node (wabit) at (0, -10) {\WABIT};
\node (labit) at (0, -12) {\LABIT};

\node[anchor=west] (eq) at (0.5, -13.5) {\EQ};
\node[anchor=east] (ot) at (-0.5, -13.5) {\OT};

% Draw the arrows
\draw[->,>=stealth'] (deal) -- (twopc);
\draw[-] (aot) -- (deal);
\draw[-] (aand) -- (deal);
\draw[->,>=stealth'] (laot) -- (aot);
\draw[->,>=stealth'] (laand) -- (aand);
\draw[->,>=stealth'] (abit) -- (laot);
\draw[->,>=stealth'] (abit) -- (laand);
\draw[->,>=stealth'] (wabit) -- (abit);
\draw[->,>=stealth'] (labit) -- (wabit);
\draw[->,>=stealth'] (ot) -- (labit);
\draw[->,>=stealth'] (eq) -- (labit);
\draw[-] (abit) to [bend left=45] (deal);
\draw[->,>=stealth'] (eq) -- (laand);
\draw[->,>=stealth'] (eq) to [bend right=20] (laot);

% Draw the braces
\draw[decorate, decoration={brace}] (3.05,0) --
(3.05,-3.95) node[right, midway] {\small \secref{twopcfromaot}};
\draw[decorate, decoration={brace}] (3.05, -4) --
(3.05, -7.95) node[right, midway] {\small \secref{aot} and \ref{sec:aand}};
\draw[decorate, decoration={brace}] (3.05, -8.05) -- (3.05, -13.5) 
node[right, midway] {\small \secref{abit}};

% Very ugly hack to move picture to the right
\node (balance) at (-6,0) {};

\end{tikzpicture}
\caption{Full paper outline.}
\figlab{FullOverview}
\end{center}
\end{figure}

\end{document}